\documentclass{article}
\usepackage{authblk}

\usepackage[english]{babel}

\usepackage[a4paper,top=1in,bottom=1in,left=1in,right=1in]{geometry}

\usepackage{amsmath}
\usepackage{amsthm}
\usepackage{tabulary}
\usepackage{graphicx}
\usepackage[colorlinks=true, allcolors=blue]{hyperref}
\usepackage[dvipsnames]{xcolor}
\usepackage[T1]{fontenc}
\usepackage[mode=buildnew]{standalone}
\usepackage{cleveref}
\usepackage{comment}
\usepackage{mathtools}
\usepackage{url}
\usepackage{multirow}
\usepackage{wrapfig}
\usepackage{graphicx}
\usepackage{caption}
\usepackage{amsfonts}
\usepackage{enumerate}
\usepackage{enumitem}
\usepackage{amsthm}
\usepackage{thmtools}
\usepackage{thm-restate}
\usepackage{array}
\usepackage[inkscapeformat=pdf]{svg}

\newtheorem{theorem}{Theorem}

\newtheorem{definition}{Definition}
\newtheorem{lemma}{Lemma}

\newtheorem{observation}[theorem]{Observation}
\newtheorem{prop}{Proposition}

\newtheorem{remark}{Remark}

\usepackage[ruled,vlined,linesnumbered]{algorithm2e}
\crefname{algocf}{algorithm}{algorithms}
\Crefname{algocf}{Algorithm}{Algorithms}

\usepackage{xcolor}

\newcommand{\noextremity}[2]{\mathsf{State_{#1,#2}[NoInnerExtr]}}
\newcommand{\acyclic}[2]{\mathsf{State_{#1,#2}[Acyclic]}}
\newcommand{\reachesst}[2]{\mathsf{State_{#1,#2}[Reaches_{st}]}}
\newcommand{\reachests}[2]{\mathsf{State_{#1,#2}[Reaches_{ts}]}}
\newcommand{\state}[2]{\mathsf{State_{#1,#2}[\cdot]}}
\newcommand{\true}{\mathsf{True}}
\newcommand{\false}{\mathsf{False}}
\newcommand{\Null}{\mathsf{Null}}

\DeclareMathOperator{\skel}{skeleton}
\DeclareMathOperator{\dirskel}{skeleton^*}
\DeclareMathOperator{\expansion}{expansion}

\newcommand{\iink}{i \in \{1,\dots,k\}}
\newcommand{\iinkz}{i \in \{0,\dots,k\}}
\newcommand{\jink}{j \in \{1,\dots,k\}}
\newcommand{\jinkz}{j \in \{0,\dots,k\}}
\newcommand{\iinl}{i \in \{1,\dots,\ell\}}
\newcommand{\Bst}{B_{st}}

\newcommand{\signs}{\{+,-\}}

\newcommand{\alex}[1]{\textcolor{Green}{(AT\ifstrempty{#1}{}{: #1})}}
\newcommand{\todo}[1]{\textcolor{red}{(TODO\ifstrempty{#1}{}{: #1})}}

\usepackage[table]{xcolor}
\usepackage{booktabs}
\usepackage{adjustbox}
\usepackage{graphicx}
\usepackage{array}
\usepackage{tabularx} 

\definecolor{famBiSB}{HTML}{F9D4D4}
\definecolor{famSnarl}{HTML}{CFE0EF}

\newcommand{\equalcontrib}{\thanks{These authors contributed equally.}}
\newcommand{\cosupervised}{\thanks{These authors jointly supervised this work.}}

\title{Identifying all snarls and superbubbles in linear-time, \\via a unified SPQR-tree framework}
\author[1]{Francisco Sena\equalcontrib}
\author[1]{Aleksandr Politov$^{*}$}
\author[2]{Corentin Moumard}
\author[3]{Manuel Cáceres\cosupervised}
\author[1]{Sebastian Schmidt$^\dagger$}
\author[1]{Juha Harviainen$^\dagger$}
\author[1]{Alexandru I. Tomescu$^\dagger$}

\affil[1]{\small Department of Computer Science, University of Helsinki, Helsinki, Finland\\
\texttt{\{francisco.sena,aleksandr.politov\\sebastian.schmidt,juha.harviainen,alexandru.tomescu\}@helsinki.fi}}
\affil[2]{\small ENS Lyon, Lyon, France\\
\texttt{corentin.moumard@ens-lyon.fr}}
\affil[3]{\small Department of Computer Science, Aalto University, Espoo, Finland\\
\texttt{manuel.caceres@aalto.fi}}
\date{}

\begin{document}

\maketitle
\begingroup
\renewcommand\thefootnote{}\footnote{Co-funded by the European Union (ERC, SCALEBIO, 101169716). Views and opinions expressed are however those of the author(s) only and do not necessarily reflect those of the European Union or the European Research Council. Neither the European Union nor the granting authority can be held responsible for them. Juha Harviainen was supported by the Research Council of Finland, Grant 351156.}\addtocounter{footnote}{-1}
\renewcommand\thefootnote{}\footnote{We are grateful to Benedict Paten for very helpful explanations and clarifications on snarls and the snarl decomposition and to Romeo Rizzi for helpful comments on the manuscript.\\[0.2cm]\includegraphics[width=4cm]{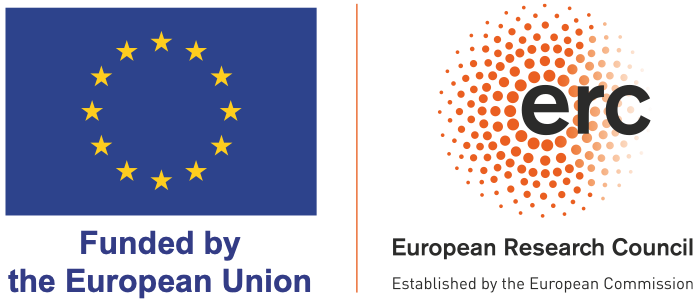}}\addtocounter{footnote}{-1}
\endgroup

\begin{abstract}
\emph{Snarls} and \emph{superbubbles} are fundamental pangenome decompositions capturing variant sites. These bubble-like structures underpin key tasks in computational pangenomics, including structural-variant genotyping, distance indexing, haplotype sampling, and variant annotation. Snarls can be quadratically-many in the size of the graph, and since their introduction in 2018 with the \texttt{vg} toolkit, there has been no published work on identifying \emph{all} snarls in linear time. Moreover, even though it is known how to find superbubbles in linear time, this result is a highly specialized and tailored solution only achieved after a long series of papers.

We present the first algorithm identifying all snarls in linear time. This is based on a new representation of all snarls, whose size is linear in input graph size, and which can be computed in linear time. Moreover, our algorithm is based on a \emph{unified framework} that also provides a new linear-time algorithm for finding superbubbles. The key observation behind our results is that all such bubble-like structures are separated from the rest of the graph by two vertices (except for cases which are trivially computable), i.e.~their endpoints are a \emph{2-separator} of the underlying undirected graph. Based on this, we employ the well-known SPQR tree decomposition, which encodes all 2-separators, to guide a traversal that finds the bubble-like structures efficiently. 

We implemented our algorithms in C++ (available at \url{https://github.com/algbio/BubbleFinder}) and evaluated them on various pangenomic datasets.
Our algorithms outcompete or they are on the same level of existing methods.
For snarls, we are up to two times faster than \texttt{vg}, while identifying all snarls.
When computing superbubbles, we are up to 50 times faster than BubbleGun.

Our SPQR tree framework provides a unifying perspective on bubble-like structures in pangenomics, together with a template for finding other bubble-like structures efficiently.

\end{abstract}


\thispagestyle{empty}
\setcounter{page}{0}
\newpage

\section{Introduction}

\textbf{Background and motivation.} It is well known that genomic variation between sequences induces various structures in graphs built from them. For example, \emph{bubbles}~\cite{zerbino2008velvet,fasulo2002efficiently} in assembly graphs are two parallel paths with the same start vertex and the same end vertex, which are used in identifying sequencing errors~\cite{lin2016assembly,iqbal2012novo}, and in de novo variant calling~\cite{reads2010identifying,uricaru2015reference}. The most popular generalization of a bubble is that of a \emph{superbubble}~\cite{onodera2013detecting}, where the subgraph between the start and end vertex can be any acyclic graph, separated from the rest of the graph. While genome graphs are usually bidirected~\cite{medvedev2007computability}, superbubbles can also be applied to them by considering the ``doubled'' directed version of a bidirected graph~\cite{dabbaghie2022bubblegun,garg2018graph,iqbal2012novo,shafin2020nanopore}. Superbubbles have witnessed various recent applications, including long-read metagenome assembly~\cite{kolmogorov2020metaflye}, telomere-to-telomere assembly of diploid chromosomes~\cite{rautiainen2023telomere} and multiple whole-genome alignment~\cite{minkin2020scalable}.

The development of variation and pangenome graphs led to the introduction of proper generalizations of superbubbles, to explicitly handle the bidirected nature of the graphs. The most well-known notions are \emph{snarls} and \emph{ultrabubles}~\cite{paten2018ultrabubbles}. At a high-level, snarls can be thought of as minimal subgraphs with two endpoints that separate the interior from the rest of the graph, and at each endpoint, all incidences of a same sign are inside the snarl, and all incidences of opposite sign are outside the snarl.\footnote{An \emph{incidence} is a vertex $v$ followed by a sign $d_v \in \{+,-\}$, e.g. $v+$. A snarl can be identified by a pair of incidences $\{vd_v,ud_u\}$, e.g. $\{v+,u-\}$ where the $+$ incidences of $v$, and the $-$ incidences of $u$, are inside the snarl. See \Cref{sec:preliminaries,sec:snarls-main-text}.} Snarls were introduced for decomposing pangenome graphs into sites for variant calling with the \texttt{vg}~toolkit~\cite{garrison2018variation}, and have seen further applications in variant annotation \cite{hickey2024pangenome,wang2025population}, distance-based seed indexing~\cite{chang2020distance}, structural-variant genotyping~\cite{siren2021pangenomics}, haplotype sampling for personalized pangenome references~\cite{siren2024personalized}. Other bubble-like structures tailored to bidirected graphs include also \emph{bibubbles}~\cite{li2024exploring}, \emph{flubbles}~\cite{mwaniki2024popping}, \emph{bundles}~\cite{bundles}.


Despite the similarities between the large number of bubble-like structures, there is no unified methodology for efficiently computing them. Moreover, the number of all snarls can be quadratic in the input graph size. The authors of~\cite{paten2018ultrabubbles} used the \emph{cactus graph}~\cite{paten2011cactus} to prune some snarls in order to compute in linear time a \emph{snarl decomposition}, whose size is linear in the input graph size. While a snarl decomposition suffices for some applications (e.g. distance-based seed indexing), the goal of this paper is to identify \emph{all} bubble-like structures existing in the graph. 
Moreover, since superbubbles can indeed be computed in linear time~\cite{gartner2019direct}, it would seem reasonable to assume that such an approach can be adapted also for snarls. However, this achievement is heavily tailored to superbubbles, and, crucially, it also relies on the \emph{directed} nature of the input graph. Finally, this has been obtained after a series of papers~\cite{onodera2013detecting,loglinear,brankovic2016linear,gartner-revisited}, begging the question of how large undertaking obtaining a linear-time snarl algorithm is.\footnote{Initially, superbubbles were computable in $O(|V|(|E|+|V|))$-time in~\cite{onodera2013detecting}, and later improved to $O(\log|V|(|E|+|V|))$-time~\cite{loglinear}. Linear time $O(|V| + |E|)$ was first obtained for acyclic graphs in~\cite{brankovic2016linear}, then for general graphs in~\cite{gartner-revisited}, and then further simplified by~\cite{gartner2019direct}.}




\textbf{Contributions.} In this paper we show that snarls admit a representation whose size is linear in the input graph size, and which can also be computed in time linear in the input size. Thus, we can \emph{identify} all snarls in time linear in the input size. We obtain these algorithms via a method that also leads to a new linear-time algorithm computing all superbubbles of a directed graph. These lead, for the first time, to a \emph{unified} framework for finding both of types of structures, giving hope that other bubble-like structures could also be computed in linear time in the same manner (e.g. ultrabubbles and bibubbles, for which linear-time algorithms are currently missing).

One property that we use in this paper is that bubble-like structures previously mentioned are such that the two extremity vertices separate the inside of the structure from the rest of the graph, namely they form a \emph{2-separator} of the undirected counterpart of the input graph. We can then use graph theory developed for 2-separators. A key such notion is the \emph{SQPR tree} of a (biconnected) graph~\cite{battista1990on-line,gutwenger2001linear}, which encodes all 2-separators of the graph. While SPQR trees have already been applied in computational genomics (e.g.~to visualize scaffold graphs for metagenomic sequence assembly~\cite{metagenomescope} and for inference of local ancestry~\cite{Jafarzadeh2025.07.05.662656}), they have not been applied to solve the enumeration of existing bubble-like notions.

At a high level, our algorithms traverse the SPQR tree of the undirected counterpart of the input graph and, by using the original edge (bi)directions, evaluate certain properties in the tree, allowing us to conclude whether the 2-separators correspond to bubble-like structures. The power of the SPQR tree is two fold. On the one hand, it encodes candidates to bubble-like structures (the 2-separators). On the other hand, its tree structure allows for a dynamic programming-like approach where local information (of a component) can be re-used later (for another component). To obtain the linear running times we also overcome the following challenges independent to each structure.

    \textbf{Snarls.} The initial results on snarls by \cite{paten2018ultrabubbles} used cactus graphs (constructable in linear-time), to show that snarls correspond to either \emph{chain pairs} (whose number is linear in the input size), or to \emph{bridge pairs} (whose number can be quadratic in the input size). Here, we show that all snarls admit a representation which is linear in the input size, and can also be computed in linear time. 
    
    For this, we identify the quadratically-many snarls by characterizing a special set of vertices, which we call \emph{sign-consistent}. These are just cutvertices whose incidences in each corresponding component all have the same sign. By splitting each of these vertices we obtain a set of disjoint graphs, the \emph{sign-cut graphs}, which preserve the set of all original snarls and have the property that every snarl has its defining incidences in a \emph{single} sign-cut graph. Moreover, snarl components are fully contained inside sign-cut graphs, except for one simple case.
    Inside a sign-cut graph, the endpoints of snarls are either (i) both tips (vertices with all incident edges having the same sign in the sign-cut graph)---these can be quadratically-many; or (ii) both non-tips, which, importantly, form a 2-separator of a block (except the one simple case mentioned above)---which are linearly-many. To identify which 2-separators form a snarl, we characterize how snarls interact with each type of tree nodes (S, P, or R).
    Essentially, separability comes from analyzing the incidences at the vertices of the 2-separators and minimality follows from the connectivity properties of each type of node.
    
    For the linear-size representation, we further observe that \emph{any} pair of tip-tip vertices from case (i) forms a snarl. As such, these quadratically-many snarls can be simply represented as the set of tips in each sign-cut graph. Putting everything together, we obtain the following result:

    \begin{theorem}
        \label{thm:main}
        Given a bidirected graph $G$, there exists a representation of all snarls of $G$ of size $|V(G)| + |E(G)|$ consisting of sets $T_1, T_2, \dots, T_k$ and $S_1, S_2, \dots, S_\ell$, where
            \begin{enumerate}[nosep]
                \item each $T_i$ is a set of incidences of $G$, and any pair of incidences in $T_i$ identifies a snarl of $G$;
                \item each $S_i$ is a pair of incidences $\{ud_u, vd_v\}$ identifying a snarl of $G$;
                \item $\sum_{i=1}^{k} |T_i|=O(|V(G)|)$ and $\sum_{i=1}^{\ell} |S_i|= 2\ell = O(|V(G)|+|E(G)|)$.
            \end{enumerate}
        Moreover, this representation can be computed in time $O(|V(G)| + |E(G)|)$.
    \end{theorem}
    
    \textbf{Superbubbles.} Superbubbles present similarities to snarls, but also different challenges. Similarly to snarls (whose endpoints are in the same sign-cut graph) we show that the endpoints of superbubbles are in the same biconnected component of the graph (which now fully contains their inner component). The challenge with superbubbles is that their interior must be acyclic. This now requires traversing the tree in multiple phases, maintaining several properties (including acyclicity) via dynamic programming.

    First, we show that no superbubble can have endpoints that are non-adjacent vertices in an S-node of the SPQR tree. Thus, all 2-separators that are superbubble candidates are in SPQR-tree edges. Every such edge induces a separation of the graph, and each must be checked if it is a superbubble. As for snarls, we also characterize how superbubbles interact with each type of tree nodes (S, P, or R).
    
    To ensure linear-time, we reuse properties from the subtrees of the current tree edge. More specifically, we need to keep track of whether the subgraph contains graph sources, graph sinks, or if it is acyclic. If all these properties hold for a side of a 2-separator $\{s,t\}$, by a simple observation we have that $s$ reaches $t$ (or $t$ reaches $s$). Thus, at later states of the algorithm, this subgraph does not need to be reinspected, as the only relevant properties have been computed and the entire subgraph can be emulated as an edge $(s,t)$ (or $(t,s)$). During the DP, besides standard state examinations, essentially only these reachabilities are relevant, as they allow us to assign directions to the edges contained in the nodes of the SPQR tree. We thus obtain a \emph{directed graph} per tree-node, which will be used to decide acyclicity at later stages.
    
    A final challenge is the following. Let $\{\mu,\mu_1\}, \dots, \{\mu,\mu_k\}$ be SPQR tree edges, with $\mu,\mu_1,\dots,\mu_k$ SPQR tree nodes, encoding 2-partitions of the (edges of) the graph, $(A_1,B_1)$,\dots, $(A_k,B_k)$, where $\mu$ is ``contained'' in each $A_i$.
    Here, we have to decide the acyclicity of each of the graphs $A_i$. 
    While we do have the directed graphs mentioned above, we still cannot afford to solve each of these problems independently, otherwise we may get a quadratic running time algorithm. Thus, we manage to decide the acyclicity of all these graphs at once. For that, we develop a reduction to computing the set of \emph{feedback edges} (whose removal makes the graph acyclic), by adapting a classic algorithm~\cite{garey1978linear}. Due to lack of space, we present all results on superbubbles in \Cref{sec:superbubbles}. 

\textbf{Implementation and experiments.}
We implemented our algorithms in C++ (\url{https://github.com/algbio/BubbleFinder}) and evaluate our implementations on graphs built with PGGB~\cite{garrison2024pangenomegraphs}, \texttt{vg}~\cite{garrison2018variation} variation graphs and a pangenome de Bruijn graph.
Our experimental results indicate that even though our algorithmic framework is of a much more generic nature, we outcompete specialized algorithms on most datasets.
We are up to two times faster than \texttt{vg} when computing snarls (and identifyin them all), and up to more than 50 times faster than BubbleGun~\cite{dabbaghie2022bubblegun} when computing superbubbles.
And even when we are slower, we never take more than twice the time of the previous methods.




\section{Preliminaries}
\label{sec:preliminaries}



\textbf{Bidirected graphs.}
A \emph{bidirected graph} $G = (V, E)$ has a set of vertices $V=V(G)$ and a set of edges $E=E(G)$.
A \emph{sign} is a symbol $d \in \{+, -\}$, and the \emph{opposite sign} $\hat{d}$ of $d$ is defined as $\hat{+} = -$ and $\hat{-} = +$.
A pair $(v,d)$ where $v \in V$ and $d \in \{+, -\}$ is an \emph{incidence}, which we concisely write as $vd$, e.g.~$v+$ or $v-$.
An edge $e \in E(G)$ is an unordered pair of incidences $\{u d_u, v d_v\}$, and we say that $e$ is incident in/at $u$ (resp. $v$) with sign $d_u$ (resp. $d_v$).
We let $N^+_G(v)$ (resp. $N^-_G(v)$) denote the set of those vertices $x$ for which there is an edge $\{v+,xd_x\}$ (resp. $\{v-,xd_x\}$) in $G$.
We say that a bidirected graph $H$ is a \emph{subgraph} of $G$ (and write $H \subseteq G$) if $V(H) \subseteq V(G)$ and $E(H) \subseteq E(G)$.
We say that $v$ is a \emph{tip} in $G$ if no two incidences of $G$ in $v$ have distinct signs.
A subgraph $H$ of $G$ is \emph{maximal} w.r.t.~a given property if no proper supergraph of $H$ contained in $G$ has that property. 


Standard notions such as walk and cycle exist in the context of bidirected graphs (see., e.g.~\cite{medvedev2007computability}). From these one can also define standard concepts of connectivity. However, in this work, we refer to the connectivity of a bidirected graph in terms of the connectivity of its underlying undirected graph. The undirected graph of $G$ is denoted by $U(G)$ and is obtained from $G$ by ignoring the signs in all its incidences (and keeping parallel edges that possibly appear).

\textbf{Undirected graphs and connectivity.}
Let $H$ be an undirected graph and let $u,v$ be vertices. If there is an edge in $H$ whose endpoints are $u$ and $v$ then we denote that edge as $\{u,v\}$.
A \emph{$u$-$v$ path} in $H$ is a path between $u$ and $v$. The \emph{internal vertices} of a path are the vertices contained in the path except $u$ and $v$.
Graph $H$ is \emph{$k$-connected} if it has more than $k$ vertices and no subset of fewer than $k$ vertices disconnects the graph. By Menger's theorem~\cite{menger}, if a graph $H$ is $k$-connected then $H$ has $k$ internally vertex-disjoint paths between any two of its vertices.
A \emph{connected component} (or just \emph{component}) is a maximally connected subgraph.
We call a vertex a \emph{cutvertex} if its removal increases the number of connected components of the graph. A connected graph with no cutvertex is also \emph{biconnected}. A set of two vertices whose removal increases the number of connected components is called a \emph{separation pair}. A biconnected graph with no separation pair is also \emph{triconnected}. Notice that we allow biconnected (triconnected) graphs to have fewer than three (four) vertices. We call an edge a \emph{bridge} if its removal increases the number of connected components; a set of parallel edges whose removal increases the number of connected components is called a \emph{multi-bridge}.
A set $X \subseteq V(H)$ of at least $k$ vertices is \emph{$(<k)$-inseparable} in $H$ if no two vertices of $X$ can be separated (i.e., they end up in different components) by removing fewer than $k$ other vertices. A maximal $(<k)$-inseparable set of vertices is called a \emph{$k$-block}. A \emph{separation} of $H$ is a pair of vertex sets $(A,B)$ such that $V = A \cup B$, $A\setminus B$ and $B\setminus A$ are nonempty, and there is no edge between $A\setminus B$ and $B\setminus A$.

\begin{figure}[t]
   \centering
   \includegraphics{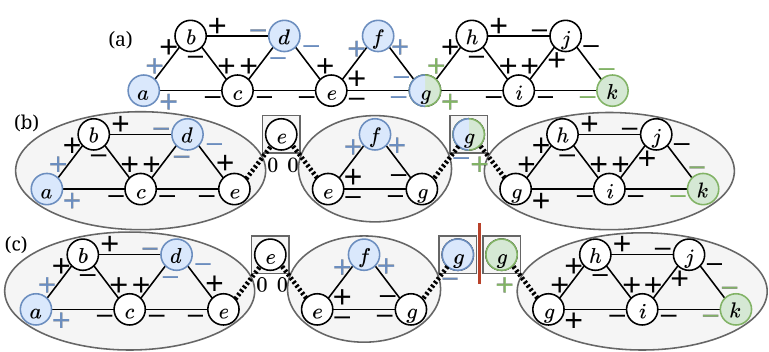}
   \caption{\textbf{Block-cut trees and representing quadratically many snarls in linear size.} (a) A bidirected graph $G$ with snarls $\{a+, d-\}$, $\{a+, f+\}$, $\{a+, g-\}$, $\{d-, f+\}$, $\{d-, g-\}$, $\{f+, g-\}$ (blue) and $\{g+, k-\}$ (green).
   (b) The block-cut tree of $G$.
   Block nodes are drawn as ellipses, and cutnodes as squares.
   Tree edges are bold dashed lines.
   On the cutnodes, we annotate the incidence signs for each block, using the special value $0$ for mixed signs.
   (c) The two sign-cut graphs of $G$), which are obtained by splitting the sign-consistent cutvertices (here, only $g$). The tips of the sign-cut graph on the left are $\{a+,d-,f+,g-\}$ (in blue), and those on the right are $\{g+,k-\}$ (in green). This is a linear-size representation of quadratically-many tip-tip snarls (i.e.~all pairs of vertices in each tip set).
   Cutvertex $e$ has at least one $0$, hence it is not sign-consistent, and hence it is not split.
   Cutvertex $g$ has no $0$, and thus it is split.
   No snarl can have endpoints in distinct sign-cut graphs (e.g. across the red cut), because otherwise $g$ would disprove its minimality. }
   \label{fig:bc-tree}
\end{figure}


\textbf{Block-cut trees.}
Let $H=(V,E)$ be an undirected connected graph with at least two vertices. It follows from the definition of $k$-block that a 2-block of $H$ is a maximal connected subgraph without cutvertices (see~\cite{diestel}). (For simplicity, we will refer to 2-blocks simply as blocks.) The \emph{block-cut tree} of $H$ is a tree with \emph{node set} $N$ and edge set $A$. The nodes in $N$ are of two types: \emph{block nodes} (either a maximal 2-connected subgraph or multi-bridges of $H$), and \emph{cutnodes}, (cutvertices of $H$).

The edges in $A$ represent how the blocks of $H$ are ``connected'' via the cutvertices of $H$ as follows.
Let $v$ be a cutvertex of $H$ and let $\mu$ be the cutnode of $N$ corresponding to $v$. Then $H-v$ consists of components $C_1,\dots,C_\ell$ ($\ell \geq 2$) and $\mu$ has $\ell$ neighbours in the tree, each corresponding to the block contained in $C_i + v$ that meets $v$.
Every edge of $H$ lies in a unique block of $H$~\cite{diestel}.
Notice that for any two vertices, there exists at most one block containing them both.


\begin{figure}
    \centering
    \includegraphics{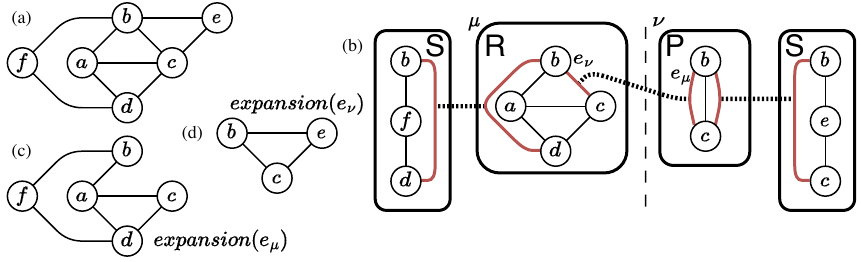}
    \caption{\textbf{Concepts of SPQR trees}. (a) An undirected graph $H$. (b) The SPQR tree of $H$. Tree edges are bold dashed lines. Red edges are virtual and black solid edges are real edges. Node $\mu$ is an R-node and its skeleton is 3-connected. Node $\nu$ is a P-node with $k=3$. The virtual edge $e_\nu$ of $\skel(\mu)$ pertains to $\nu$ and the virtual edge $e_\mu \in E(\skel(\nu))$ pertains to $\mu$. The tree edge $\{\nu,\mu\}$ induces a separation of $H$ with vertex sets $\{a,b,c,d,f\}$ and $\{b,c,e\}$. (c) The graph $\expansion(e_\mu)$. The real edges appearing to the left of the separation induce $\expansion(e_\mu)$. (d) The graph $\expansion(e_\nu)$.}
    \label{fig:spqr-nodes}
\end{figure}

\textbf{SPQR trees.}
SPQR trees represent the decomposition of a biconnected graph according to its separation pairs in a tree-like way, thus exposing the 3-blocks of the graph (in analogously to how block-cut trees exposed the 2-blocks and how they are connected via cutvertices). They were first formally defined by Tamassia and Di Battista~\cite{battista1990on-line}, but were informally known before~\cite{lane1937structural,hopcroft1973dividing,bienstock1988complexity}.
They can be constructed in linear time~\cite{hopcroft1973dividing,gutwenger2001linear}, and if unrooted they are unique in undirected graphs~\cite{battista1990on-line}. SPQR trees proved to be a valuable tool in the design of algorithms for different problems~\cite{rotenberg,di1996line,onlinegraphalgorithms}.

To define SPQR trees we need some basic definitions. Let $H$ be an undirected biconnected graph with at least two edges. A \emph{split pair} of $H$ is a separation pair or an edge of $H$. A \emph{split component} of a split pair $\{u,v\}$ is an edge $\{u,v\}$ or a maximal subgraph $C$ of $H$ such that $\{u, v\}$ is not a split pair of $C$. Let $\{s,t\}$ be a split pair of $H$. A maximal split pair $\{u, v\}$ of $H$ with respect to $\{s, t\}$ is such that for any other split pair $\{u', v'\}$ vertices $u, v, s$, and $t$ are in the same split component.

\textbf{SPQR tree construction.} We now give a recursive definition of SPQR trees based on~\cite{gutwenger2005inserting} (see~\cite{rotenberg} for a different but equivalent definition).
Let $e=\{s,t\}$ be a designated edge, called the \emph{reference edge}. The SPQR tree $T$ of $H$ with respect to $e$ is defined as a rooted tree with nodes of four types: S~(series), P~(parallel), Q~(single edge), and R~(rigid). Each node $\mu$ in $T$ has an associated biconnected graph \emph{$\skel(\mu)$}, called the \emph{skeleton} of $\mu$ with $V(\skel(\mu)) \subseteq V(H)$. Tree $T$ is one of the following cases:

\begin{description}[nosep]
    \item[Trivial case:] If $H$ consists of exactly two parallel edges between $s$ and $t$, then $T$ consists of a single Q-node whose skeleton is $H$ itself.
    \item[Parallel case:] If the split pair $e = \{s,t\}$ has $k+1$ split components $H_0,\dots,H_k$ with $k\geq 2$ where $H_0$ is the split component containing $e$, the root of $T$ is a P-node $\mu$ whose skeleton consists of $k+1$ parallel edges $e_0,\dots,e_k$ between $s$ and $t$, where $e_0 = e$.
    \item[Series case:] Otherwise, if the split pair $\{s,t\}$ has exactly two split components, where one of them is $e$ and the other is denoted as $H'$. If $H'$ is a chain of 2-blocks $H_1,\dots,H_k$ separated by cutvertices $c_1,\dots,c_{k-1}$ $(k \geq 2)$ in this order from $s$ to $t$, then the root of $T$ is an S-node $\mu$ whose skeleton is a cycle $e_0, e_1,..., e_k$ where $e_0$ = $e$, $c_0=s$, $c_k=t$, and $e_i = (c_{i-1}, c_i)$ $(i=1,\dots,k)$.
    \item[Rigid case:] If none of the above cases applies, let $\{s_1,t_1\},\dots,\{s_k,t_k\}$ be the maximal split pairs of $H$ with respect to $\{s,t\}$ $(k\geq1)$, and, for $i=1,\dots,k$ let $H_i$ be the union of all the split components of $\{s_i,t_i\}$ except the one containing $e$. The root of $T$ is an R-node $\mu$ whose skeleton is obtained from $H$ by replacing each subgraph $H_i$ with the edge $e_i = \{s_i, t_i\}$.
\end{description}
Except for the trivial case, $\mu$ has children $\mu_1,\dots,\mu_k$, such that $\mu_i$ is the root of the SPQR tree of $H_i \cup e_i$ with respect to $e_i$ for $i=1,\dots,k$; notice how the reference edge $e_i$ in $\skel(\mu_i)$ ensures that $\skel(\mu_i)$ is biconnected. Once the construction/recursion is finished, we add a Q-node with vertex set $\{s,t\}$ representing the first reference edge $e=\{s,t\}$ and make it a child of the root of $T$.
Node $\mu_i$ is associated with edge $e_i$ of the skeleton of its parent $\mu$, called the \emph{virtual edge} of $\mu_i$ in $\skel(\mu)$ $(i=1,\dots,k)$. Conversely, $\mu$ is implicitly associated with the reference edge $e_i$ in $\skel(\mu_i)$.
Notice that reference and virtual edges encode the same information: two subgraphs of $H$ and how they attach to each other. Indeed, a reference edge $e$ in a node $\mu$ is just another virtual edge with the additional property of pointing to the parent of $\mu$ in $T$.
We say that $\mu$ is the \emph{pertinent node} of $e_i \in E(\skel(\mu_i))$ (or that $e_i \in E(\skel(\mu_i))$ pertains to $\mu$), and that $\mu_i$ is the pertinent node of $e_i \in E(\skel(\mu))$ (or that $e_i \in E(\skel(\mu))$ pertains to $\mu_i$).

\textbf{Additional definitions.}
For simplicity, we will omit Q-nodes from the SPQR tree.
This amounts to replacing every virtual edge pertaining to a Q-node by a \emph{real edge}, and then deleting every Q-node from the tree. The edges of a skeleton are then either real or virtual.

Suppose now that $\nu$ is the parent of $\mu$ in $T$. Let $e_\nu \in \skel(\mu)$ be the edge pertaining to $\nu$ and let $e_\mu \in \skel(\nu)$ be the edge pertaining to $\mu$. Let $\{s,t\}$ be the endpoints of $e_\nu$ and $e_\mu$. Deleting the edge $\{\nu,\mu\}$ from $T$ disconnects $T$ into two subtrees, $T_\nu$ containing $\nu$ and $T_\mu$ containing $\mu$. The \emph{expansion graph} of $e_\nu$, denoted as $\expansion(e_\nu)$, is the subgraph induced in $H$ by the real edges contained in the skeletons of the nodes in $T_\nu$. The graph $\expansion(e_\mu)$ is defined analogously with respect to $T_\mu$. If $e=\{u,v\}$ is a real edge in $skeleton(\nu)$ then $\expansion(e)$ is a graph with a single edge $\{u,v\}$.
Notice that each edge of $T$ encodes a separation. More specifically, $(V(\expansion(e_\nu)),V(\expansion(e_\mu)))$ is a separation, $\expansion(e_\nu)\cup \expansion(e_{\mu})=H$, $E(\expansion(e_\nu))\cap E(\expansion(e_{\mu}))=\emptyset$, and $V(\expansion(e_\nu)) \cap V(\expansion(e_\mu)) = V(\skel(\nu) \cap V(\skel(\mu)) = \{s,t\}$.
Notice also that for every node $\mu$ of $T$ whose skeleton has edges $e_1,\dots,e_k$, the graph $\bigcup_{i=1}^k \expansion(e_i)$ is just $H$. In SPQR trees, no two S-nodes and no two P-nodes are adjacent~\cite{di1996line}.

For simplicity, we allow building the SPQR tree of bidirected/directed graphs (connectivity is seen from their undirected counterparts). Moreover, we assume that real edges encode their relevant properties in the bidirected/directed graph, which also applies to the $\expansion$ operator and to the blocks of $U(G)$. Furthermore, in this paper we only build SPQR trees of 2-connected graphs.




The next statements are well known results about SPQR trees. \Cref{lem:spqr-total-size} below is given in a context where Q-nodes are part of the tree. Clearly, by removing Q-nodes the bounds remain valid.

\begin{prop}[SPQR trees and separation/split pairs]\label{prop:spqr-tree-contains-split-pairs}
    Let $H$ be an undirected 2-connected graph and let $T$ be its SPQR trees with Q-nodes omitted. For each S-node $\mu$ of $T$, let $X_\mu$ denote the set of all pairs of nonadjacent vertices in $\skel(\mu)$. Then the union of the virtual edges over the skeletons of the nodes of $T$ together with the union of all the $X_\mu$ is exactly the set of separation pairs of $H$. If Q-nodes are included in the tree, then the resulting set is exactly the set of split pairs of $H$.
\end{prop}

\begin{lemma}[SPQR trees require linear space~\cite{onlinegraphalgorithms}]
\label{lem:spqr-total-size}
    Let $H=(V,E)$ be an undirected biconnected graph. The SPQR tree $T$ of $H$ has $O(|V(H)|)$ nodes and the total number of edges in the skeletons is $O(|E(H)|)$.
\end{lemma}


\section{The snarl algorithm}
\label{sec:snarls-main-text}

Originally, snarls have been defined on a \emph{biedged} graph~\cite{paten2018ultrabubbles}. In this paper we use an equivalent definition of snarls in bidirected graphs which exposes more conveniently the property that the endpoint vertices of the snarl may form a separation pair. We prove this equivalence in \Cref{sec:biedged-snarls}.

In the rest of this section we will assume, without loss of generality, that $G=(V,E)$ is a \emph{connected} bidirected graph. To give our equivalent snarl characterization we need to introduce some more terminology. The \emph{splitting} operation takes an incidence $xd_x$ and adds a new vertex $x'$ to $G$.
Then it changes the incidences $x\hat{d_x}$ in all edges of $G$ into incidences $x'\hat{d_x}$.
As a result, all edges incident with sign $d_x$ to $x$ will be incident to $x'$ instead.
We assume that there are no parallel edges as they have no effect on snarls (two edges $\{xd_x,yd_y\}$ and $\{zd_z,wd_w\}$ are parallel if $x=z$, $d_x=d_z$, $y=w$, $d_y=d_w$).
We can now define (\Cref{def:snarl}) snarls and snarl components (see also \Cref{fig:SPQR tree}). Snarl components exist in the graph resulting from splitting both incidences but contain only vertices and edges of the original graph.

\begin{restatable}[Snarl, Snarl component]{definition}{snarlandcomponent}
    \label{def:snarl}
    A pair of incidences $\{xd_x, yd_y\}$ with $x \neq y$ is a \emph{snarl} if
    \begin{enumerate}[nosep]
        \item[(a)] \emph{separable:}
        the graph created by splitting the incidences $xd_x$ and $yd_y$ contains a separate component $X$ containing $x$ and $y$ but not the vertices $x'$ and $y'$ created by the split operation.
        We call $X$ the \emph{snarl component} of $\{xd_x, yd_y\}$.
        
        \item[(b)] \emph{minimal:}
        no incidence $zd_z$ with vertex $z \in X$ different from $x$ and $y$ exists such that $\{xd_x, zd_z\}$ and $\{z\hat{d}_z, yd_y\}$ are separable.
    \end{enumerate}
\end{restatable}

\begin{figure}[t]
   \centering
   \includegraphics{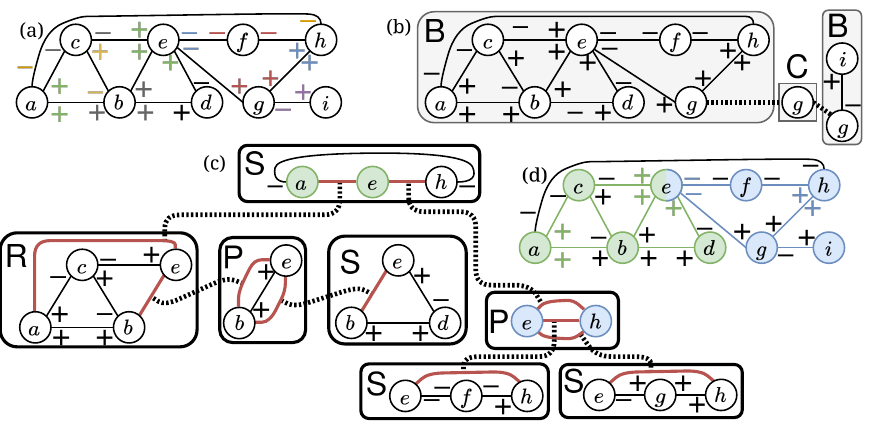}
   \caption{\textbf{Snarls and 2-separators in the SPQR tree.}
   (a) A bidirected graph.
   The snarls are $\{a+, e+\}$ (green), $\{c+, b-\}$ (yellow), $\{c-, b+\}$ (gray), $\{e-, h+\}$ (blue), $\{a-, h-\}$ (orange), $\{f-, g+\}$ (red), $\{g-, i+\}$ (purple).
   Snarl $\{c+, b-\}$ is contained in snarl $\{a+, e+\}$.
   (b) The corresponding block-cut~tree.
   Blocks are marked with a B, and cutnodes with a C.
   There is one cutnode $g$, which separates the two blocks.
   (c) The SPQR~tree of the largest block.
   Its nodes are marked with their types and contain their skeletons.
   Red edges are virtual, and dashed black edges connect pairs of virtual edges.
   The snarls $\{a+, e+\}$ (green), $\{e-, h+\}$ (blue) are highlighted.
   Snarl $\{g-, i+\}$ is not part of the block. Snarl $\{f-, g+\}$ is a special case that spans the whole block and is handled when finding snarls in the sign-cut tree (see \Cref{thm:where-are-snarls-after-cutting} and \Cref{prop:snarls-both-tips-component}).
   Snarls $\{c+, b-\}$, $\{c-, b+\}$ and $\{a-, h-\}$ are single-edge snarls that are handled separately (see \Cref{prop:snarl-edge-case}).
   (d) The snarl components of $\{a+, e+\}$ (green) and $\{e-, h+\}$ (blue).
   Vertex $e$ is in both snarl components.}
   \label{fig:SPQR tree}
\end{figure}




\subsection{Properties of snarls and connectivity}

Let $G$ be a bidirected graph. Let us call a vertex $x$ \emph{sign-consistent} if $x$ is a cutvertex of $G$ such that all incidences in each component of $G-v$ have the same sign. Notice that splitting a sign-consistent vertex $u$ paired with any incidence in $\{+,-\}$ creates two components, one containing the positive incidences at $u$ in $G$ and the other containing the negative incidences.
We can now define sign-cut graphs.

\begin{definition}[Sign-cut graphs]
    Let $G$ be a bidirected graph and let $Y \subseteq V(G)$ be the set of sign-consistent vertices of $G$. The \emph{sign-cut graphs} of $G$ are the graphs resulting from splitting each cutvertex $y\in Y$ together with any sign in $\{+,-\}$, and relabeling each new vertex $y'$ as $y$.
\end{definition}

It is easily seen that for incidences $ud_u$ and $vd_v$ with $u \neq v$, splitting $ud_u$ and then $vd_v$ yields the same graph as splitting $vd_v$ and then $ud_u$. Hence, any ordering of splits during the construction of sign-cut graphs yields the same set of components, and thus sign-cut graphs are well defined.

Every sign-consistent vertex of $G$ becomes a tip in both sign-cut graphs it appears in (one with its positive incidences and the other with its negative incidences), and moreover a vertex of $G$ is contained in two sign-cut graphs if and only if it is sign-consistent. One way to see sign-cut graphs is through the block-cut tree of $G$ (see~\Cref{fig:bc-tree}). Moreover, sign-cut graphs are easily built in linear time from block-cut trees by examining the incidences of $G$ at its cutvertices. Clearly, sign-cut graphs partition the incidences of $G$ and the set of blocks of $G$ coincide with the union of the blocks over its sign-cut graphs.

The next result tells us that the incidences of snarls are confined to sign-cut graphs, so for any snarl $\{ud_u,vd_v\}$ there is exactly one sign-cut graph containing $ud_u$ and $vd_v$.

\begin{restatable}{lemma}{snarlsinsidesigncuttrees}
\label{lem:snarls-inside-sign-cut-tree}
    Let $G$ be a bidirected graph, let $F_1$ and $F_2$ be distinct sign-cut graphs of $G$, and let $u,v$ be vertices where $u \in V(F_1)$ and $v \in V(F_2)$.
    If $F_1$ has an incidence in $u$ with sign $d_u$ and $F_2$ has an incidence in $v$ with sign $d_v$ then $\{u d_u,v d_v\}$ is not a snarl of $G$.
\end{restatable}

In fact, we can show an equivalence between snarls of $G$ and snarls of the sign-cut graphs of $G$.

\begin{restatable}{lemma}{snarlsGF}
\label{lem:snarls-G-F}
    Let $G$ be a bidirected graph, let $\{ud_u,vd_v\}$ be a pair of incidences. Then $\{ud_u,vd_v\}$ is a snarl of $G$ if and only if there is a sign-cut graph $F$ of $G$ such that $\{ud_u,vd_v\}$ is a snarl of $F$.
\end{restatable}

Ultimately, together with other results shown in \Cref{sec:missing-snarls}, sign-cut graphs allow us to pinpoint all potential snarls in the graph as hinted in the next theorem.

\begin{restatable}{theorem}{snarlsaftercutting}
\label{thm:where-are-snarls-after-cutting}
    Let $G$ be a bidirected graph and let $F$ be a sign-cut graph of $G$. The following hold for $G$.
    \begin{enumerate}[nosep]
        \item If $u$ is a non-tip and $x$ is a tip with sign $q\in\signs$ in $F$ then $\{ud_u,xq\}$ is not a snarl for any sign $d_u \in \signs$.
        \item If $u$ and $v$ are tips in $F$ with signs $d_u,d_v \in\signs$ then $\{ud_u,vd_v\}$ is a snarl.
        \item If $\{ud_u,vd_v\}$ is a snarl and $u$ and $v$ are non-tips in $F$ then there is a block of $F$ containing both $u$ and $v$, and no other block containing $u$ has incidences of opposite signs in $u$ and no other block containing $v$ has incidences of opposite signs in $v$.
    \end{enumerate}
\end{restatable}

Item (3) of
\Cref{thm:where-are-snarls-after-cutting} highlights a basic property of snarls that we shall use routinely in what follows. Let $v$ be a vertex and let $H$ be a block containing $v$. If $H'$ is a block distinct from $H$ that contains $v$ and has incidences of opposite signs at $v$, then $H'$ is a \emph{dangling block} at $v$ with respect to $H$ (for example, non cutvertices have no dangling blocks). The next result is a key ingredient for our algorithm. Essentially, we show that any snarl implicitly encodes a split pair of $U(G)$.

\begin{restatable}[Snarls and split pairs]{theorem}{snarlssplitpairs}
\label{thm:snarls-split-pairs}
    Let $G$ be a bidirected graph and let $\{ud_u,vd_v\}$ be a snarl of a sign-cut graph of $G$ where $u$ is a non-tip. Then $\{u,v\}$ is a split pair of a block of $U(G)$.
\end{restatable}


\begin{figure}
   \centering
   \includegraphics{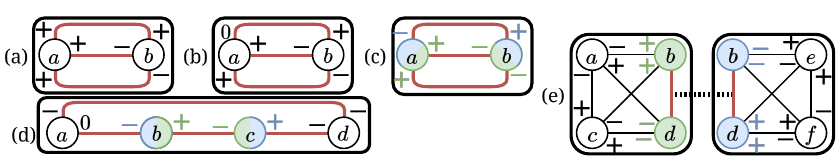}
   \caption{
   \textbf{Snarls and 2-separators in S, P and R nodes.}
   Red edges are virtual, and are marked with $+$ or $-$ at their endpoints if all their represented incidences at that endpoint are $+$ or $-$.
   If the incidences are mixed, we mark them with a $0$.
   No vertex has dangling blocks.
   (a) and (b) are P-nodes without a consistent partitioning of the edges at their endpoints.
   (c) is a P-node with a consistent partitioning of the edges, and hence it contains the snarls $\{a-, b+\}$ and $\{a+, b-\}$.
   (d) is an S-node where for $b$ and $c$ the signs are partitioned between the two virtual edges, and for $a$ and $d$ the signs are mixed.
   Hence, it contains the snarls $\{b+, c-\}$ and $\{b-, c+\}$.
   (e) is a pair of adjacent R-nodes.
   Since the signs of $b$ and $d$ are partitioned along the boundaries of the R-nodes, they contain the snarls $\{b+, d-\}$ and $\{b-, d+\}$.
   }
   \label{fig:spqr-snarls-cases}
\end{figure}



\subsection{Finding all snarls}

In this section we describe an algorithm finding all snarls in a bidirected graph.
The algorithm computes the sign-cut graphs of $G$ and only reports snarls within sign-cut graphs (correctness of this step is ensured by \Cref{lem:snarl-equivalence}).
For each sign-cut graph $F$ the algorithm first builds a list of all the incidences corresponding to tips of $F$; by (2) of \Cref{thm:where-are-snarls-after-cutting}, every pair of incidences in this list forms a snarl. Then the algorithm runs a dedicated routine on every block of $F$, where the \emph{non-tip with non-tip} snarls are reported (see \Cref{alg:snarls} in the Appendix); we describe this step next.

First, we examine pairs $\{ud_u,vd_v\}$ such that $\{u,v\}$ is a separation pair of a block of $U(G)$. For this, we use the SPQR tree representation to find separation pairs, which, together with some signs, are separable in the bidirected sense; the minimality property required by snarls will follow from the structure of the SPQR tree nodes (see \Cref{fig:spqr-snarls-cases} for examples of the various cases). Lastly, we examine pairs $\{ud_u,vd_v\}$ such that $\{u,v\}$ is an edge of $U(G)$ with some simple conditions (see \Cref{prop:snarl-edge-case} for details and also \Cref{alg:find-edge-snarls}).

In the S-nodes we mark those vertices where all incidences to their left have sign $q \in\signs$ and all incidences to their right have sign $\hat{q}$, and moreover these must not have dangling blocks. More formally, let $\mu$ be an S-node (within a block $H$) whose skeleton contains a vertex $v$ and let $e,e'\in\skel(\mu)$ denote the edges incident to $v$. Say that $v$ is \emph{good} if all incidences in $\expansion(e)$ (resp. $\expansion(e')$) in $v$ have sign $q \in\signs$ (resp. $\hat{q}$) and $v$ has no dangling blocks with respect to $H$.
We can build a circular list respecting the order of the vertices in $\skel(\mu)$ such that any two consecutive vertices $u$ and $v$ together with the incidences to their right, and left, respectively, are separable, and moreover no vertex in the resulting component violates minimality (essentially, such vertex has to be good, but $u$ and $v$ are consecutive in the list). As a result, we can report as snarls all pairs of consecutive incidences in the list (see \Cref{fig:spqr-snarls-cases}(d) and \Cref{alg:find-snarls-in-S} in the Appendix). We formalize this discussion in the following:

\begin{restatable}[Snarls and S-nodes]{prop}{snodesnarls}
\label{prop:S-node-snarls}
    Let $G$ be a bidirected graph, let $H$ be a maximal 2-connected subgraph of $G$, let $T$ be the SPQR tree of $H$, and let $\mu$ be an S-node of $H$ with vertices $v_1,\dots,v_k$ and virtual edges $e_1,\dots,e_k$ with the endpoints of the edge $e_i$ being $v_i$ and $v_{(i+1)\,\textrm{mod}\,k}$ $(k \geq 3)$.
    Let $v_{i_1},\dots,v_{i_q}$ be the list of the good vertices of $\mu$ where $i_1,\dots,i_q \subseteq \{1,\dots,k\}$ and $i_1 < \dots < i_q$, and let $\hat{d}_{i_j},d_{i_j} \in\signs$ denote the signs of the incidences at $v_{i_j}$ in the expansions of the edges $e_{(i - 1)\,\textrm{mod}\,k}$ and $e_i$, respectively, for $j \in \{1,\dots,q\}$ $(q \geq 0)$.
    Then $\{v_{i_1}d_{v_1}, v_{i_2}\hat{d}_{v_2}\}, \dots, \{v_{i_q}d_{v_q}, v_{i_1}\hat{d}_{v_1}\}$, are snarls.
\end{restatable}

In the P-nodes (with vertex set $\{u,v\}$) we group the virtual edges $e$ whose expansions have $+$ incidences in $u$ and those whose expansions have $-$ incidences in $u$, and do the same for $v$. The resulting sets with respect to $u$ have to be disjoint, and the same for $v$ (notice that a virtual edge whose expansion has incidences of opposite signs in, e.g., $u$, appears in both sets of $u$). Further, one set of $u$ has to match exactly one of $v$, implying that the remaining two sets also match. If any of these conditions does not hold then separability is violated in any pair of incidences containing the vertices $u$ or $v$.
More formally, we can give a characterization of separability of two incidences $\{ud_u,vd_v\}$ whenever a P-node has vertex set $\{u,v\}$ (see \Cref{fig:spqr-snarls-cases}(a)--(c) and \Cref{alg:find-snarls-in-P} in the Appendix). Minimality is handled later on and follows from the structure of the P-nodes. Here, there is also an interplay with S-nodes, hence why S-nodes are examined first.

\begin{restatable}[Snarls and P-nodes]{prop}{pnodesnarls}
\label{prop:P-node-snarls}
    Let $G$ be a bidirected graph, let $F$ be a sign-cut graph of $G$, and let $H$ be a maximal 2-connected subgraph of $H$.
    Let $T$ be the SPQR tree of $H$ and let $\mu \in V(T)$ be a P-node. Let $e_1,\dots,e_k$ denote the edges of $\skel(\mu)$ with endpoints $\{u,v\}$ $(k\geq3)$. Let $d_u,d_v\in\signs$ be signs.
    Let $E^{d_u}_u = \{ e_i : V(\expansion(e_i))\cap N^{d_u}(u) \neq \emptyset \}$, $E^{\hat{d}_u}_u = \{ e_i : V(\expansion(e_i))\cap N^{\hat{d}_u}(u) \neq \emptyset \}$, and $E^{d_v}_v = \{ e_i : V(\expansion(e_i))\cap N^{d_v}(v) \neq \emptyset \}$.
    Then $\{ud_u,vd_v\}$ is separable in $F$ if and only if $E^{d_u}_u \neq \emptyset$, $E^{d_u}_u \cap E^{\hat{d}_u}_u = \emptyset$, $E^{d_v}_v \cap E^{\hat{d}_v}_v = \emptyset$, $E^{d_u}_u = E^{d_v}_v$, and $u$ and $v$ have no dangling blocks with respect to $H$.
\end{restatable}

Now, separation pairs which are the vertices of an S-node or a P-node have been processed. Thus, the only separation pairs left to be examined are those encoded in a tree-edge with R-nodes endpoints.
This case is similar to the P-node. If $\{u,v\}$ is the virtual edge (with one copy in each R-node), then the incidences in one side of the separation at $u$ should all have the same sign, and the same for $v$. Then, the incidences in the other side of the separation should have the opposite signs at $u$ and $v$, respectively. As usual, $u,v$ have no dangling blocks. Minimality follows from the fact that R-nodes have too many $u$-$v$ paths for a single vertex to violate minimality. See \Cref{fig:spqr-snarls-cases}(e) and \Cref{alg:find-snarls-in-RR} in the Appendix.

\begin{restatable}[Snarls and R-nodes]{prop}{rnodesnarls}
\label{prop:R-node-snarls}
    Let $G$ be a bidirected graph, $H$ a maximal 2-connected subgraph of $G$, $T$ the SPQR tree of $H$, and $\nu$ an R-node of $T$. Let $e_\mu \in E(\skel(\nu))$ be a virtual edge with pertaining node $\mu$. Let $e_\nu \in E(\skel(\mu))$ be the virtual edge pertaining to $\nu$.
    If $\expansion(e_\nu)$ has only $ud_u$ and $vd_v$ incidences at $u$ and $v$ respectively, and $\expansion(e_\mu)$ has only $u\hat{d}_u$ and $v\hat{d}_v$ incidences at $u$ and $v$ respectively, and $u$ and $v$ have no dangling blocks with respect to $H$, then $\{ud_u,vd_v\}$ is a snarl.
\end{restatable}

See \Cref{sec:missing-snarls} for a detailed description of the algorithm (\Cref{alg:snarls}) cited in \Cref{thm:snarls-correct-time} below. From the proof of \Cref{thm:snarls-correct-time}, also \Cref{thm:main} follows.

\begin{restatable}{theorem}{snarlscorrecttime}
\label{thm:snarls-correct-time}
    Let $G$ be a bidirected graph. There exists an algorithm identifying all snarls, which can be implemented to run in time $O(|V(G)| + |E(G)|)$.
\end{restatable}

\section{Experiments}

\begin{table}
\centering
\begin{adjustbox}{max width=\textwidth}
\begin{tabular}{llcccccccc}
\toprule
Group & Dataset &
\multicolumn{1}{c}{$n$} &
\multicolumn{1}{c}{$m$} &
\multicolumn{2}{c}{\cellcolor{famSnarl}\textbf{BubbleFinder}} &
\cellcolor{famSnarl}\textbf{vg snarls} &
\multicolumn{2}{c}{\cellcolor{famBiSB}\textbf{BubbleFinder}} &
\cellcolor{famBiSB}\textbf{BubbleGun} \\
& &
\multicolumn{1}{c}{$(\times 10^6)$} &
\multicolumn{1}{c}{$(\times 10^6)$} &
\cellcolor{famSnarl}{I/O+ALGO} &
\cellcolor{famSnarl}{all} &
\cellcolor{famSnarl}{} &
\cellcolor{famBiSB}{I/O+ALGO} &
\cellcolor{famBiSB}{all} &
\cellcolor{famBiSB}{} \\
\midrule
PGGB & \textsf{50x E. coli}    & $1.6$  & $2.1$
    & \textbf{0:22} & \textbf{0:33} & 1:10
    & \textbf{0:26} & 0:48
    & 0:40 \\
& \textsf{Primate Chr. 6} & $34.3$ & $47.1$
    & \textbf{11:09} & \textbf{14:23} & 22:29
    & \textbf{11:20} & \textbf{19:47}
    & 2:21:26 \\
& \textsf{Tomato Chr. 2}  & $2.3$  & $3.2$
    & \textbf{0:33} & \textbf{0:52} & 1:18
    & \textbf{0:35} & 1:07
    & 0:52 \\
& \textsf{Mouse Chr.~19}  & $6.2$  & $8.6$
    & \textbf{1:53} &   \textbf{2:25}  & 3:13
    & \textbf{2:02} & \textbf{3:33}
    & 6:49 \\
\midrule
\texttt{vg} & \textsf{Chromosome 1}   & $18.8$ & $25.7$
    & 12:30 & 27:12 & \textbf{7:24}
    & \textbf{21:27} & \textbf{42:23}
    & \cellcolor[HTML]{F2F2F2} TO \\
& \textsf{Chromosome 10}  & $11.6$ & $15.9$
    & 7:27 & 16:20  & \textbf{4:15}
    & \textbf{12:52} & \textbf{25:28}
    & \cellcolor[HTML]{F2F2F2} TO \\
& \textsf{Chromosome 22}  & $3.2$  & $4.4$
    & 1:59  &  4:23  & \textbf{1:04}
    & \textbf{3:29}  & \textbf{6:56}
    & \cellcolor[HTML]{F2F2F2} TO \\
\midrule
DBG & \textsf{10x M. xanthus} & $1.6$ & $2.1$
    & 0:31 & 0:42 & \textbf{0:29} & 0:50 & 1:26 & \textbf{0:26} \\
\bottomrule
\end{tabular}
\end{adjustbox}
\caption{Wall clock execution times measured for each dataset and tool, in the format min:sec or h:min:sec. Snarls are in blue, superbubbles are in red.
Cells containing ``TO'' indicate a time out after three hours. Columns $n$ and $m$ denote the number of nodes and edges, in millions.
Columns labeled ``I/O+ALGO'' indicate that the timing does not include the building of the BC and SPQR trees.
We bold either the state-of-the-art or our running times that beat the state-of-the-art.
}
\label{tab:wall_clock_seq}
\end{table}

\textbf{Datasets.} We evaluate our algorithm on three groups of pangenome graphs. See \Cref{sec:additional-experimental-details} for accession codes and more detailed descriptions.
Datasets \textsf{50x~E.~coli}, \textsf{Primate~Chr.~6}, \textsf{Tomato~Chr.~2} and \textsf{Mouse~Chr.~19} are pangenome graphs constructed with PGGB~\cite{garrison2024pangenomegraphs}.
Datasets \textsf{Chromosome~1|10|22} are variant graphs for human chromosome 1|10|22 constructed from the ($\sim$6.5~millions) variant calls from phase~3 of the 1000~Genomes Project~\cite{1000genomesproject}.
We constructed these graphs using the \texttt{vg}~toolkit~\cite{garrison2018variation}.
Dataset \textsf{10x~M.~xanthus} is a de Bruijn graph with a $k$-mer size of 41 representing the pangenome of 10 Myxococcus xanthus genomes, which was used to evaluate BubbleGun~\cite{dabbaghie2022bubblegun}.

\textbf{Our tool.} We implemented our algorithms in C++ in our new tool BubbleFinder (\url{https://github.com/algbio/BubbleFinder}) and our experimental pipeline in Snakemake~\cite{snakemake} (\url{https://github.com/algbio/BubbleFinder-experiments}).
We use the graph data structures of BC tree and SPQR tree provided by the Open Graph Drawing Framework~\cite{chimani13ogdf}.
To the best of our knowledge, OGDF is the fastest SPQR tree implementation publicly available. However, OGDF does not scale to biological-sized datasets: we had to increase the process stack size to hundreds of gigabytes to build BC/SPQR trees for the largest datasets. We believe that, despite these issues, our implementation provides a good first impression of the practical applicability of our algorithms.


\textbf{Existing tools.} We compare our implementation for finding snarls against the \texttt{snarls} subcommand of the \texttt{vg}~toolkit~\cite{paten2018ultrabubbles}.
We use flag \texttt{-a} to compute snarls (default is ultrabubbles) and flag \texttt{-T} to report also single-edge snarls.
Note that \texttt{vg snarls -a -T} does not compute all snarls, but only an ``snarl decomposition'', which contains a linear amount of snarls.
In contrast, BubbleFinder computes all snarls. For superbubbles, we compare our implementation against BubbleGun~\cite{dabbaghie2022bubblegun}.
BubbleGun is a Python implementation of the initial $O(|V|(|E|+|V|))$-time superbubble algorithm~\cite{onodera2013detecting} on the (virtually) doubled representation of the bidirected graph.
Besides computing superbubbles, BubbleGun also chains them and filters circular chains.
Additionally, BubbleGun allows the arc $(t, s)$ to exist in a superbubble $(s, t)$.

\begin{table}
\centering
\begin{adjustbox}{max width=\textwidth}
\begin{tabular}{l c c c c c c c c c c c >{\centering\arraybackslash}p{2.0cm}}
\toprule
Dataset & $n$ & $m$ &
\multicolumn{4}{c}{\cellcolor{famSnarl}\textbf{BubbleFinder (all)}} &
\cellcolor{famSnarl}\textbf{vg snarls} &
\multicolumn{4}{c}{\cellcolor{famBiSB}\textbf{BubbleFinder (all)}} &
\cellcolor{famBiSB}\textbf{BubbleGun} \\
&
$(\times 10^6)$ &
$(\times 10^6)$ &
\cellcolor{famSnarl}{$t=1$} &
\cellcolor{famSnarl}{$t=4$} &
\cellcolor{famSnarl}{$t=8$} &
\cellcolor{famSnarl}{$t=16$} &
\cellcolor{famSnarl}{} &
\cellcolor{famBiSB}{$t=1$} &
\cellcolor{famBiSB}{$t=4$} &
\cellcolor{famBiSB}{$t=8$} &
\cellcolor{famBiSB}{$t=16$} &
\cellcolor{famBiSB}{} \\
\midrule
\textsf{Chromosome 1}   & $18.8$ & $25.7$ & 27:12 & 11:32 & 8:46 & 7:33 & \textbf{7:24} & 42:23 & 17:35 & 13:04 & \textbf{11:09} & \cellcolor[HTML]{F2F2F2} TO  \\
\textsf{Chromosome 10}  & $11.6$ & $15.9$ & 16:20 & 6:57 & 5:16 & 4:29 & \textbf{4:15} & 25:28 & 10:39 & 7:49 & \textbf{6:37} & \cellcolor[HTML]{F2F2F2} TO \\
\textsf{Chromosome 22}  & $3.2$ & $4.4$ & 4:23 & 1:51 & 1:23 & 1:11 & \textbf{1:04} & 6:56 & 2:47 & 2:03 & \textbf{1:43} & \cellcolor[HTML]{F2F2F2} TO  \\
\bottomrule
\end{tabular}
\end{adjustbox}
\caption{Wall clock times for executing BubbleFinder with multiple threads, in the format min:sec. Snarls are in blue and superbubbles in red.
Cells containing ``TO'' indicate a time out after three hours. Columns $n$ and $m$ denote the number of nodes and edges in millions, respectively, and $t$ indicates the number of threads used.
We bold the fastest runtimes.
}
\label{tab:wall_clock_par}
\end{table}

\textbf{Results on snarls.} 
\Cref{tab:wall_clock_seq} shows the execution times for computing snarls and superbubbles.
To simulate building the BC~tree and the SPQR~trees in a separate preprocessing step, we subtract the time it takes to build the trees from the total running time of our tool.
This simulated timing is given in the columns labeled ``I/O+ALGO'' and the total timing is given in the columns labeled ``all''.
When analyzing the PGGB graphs \textsf{50x~E.~coli}, \textsf{Primate~Chr.~6}, \textsf{Tomato~Chr.~2} and \textsf{Mouse~Chr.~19}, we observe that BubbleFinder is up to 2 times faster than \texttt{vg} on the larger graphs, 3 times faster on the smallest graph and in general always faster.
When analysing the graphs for \textsf{Chromosome~1|10|22}, we observe that BubbleFinder is up to two times slower than \texttt{vg}.
When analysing the de Bruijn graph \textsf{10x~M.~xanthus}, we observe that BubbleFinder is about as fast as \texttt{vg}.
All in all, taking the block-size histograms in \Cref{fig:block_dist} in the Appendix into account, we see that BubbleFinder is faster on datasets with large blocks.
Recall that our algorithm actually computes all snarls, whereas \texttt{vg} computes only an amount linear in the size of the graph.
Therefore, we conclude that our algorithm outcompetes \texttt{vg} in terms of runtime.

Note that even if we compare the ``all'' running times against the other tools, our implementation is still competitive for three of the PGGB datasets, indicating that even with using an SPQR-tree builder that is not meant for biological data, our approach is the new state-of-the-art in some biological applications.

\Cref{tab:wall_clock_par} shows the execution times with different amounts of threads for the \texttt{vg} datasets.
We see that BubbleFinder achieves a speedup of up to $4\times$ when run with 16 threads, obtaining similar running times as \texttt{vg}. This is thanks to the flat distribution of block sizes in these datasets.
As visualized in \Cref{fig:block_dist} in the Appendix, no block has more than 100 edges.
Other datasets contain single large blocks with millions of edges, which counter our block-level parallelization strategy.
Note that \texttt{vg snarls} also supports multiple threads in its CLI, but its publication does not mention any parallelisation, and running it with multiple threads did not change the runtime. 
Hence we ran it only single-threaded.

\Cref{tab:max_rss} in the Appendix shows the RAM usage.
The RAM usage of BubbleFinder is about 4-10 times higher than that of \texttt{vg}, if we do not exclude the tree building.
But if we focus only on I/O and our algorithms, our algorithm uses between half and twice as much memory than \texttt{vg}.
Hence, the highest memory consumption occurs during building the trees, which suggests that with an efficient bioinformatics-ready SPQR tree builder, our algorithms will easily be competitive in terms of memory.

\textbf{Results on superbubbles.}
For superbubbles we obtain similar results. The runs on superbubbles are roughly 50\% slower than the corresponding ones on snarls, which can be explained with two reasons: our algorithm for superbubbles has to keep track of more complex properties such as reachability and acyclicity (more complex algorithm); and while we output a linear-size representation of all snarls, we list all superbubbles in a plain format (larger output).
We note that BubbleFinder runs in similar times as BubbleGun on small graphs, but BubbleFinder is about 10 times faster on larger graphs.

\section{Conclusions}

From a theoretical point of view, our new SPQR tree framework has great potential for being extended to other bubble-like structures, such as ultrabubbles and bibubbles. However, at the moment there exists a lack of basic graph algorithms for bidirected graphs, which would be needed for such extensions. For example, for ultrabubbles we would need a linear-time algorithm for feedback edges in bidirected graphs and for bibubbles, we would need an efficient algorithm for (source/sink) reachability under edge removals. Both of these are interesting challenges to explore in the future. From a practical point of view, we believe that the main bottleneck of our implementation is OGDF. By further investigating the library, we note that while the SPQR tree is linear in size (in the input graph) the overhead of metadata we obtain from OGDF is prohibiting their usage on a large pangenomic level. Also, the construction algorithms from OGDF perform recursive searches, which also restricts their use in large pangenomic levels (due to memory consumption and stack size limitation). Both of these motivate the (re-)implementation of a pangenome-level BC tree and SPQR tree construction (as well as the algorithms on top of them). Furthermore, the current bottleneck of parallelization is that both the OGDF library as well as our implementation do not support parallelization on a deeper level than that of a block, which motivates further methods of parallelization within blocks.

\bibliographystyle{plain}
\bibliography{bibliography}

\newpage
\appendix

\section*{\centering Appendix of\\``Identifying all snarls and superbubbles in linear-time, \\via a unified SPQR-tree framework''}

\section{Defining snarls in bidirected graphs}
\label{sec:biedged-snarls}

\begin{figure}[h]
   \centering
   (a)\includegraphics{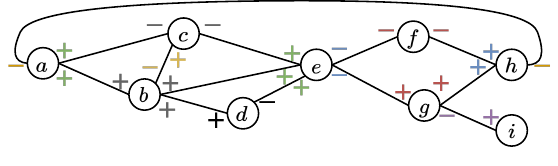}

   (b)\includegraphics{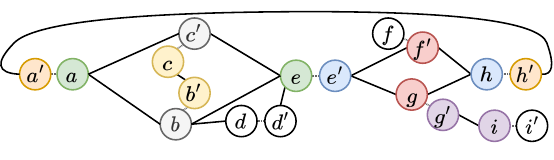}
   \caption{
   (a) A bidirected graph.
   The snarls are $\{a+, e+\}$ (green), $\{c+, b-\}$ (yellow), $\{c-, b+\}$ (gray), $\{e-, h+\}$ (blue), $\{a-, h-\}$ (orange), $\{f-, g+\}$ (red), $\{g-, i+\}$ (purple).
   Snarl $\{c+, b-\}$ is contained in snarl $\{a+, e+\}$.
   (b) The corresponding biedged graph.
   The outer edges are black, and the inner edges are dashed in gray.
   The snarls are highlighted in the same colors.
   }
   \label{fig:biedged}
\end{figure}

Paten et al.~\cite{paten2018ultrabubbles} define snarls via the \emph{biedged graph}, which is an undirected graph $B(G) = (B(V), B(E))$ created from a bidirected graph $G = (V, E)$ as follows.
We perform the \emph{splitting} operation on each node, and the replace each edge with an undirected $\{v_1 d_1, v_2 d_2\}$ with an undirected \emph{outer} edge $\{a, b\}$.
Then, we insert an \emph{inner} edge $\{v, v'\}$ for each pair of nodes where $v'$ was produced by splitting $v$.
We call $v$ the \emph{opposite} of $v'$ and vice-versa.
If an inner edge has a parallel outer edge, we consider them separate edges.
Otherwise, we assume that there are no parallel edges, as they have no effects on snarls.
See \Cref{fig:biedged} for an example of a biedged graph and its corresponding bidirected graph.
In the biedged graph, Paten~et~al. define snarls as follows.

\begin{definition}[Snarl in biedged graphs]
    \label{def:snarl-biedged}
    A pair set of distinct, non-opposite nodes $\{x, y\}$ is a \emph{snarl} if
    \begin{enumerate}
        \item[(a)] \emph{separable:} the removal of the inner edges incident with $x$ and $y$ disconnects the graph, creating a separate component $X$ that contains $x$ and $y$ and not their opposites.
        We call $X$ the \emph{snarl component} of $\{x, y\}$.
        \item[(b)] \emph{minimal:} no pair of opposites $\{z, \hat{z}\}$ in $X$ different from $x$ and $y$ exists such that $\{x, z\}$ and $\{y, \hat{z}\}$ are separable.
    \end{enumerate}
\end{definition}

To avoid using the biedged graph in our algorithm, we propose an equivalent definition of snarls in bidirected graphs.

\snarlandcomponent*

See \Cref{fig:biedged} for an example of snarls in a bidirected graph and the corresponding snarls in the corresponding biedged graph.
We show that the two definitions of snarls are equivalent.

\begin{lemma}[Equivalence of snarl definitions]
    \label{lem:snarl-equivalence}
    Let $\{xd_x, yd_y\}$ be a pair set of incidences in a bidirected graph $G$.
    Let $\{a, b\}$ be the corresponding pair set in the corresponding biedged graph $B(G)$.
    Then $\{xd_x, yd_y\}$ is a snarl in $G$ by \Cref{def:snarl} if and only if $\{a, b\}$ is a snarl in $B(G)$ by \Cref{def:snarl-biedged}.
\end{lemma}
\begin{proof}
    $(\Rightarrow)$
    Let $\{xd_x, yd_y\}$ be a snarl in $G$.
    Let $\{a, b\}$ be the corresponding nodes in $B(G)$.
    We show that $\{a, b\}$ fulfills each condition of \Cref{def:snarl-biedged}.
    \begin{itemize}
        \item
        By \Cref{def:snarl}, it holds that $x \neq y$ and hence $a$ and $b$ are distinct and non-opposite.
        
        \item \emph{separable}:
        By \Cref{def:snarl}, it holds that when splitting the incidences $xd_x$ and $yd_y$ in $G$, then the graph contains a separate component $X$ containing $x$ and $y$ but not the nodes $x'$ and $y'$ created by the node separation operation.
        Hence, when removing the inner edges incident with $a$ and $b$ in $B(G)$, we create a separate component that contains $a$ and $b$, but not their opposites.
        Therefore, $\{a, b\}$ is separable by \Cref{def:snarl-biedged}.
        
        \item \emph{minimal}:
        By \Cref{def:snarl}, it holds that no incidence $zd_z$ with $z \in X$ different from $x$ and $y$ exists such that $\{xd_x, zd_z\}$ and $\{yd_y, z\hat{d_z}\}$ are separable.
        Since we have shown the equivalence of separability above, this implies that in $B(G)$, no node $c \in X$ different from $a$ and $b$ exists such that $\{a, c\}$ and $\{b, \hat{c}\}$ are separable.
        Hence, $\{a, b\}$ is minimal by \Cref{def:snarl-biedged}.
    \end{itemize}
    
    Therefore, if $\{xd_x, yd_y\}$ is a snarl in $G$, it holds that $\{a, b\}$ is a snarl in $B(G)$.

    $(\Leftarrow)$
    Let $\{a, b\}$ be a snarl in $B(G)$.
    Let $\{xd_x, yd_y\}$ be the corresponding incidences in $G$.
    We show that $\{xd_x, yd_y\}$ fulfills each condition of \Cref{def:snarl}.
    \begin{itemize}
        \item
        By \Cref{def:snarl-biedged}, it holds that $a$ and $b$ are distinct and non-opposite, and hence $x \neq y$.
        
        \item \emph{separable}:
        By \Cref{def:snarl-biedged}, it holds that the removal of the inner edges $\{a, \hat{a}\}$ and $\{b, \hat{b}\}$ disconnects $B(G)$, creating a separate component $X$ that contains $a$ and $b$ and neither $\hat{a}$ nor $\hat{b}$.
        Hence, when splitting the incidences $xd_x$ and $yd_y$ in $G$, we create a separate component $X$ that contains $x$ and $y$ but not the nodes $x'$ and $y'$ created by the node separation operation.
        Therefore, $\{xd_x, yd_y\}$ is separable by \Cref{def:snarl}.
        
        \item \emph{minimal}:
        By \Cref{def:snarl-biedged}, it holds that no node $c \in X$ different from $a$ and $b$ exists such that $\{a, c\}$ and $\{b, \hat{c}\}$ are separable.
        Since we have shown the equivalence of separability above, this implies that in $G$, no incidence $zd_z$ with $z \in X$ different from $x$ and $y$ exists such that $\{xd_x, zd_z\}$ and $\{yd_y, z\hat{d_z}\}$ are separable.
        Hence, $\{xd_x, yd_y\}$ is minimal by \Cref{def:snarl}.
    \end{itemize}
    
    Therefore, if $\{a, b\}$ is a snarl in $B(G)$, it holds that $\{xd_x, yd_y\}$ is a snarl in $G$.
\end{proof}

\newpage
\section{Missing results on snarls}
\label{sec:missing-snarls}

In this section we give a detailed explanation of our results on snarls. Recall that we assume implicitly that our bidirected graphs are connected, i.e., they consist of a single component.

\subsection{Properties of snarls and connectivity}

Let $v$ be a vertex of a bidirected graph $G$. We let $\mathcal{C}^+_G(v)$ ($\mathcal{C}^-_G(v)$) denote the set of components of $G-v$ containing only positive (negative) incidence in $v$, and $\mathcal{C}^{\pm}_G(v)$ is the set components of $G-v$ containing at least one positive and at least one negative incidence in $v$. For instance, if $v$ is sign-consistent then $\mathcal{C}^\pm_G(v)=0$, if $\mathcal{C}^\pm_G(v)=0$ then $v$ has no dangling blocks from anywhere, and if $\mathcal{C}^\pm_G(v)=1$ then $v$ may or may not have dangling blocks depending on the block in question.

The next result is the motivation behind sign-cut graphs. 
Essentially, we show that a sign-consistent vertex $x$ violates minimality between any two incidences whose vertices are separated by $x$.

\begin{prop}
\label{prop:snarls-good-cutvertex}
    Let $G$ be a bidirected graph and let $x$ be a cutvertex of $G$ such that $\mathcal{C}^\pm_G(x)$ is empty. If $u$ is in a component of $\mathcal{C}^+_G(x)$ and $v$ is in a component of $\mathcal{C}^-_G(x)$ then $\{ud_u,vd_v\}$ is not a snarl for any incidences $ud_u$ and $vd_v$ in $G$.
\end{prop}
\begin{proof}
    We can assume that there are incidences at $x$ of opposite signs, since otherwise the statement holds vacuously.
    Suppose for a contradiction that there are signs $d_u,d_v \in \signs$ in $u$ and $v$ such that $\{ud_u,vd_v\}$ is a snarl with a snarl component $X$. Since every $u$-$v$ path in $G$ contains $x$, every $u$-$v$ path in the graph after splitting $ud_u$ and $vd_v$ also contains $x$, and thus $x \in V(X)$. We argue that $\{ud_u,x+\}$ is separable in $G$ with a snarl component $X^+$. Analogously, we can then show that $\{x-, vd_v\}$ is separable, consequently violating the minimality property of the supposed snarl~$\{ud_u,vd_v\}$.
    
    Since $\mathcal{C}^\pm_G(x)$ is empty, splitting $x+$ creates two components $C^+$ and $C^-$, with the former containing all the positive incidences $x+$ of $G$ and the latter containing all the negative incidences $x-$ of $G$ reattached to $x'$.
    Clearly, $x'$ is not in $V(X^+)$. It remains to show that $u' \notin V(X^+)$.
    
    We observe that $V(X^+) \subseteq V(X)$. To see why, pick a vertex $y \in V(X^+)$. There there is an $x$-$y$ path in the graph obtained by splitting $x+$ and $ud_u$ in $G$ such that the first edge of the path corresponds to a bidirected edge with incidence $x+$. This path also exists in the graph where $x+$ is not split. Splitting $vd_v$ does not affect the existence of this path either, since $x$ is a $u$-$v$ cutvertex of $G$.
    
    Since $x \in X$ there is a $v$-$x$ path in the graph obtained by splitting $ud_u$ and $vd_v$ that enters $x$ through an incidence $x-$. We can thus simply merge the $v$-$x$ path and the $x$-$y$ path to get a $v$-$y$ path in the graph obtained by splitting $ud_u$ and $vd_v$, proving that $y \in V(X)$.
    Since $V(X^+) \subseteq V(X)$, the claim follows from the fact that $u' \notin V(X)$ because $\{ud_u,vd_v\}$ is separable in $G$. Therefore $\{ud_u,x+\}$ is separable, and analogously $\{x-, vd_v\}$ is separable, contradicting the fact that $\{ud_u,vd_v\}$ is a snarl.
\end{proof}

The proof of the next result is essentially a direct application of \Cref{prop:snarls-good-cutvertex} but it also requires handling of a case that arises from the fact that sign-consistent vertices appear in two sign-cut graphs.

\snarlsinsidesigncuttrees*
\begin{proof}
    We can assume that $u \neq v$ for otherwise $\{ud_u,vd_v\}$ is not a snarl by definition.
    
    There is a sign-consistent vertex $x \in V(G)$ that puts $u$ and $v$ in distinct sign-cut graphs of $G$ (possibly $x=u$ or $x=v$).
    If $x$ can be chosen so that it is distinct from $u$ and $v$ then we can apply \Cref{prop:snarls-good-cutvertex} to $G,x,u,v$ and conclude that $\{ud_u,vd_v\}$ is not a snarl for any $d_u,d_v \in \signs$.
    Otherwise $x$ is equal to $u$ or $v$. Suppose without loss of generality that $x=u$ and suppose for a contradiction that $\{ud_u,vd_v\}$ is a snarl with component $X$.
    Since $u=x$ is the sign-consistent vertex that puts $u$ and $v$ in different sign-cut graphs by necessity, the incidences in $F_1$ in $u$ all have sign $d_u$. Moreover, $F_2$ is the other sign-cut graph containing $u$ whose incidences have sign $\hat{d}_u$.
    Thus every path from $u$ to $v$ in $G$ uses an edge containing an incidence $u\hat{d}_u$ and hence $v\notin V(X)$, or $X$ contains an incidence in $u\hat{d}_u$ and hence $X$ contains an incidence in $u'$. Both cases violate the separability of $\{ud_u,vd_v\}$ and thus $\{ud_u,vd_v\}$ is not a snarl.
\end{proof}

We can now prove equivalence of snarls between $G$ and the sign-cut graphs of $G$.

\snarlsGF*
\begin{proof}
    We start by proving that separability is retained in both directions. Consequently, we show that a minimal snarl in $F$ satisfies minimality in $G$ and vice versa. 

    ($\Leftarrow$, \emph{separability}) Let $\{ud_u,vd_v\}$ be a snarl of $F$. We show that $\{ud_u,vd_v\}$ is a snarl of $G$. 
    Since it is a snarl of $F$, $\{ud_u, vd_v\}$ is separable and there is a $u$-$v$ path in $F$. Thus the path is also in its supergraph $G$. 
    It remains to show that $u'$ and $v'$ are separated from $u$ and $v$ in $G$ when we split $ud_u$ and $vd_v$. Assume contrarily that there is a $u$-$u'$ path in the resulting graph. The proof is analogous for the $v$-$v'$ path. If the path contains only vertices from $F$, then the path would also be in $F$ and $\{ud_u,vd_v\}$ would not be separable in $F$.
    On the other hand, if there is another sign-cut graph $F'$ such that the path also contains a vertex $w \in V(F')$, then the subpaths $u$-$w$ and $w$-$v$ pass through some cutvertex twice. This also results in a contradiction, since paths do not repeat vertices.
    
    ($\Rightarrow$, \emph{separability}) Let $\{ud_u,vd_v\}$ be a snarl of $G$. By \Cref{lem:snarls-inside-sign-cut-tree} the incidences $ud_u$ and $vd_v$ are in the same sign-cut graph $F$ of $G$. We show that $\{ud_u,vd_v\}$ is a snarl of $F$. 
    
    Since it is a snarl of $G$, there is a $u$-$v$ path in $G$. 
    Note that the path only contains vertices from $V(F)$ by assuming the contrary, i.e., there is another sign-cut graph $F'$ such that the path also contains a vertex $w \in V(F')$. This is impossible, since the $u$-$v$ path does not repeat vertices, but the $u$-$w$ subpath and the $w$-$v$ subpath pass through some cutvertex twice. 
    Since $F$ is a subgraph of $G$, separability then holds trivially.

    ($\Leftarrow$, \emph{minimality}) Let $\{ud_u,vd_v\}$ be a snarl of $F$. We prove that minimality of $\{ud_u, vd_v\}$ holds in $G$. Assume the contrary, that is, there exists $x \in V(G)$ and a sign $d_x$ such that $\{ud_u, xd_x\}$ and $\{x\hat{d}_x, vd_v\}$ are separable in $G$. By the separability result for the other direction, $\{ud_u, xd_x\}$ and $\{x\hat{d}_x, vd_v\}$ would then be separable in $F$, implying that $\{ud_u,vd_v\}$ is not a snarl of $F$.

    ($\Rightarrow$, \emph{minimality}) Let $\{ud_u,vd_v\}$ be a snarl of $G$. For minimality, suppose that there exists a vertex $w$ and sign $d_x$ such that $\{ud_u, xd_x\}$ and $\{x\hat{d}_x, vd_v\}$ are separable in $F$. Then, it follows from the other direction of the proof that they are separable in $G$, and so $\{ud_u,vd_v\}$ is not a snarl of $G$ by a violation of minimality.
\end{proof}

Now we give a series of results that support our SPQR tree approach in finding snarls. Essentially, these results relate the structure of snarls with connectivity properties of the graph.

\begin{prop}
\label{prop:tips-dont-violate-minimality}
    Let $G$ be a bidirected graph, let $ud_u$ be an incidence of $G$, and let $z \in V(G)$ be a vertex distinct from $u$. If $z$ is a tip in $G$ with sign $d_z$ then $\{ud_u,z\hat{d}_z\}$ is not separable.
\end{prop}
\begin{proof}
    Suppose for a contradiction that $\{ud_u,z\hat{d}_z\}$ is separable with component $X$.
    Splitting $z\hat{d}_z$ results in $z$ being isolated because $z$ is a tip in $G$ with sign $d_z$ (and $z'$ retains the original incidences of $G$ in $z$). Splitting $ud_u$ does not create new paths in the resulting graph, and hence there is no $u$-$z$ path in the graph resulting from these two splits. Therefore $X$ does not contain $z$, contradicting the fact that $\{ud_u,z\hat{d}_z\}$ is separable.
\end{proof}

\begin{prop}
\label{prop:mixed-signs-inside-block}
    Let $G$ be a bidirected graph, let $F$ be a sign-cut graph of $G$, and let $u \in V(F)$ be a vertex. If $u$ is a non-tip in $F$ then there is a block of $F$ with incidences of opposite signs in $u$.
\end{prop}
\begin{proof}
    Suppose for a contradiction that the incidences at $u$ in every block of $F$ all have the same sign.
    If $u$ is not a cutvertex in $F$ then $u$ is in a unique block of $F$, thus $u$ is a tip, a contradiction.
    Otherwise there is a block intersecting $u$ containing only incidences of sign $q \in\signs$ and a block intersecting $u$ containing only incidences of sign $\hat{q}$, since $u$ is a non-tip by assumption. Thus $u$ is sign-consistent with respect to $G$, a contradiction to the fact that $F$ is a sign-cut graph of $G$.
\end{proof}

\begin{prop}
\label{prop:snarls-both-tips-component}
    Let $G$ be a bidirected graph and let $\{ud_u,vd_v\}$ be a snarl of $G$ with component $X$. Then $X=G$ if and only if $u$ and $v$ are tips in $G$ with signs $d_u,d_v$, respectively.
\end{prop}
\begin{proof}
    $(\Leftarrow)$
    If $u$ and $v$ are tips in $G$ with signs $d_u$ and $d_v$ then splitting $ud_u$ and $vd_v$ results in a graph with components $G$ and two isolated vertices, $u'$ and $v'$. Since $G$ is connected\footnote{Recall that we implicitly assume that every graph is connected.} it follows that $X=G$.

    $(\Rightarrow)$
    We show that if $u$ or $v$ are non-tips in $G$ then $X \neq G$ (i.e., there is an edge or a vertex of $G$ that is not part of $X$, as $X \subseteq G$ by definition of snarl and snarl component).
    Suppose without loss of generality that $u$ is a non-tip in $G$ and that $d_u=+$.
    Since $u$ is a non-tip in $G$ by assumption, $G$ has an edge $\{u-,wd_w\}$ for some incidence $wd_w$. Let $G'$ denote the graph resulting from splitting $u+$ and $vd_v$.
    
    Suppose that $w \neq v$. If $w \in V(X)$ then there is a $u$-$w$ path in $X$ because components are connected. This path can be extended with the edge $\{wd_w,u'-\} \in E(G')$, and so there is a $u$-$u'$ path in $G'$ and thus $u' \in V(X)$, contradicting the fact that $\{ud_u,vd_v\}$ is separable. Thus $w \notin V(X)$ and so $V(X) \neq V(F)$.

    Suppose that $w = v$. Then $d_w=\hat{d}_v$, otherwise $\{ud_u,vd_v\}$ is not separable since $u$ and $u'$ would be connected in $G'$ via $v$. Moreover, $\{u'-,v'\hat{d}_v\}=\{u'-,w'd_w\}$ is not an edge of $X$ since otherwise $u',v' \in V(X)$, contradicting the separability of $\{ud_u,vd_v\}$. Thus $\{u-,wd_w\} \notin E(X)$, and so $E(X) \neq E(G)$.
\end{proof}

\begin{prop}[Dangling blocks]
\label{prop:mixed-dangling-blocks}
    Let $G$ be a bidirected graph, let $H$ be a block of $G$, and let $u,v \in V(H)$ be vertices. If $u$ or $v$ has dangling blocks with respect to $H$ then $\{ud_u,vd_v\}$ is not separable for any $d_u,d_v \in \signs$.
\end{prop}
\begin{proof}
    \sloppypar
    Without loss of generality suppose that $u$ has a dangling block $H'$ with respect to $H$. So $H'$ has edges $\{u+,xd_x\}$ and $\{u-,yd_y\}$. Notice that $u$ is a cutvertex of $G$, since otherwise there is at most one block containing $u$.
    Splitting $ud_u$ leaves $u$ and $u'$ connected by taking the edge $\{u+,xd_x\}$ followed by a $x$-$y$ path (which exists since $H'$ is block and thus is connected) followed by the edge $\{u-,yd_y\}$. Since no two blocks can contain the same two vertices, $v$ is not in $H'$. Splitting $vd_v$ does no affect the path previously constructed, and so $u$ and $u'$ remain connected in the graph resulting from splitting $ud_u$ and $vd_v$, and therefore $\{ud_u,vd_v\}$ is not separable.
\end{proof}

\begin{prop}
\label{prop:mixed-components-for-snarls}
    Let $G$ be a bidirected graph and let $ud_u,vd_v$ be incidences of $G$. If $\{ud_u,vd_v\}$ is a snarl and $|\mathcal{C}^\pm_G(u)| = 1$ or $|\mathcal{C}^\pm_G(v)| = 1$, then there is a $u$-$v$ path that starts with a $ud_u$ incidence and ends with a $vd_v$ incidence and a $u$-$v$ path that starts with a $u\hat{d}_u$ incidence and ends with a $v\hat{d}_v$ incidence.
\end{prop}
\begin{proof}
    Without loss of generality suppose that $|\mathcal{C}^\pm_G(u)|=1$. Let $C_u \in \mathcal{C}^\pm_G(u)$. Then $C_u$ has edges $\{u+, ad_a\}$ and $\{u-, bd_b\}$ for some incidences $ad_a,bd_b$ in $C_u$. 
    Notice that every $a$-$b$ path in $G-u$ contains $v$ since $a$ and $b$ remain in the same component after removing $u$ from the graph and $u$ and $u'$ are separated after splitting $ud_u$ and $vd_v$.
    
    However, this implies that $v$ is reachable from $u$ in $G$ by both a path $p_1$ that starts with a $u+$ incidence and a path $p_2$ that starts with a $u-$ incidence. Further, since $\{ud_u,vd_v\}$ is a snarl, by the separability property we have that $p_1$ and $p_2$ are vertex-disjoint (except for $u$ and $v$), the path $p_1$ ends with an incidence $vd_v$ if $d_u = +$ and $v\hat{d}_v$ otherwise, and the path  $p_2$ ends with an incidence $vd_v$ if $d_u = -$ and $v\hat{d}_v$ otherwise. The claim follows.
\end{proof}

With these results and the equivalence given by \Cref{lem:snarls-G-F}, we can now give a precise proof of \Cref{thm:where-are-snarls-after-cutting}.

\snarlsaftercutting*
\begin{proof}
    We prove each item separately with respect to $F$ (which is correct due to \Cref{lem:snarls-G-F}).
    \begin{enumerate}
        \item Since $u$ is a non-tip in $F$, by \Cref{prop:mixed-signs-inside-block}, $F$ has a block $H$ containing edges $\{u+, vd_v\}$ and $\{u-, wd_w\}$.
        We show that $u$ and $u'$ remain connected after splitting $ud_u$ and $xd_x$ in $F$, which violates the separability of $\{ud_u,xd_x\}$ and implies that $\{ud_u,xd_x\}$ is not a snarl.

        Suppose without loss of generality that $d_x=+$.
        Splitting $x+$ in $F$ results in a graph $G'$ where $x$ retains all the (positive) incidences of $F$ in $x$ and $x'$ is an isolated vertex since $x$ is a source in $F$.
        Therefore splitting $ud_u$ in $G'$ results in a graph where $u$ and $u'$ are connected since $u$ can be left with the edge $\{u+,vd_v\}$, $u'$ can be entered through the edge $\{wd_w,u'-\}$, and in between these edges we can insert a $v$-$w$ path using vertices of $G$ that avoids $u$ since $H$ is a block (if $H$ is a multi-bridge then $v=w$ and the path is trivial).
    
        \item Let $u,v$ be distinct tips in $F$ with signs $d_u,d_v \in \signs$. We show that $\{ud_u,vd_v\}$ is a snarl of $F$ by showing separability and minimality.

        Let $F'$ denote the graph resulting from splitting $ud_u$ and $vd_v$.
        Since $u$ and $v$ are tips in $F$ with signs $d_u$ and $d_v$, respectively, $F'$ consists of three components, which are $F$ and the two isolated vertices $u'$ and $v'$. Therefore $u$ and $v$ are in the same component because $F$ is connected, and $u'$ and $v'$ are not in $F$. Hence, $\{ud_u,vd_v\}$ is separable.

        To see minimality, suppose for a contradiction that there is an incidence $zd_z$ in $F$ with $z$ distinct from $u$ and $v$ such that $\{ud_u,zd_z\}$ and $\{vd_v,z\hat{d}_z\}$ are separable. We conclude that $z$ is not a tip in $F$ by applying \Cref{prop:tips-dont-violate-minimality} to $vd_v$ and $z$, since otherwise $\{vd_v,z\hat{d}_z\}$ is not separable. But since $u$ and $v$ are tips by assumption and we have already proven in (1) that a tip and a non-tip are not separable, the claim then follows.
        

        \item Let $\{ud_u,vd_v\}$ be a snarl of $F$. We first show that there is a block $H$ of $F$ containing both $u$ and $v$. Suppose for a contradiction no block of $F$ contains both $u$ and $v$. Because $u$ is a non-tip in $F$ we can apply \Cref{prop:mixed-signs-inside-block} to conclude that there are edges $\{u+,vd_v\}$ and $\{u-,wd_w\}$ within a block of $F$. Therefore there is a $v$-$w$ path $p$ in $F$ avoiding $u$ (since a block is 2-connected, or a multi-bridge in which case $v=w$ and $p$ is trivial).
        Thus, splitting $ud_u$ leaves $v$ and $w$ connected via $p$, and hence $u$ and $u'$ remain connected (similarly to the path constructed in the proof of item 1). Since $v$ is in a different block than $u$ by assumption, vertex $v$ is not in $p$. Therefore splitting $vd_v$ does not separate $u$ from $u'$, contradicting the separability of $\{ud_u,vd_v\}$.

        To conclude that no other block of $F$ containing $u$ has incidences of opposite signs in $u$, and the same for $v$, (i.e., $u$ and $v$ have no dangling blocks with respect to $H$), apply \Cref{prop:mixed-dangling-blocks}.
    \end{enumerate}
\end{proof}

\begin{prop}
\label{prop:snarl-edge-case}
    Let $G$ be a bidirected graph, let $F$ be a sign-cut graph of $G$, let $H$ be a block of $F$ containing the vertices $u$ and $v$, and let $\{ud_u,vd_v\}$ be a snarl of $F$. If $u$ and $v$ are non-tips in $H$, $\{u,v\}$ is an edge of $H$, and $\{u,v\}$ is not a separation pair of $H$, then $e=\{ud_u,vd_v\} \in E(H)$ and all incidences at $u$ and $v$ except for those in $e$ have signs $\hat{d_u},\hat{d}_v$, respectively, or $e=\{u\hat{d}_u,v\hat{d}_v\} \in E(H)$ and all incidences at $u$ and $v$ except for those in $e$ have signs $d_u,d_v$, respectively.
\end{prop}
\begin{proof}
    First notice that $\{ud_u, v\hat{d}_v\}, \{u\hat{d}_u, vd_v\} \not\in E(H)$, since otherwise there is a path from $u$ to $v'$ or from $u'$ to $v$ after splitting $ud_u$ and $vd_v$. We argue on the two cases described in the statement.
    
    Suppose that $e=\{ud_u,vd_v\} \in E(H)$ and pick an edge $\{u\hat{d}_u, ad_a\}$ of $H$ (which exists since $u$ is a non-tip in $H$).
    By \Cref{prop:mixed-dangling-blocks} it follows that $u$ and $v$ have no dangling blocks with respect to $H$, and since $u$ and $v$ are both non-tips in $H$ it follows that $H$ is the only block of $F$ containing incidences of opposite signs at $u$ and $v$. Suppose for a contradiction that $H$ has an edge $\{ud_u, bd_b\} \neq e$. Since $\{u,v\}$ is not a separation pair of $H$, there exists a path from $a$ to $b$ in $H$ that does not contain $u$ or $v$. This path remains if we split $ud_u$ and $vd_v$, thus violating separability between $u$ and $u'$, a contradiction; symmetrically, the same occurs if $H$ has an edge $\{v\hat{d}_v, cd_c\}$.
    The case when $e=\{u\hat{d}_u,v\hat{d}_v\} \in E(H)$ follows with an identical argument.
\end{proof}

The next result is a key ingredient for our algorithm.

\snarlssplitpairs*
\begin{proof}
    
    If $\{u,v\}$ is an edge of $U(G)$ we are done, so assume that $\{u,v\} \notin E(U(G))$.
    
    Since $\{ud_u,vd_v\}$ is a snarl and $u$ is a non-tip, \Cref{thm:where-are-snarls-after-cutting} implies that $u$ and $v$ are both non-tips belonging to the same block $H$ of a sign-cut graph $F$ of $G$, and thus of $G$ too since the blocks of $G$ coincide with the blocks of its sign-cut graphs. Moreover, $H$ is the only block containing incidences of opposite signs at $u$ and at $v$ by (3) of \Cref{thm:where-are-snarls-after-cutting}, which implies that $|\mathcal{C}^\pm_F(u)| = |\mathcal{C}^\pm_F(v)| = 1$.
    Thus, \Cref{prop:mixed-components-for-snarls} implies that there are vertex-disjoint paths $p_1$ and $p_2$ between $u$ and $v$ such that $p_1$ starts and ends with incidences $ud_u$ and $vd_v$, respectively, and $p_2$ starts and ends with incidences $u\hat{d}_u$ and $v\hat{d}_v$, respectively. Since $\{u,v\}$ is not an edge, both paths must contain an internal vertex distinct from $u$ and $v$ which moreover is contained in $H$. Denote these by $x$ and $y$ for $p_1$ and $p_2$, respectively. By the definition of snarl component, we have $x \in V(X)$ and $y \in V(H) \setminus V(X)$. Separability further implies that all $x$-$y$ paths go through $u$ or $v$, since otherwise there would be a path from $u$ to $u'$ in the graph obtained by splitting the incidences $ud_u$ and $vd_v$. In other words, $\{u,v\}$ is a separation pair of $H$.
\end{proof}

The case when $\{u,v\}$ is an edge of $U(G)$ is slightly overlooked in the proof of \Cref{lem:snarls-G-F}. However, the goal of this result is to conclude that the vertices of a snarl may form a separation pair. The edge-case (which is essentially described in \Cref{prop:snarl-edge-case}) is easy to handle algorithmically, although its correctness requires some effort in proving.

\subsection{Finding all snarls}
\label{sec:finding-snarls-inside-blocks}

In this section we develop an algorithm to find snarls whose vertices form a separation pair of a block of $G$, as described in \Cref{thm:snarls-split-pairs}. Since snarls are only defined by their separability and minimality, it is not required to maintain any partial information during the algorithm.
We start by giving a useful result for showing minimality of separable pairs of incidences. Then we show results giving (mostly) sufficient conditions for a snarl to exist within the different nodes of the SPQR tree. Finally we give the correctness and running-time proof of our algorithm.

\begin{prop}\label{prop:snarls-disjoint-paths}
    Let $G$ be a bidirected graph and let $ud_u,vd_v$ be incidences of $G$ so that $\{ud_u,vd_v\}$ is separable with component $X$. If $X$ has two internally vertex-disjoint $u$-$v$ paths then $\{ud_u,vd_v\}$ is a snarl.
\end{prop}
\begin{proof}
    Let $p_1$ and $p_2$ be two internally-vertex disjoint $u$-$v$ paths in $X$. Because they are in $X$, they start with a $ud_u$ incidence and end with a $vd_v$ incidence.
    
    Suppose for a contradiction that there is a vertex $w \in X \setminus \{u, v\}$ such that $\{ud_u, wd_w\}$ and $\{vd_v, w\hat{d}_w\}$ are separable with $d_w\in\signs$. Assume without loss of generality that $w$ is not in $p_2$. The two incidences $wd_1$ and $wd_2$ in $p_1$ must have identical signs $d_1, d_2 \in \signs$, because otherwise there would be a path from $w$ to $w'$ after splitting $wd_1$ and either $ud_u$ or $vd_v$ by splitting $p_1$ on $w$ into a $u-w$ path and a $w'$-$v$ path and merging the endpoints of the paths $u$ and $v$ with the $p_2$ path.
    
    For the remainder of the proof, we can thus assume that no $u$-$v$ path containing $w$ has both incidences $w+$ and $w-$. It cannot be the case that there are paths with two $w+$ incidences and paths with two $w-$ incidences, because then there is also a path with a single $w+$ and a single $w-$ incidences. Without loss of generality, suppose then both incidences on any $u$-$v$ path containing $w$ are $w+$. This also implies all $w$-$u$ paths and $w$-$v$ paths start with a $w+$ incidence. Consequently, neither $\{w-, ud_u\}$ nor $\{w-, vd_v\}$ is separable, contradicting the earlier assumptions.
\end{proof}

The technique to find snarls within S-nodes is similar to the technique used to define sign-cut graphs. Let $\mu$ be an S-node with edges vertices $v_1,\dots,v_k$ and edges $e_1,\dots,e_k$ $(k \geq 3)$ such that each vertex $v_i$ is an endpoint of exactly two consecutive edges $e_i$ and $e_{(i+1) \mod k}$ of $\skel(\mu)$ for every $\iink$. Let us say that $v$ is \emph{good} if the incidences in $\expansion(e_i)$ into $v$ all have sign $d\in \{+,-\}$, the incidences in $\expansion(e_{(i+1) \mod k})$ into $v$ all have sign $\hat{d}$, and there are no dangling blocks with respect to $H$ at $u$ and $v$.

\begin{algorithm}[ht]
\caption{$\mathsf{FindSnarlsInSnodes}(F,H,T)$\label{alg:find-snarls-in-S}}
\KwIn{Sign-cut graph $F$ of a bidirected graph, maximal 2-connected subgraph $H$ of $G$, SPQR tree $T$ of $H$}
\For{each S-node $\mu$ of $T$}{
    $v_1,\dots,v_k \gets$ ordered sequence of the vertices of $\skel(\mu)$\;
    $e_1,\dots,e_k \gets$ ordered sequence of the edges of $\skel(\mu)$ such that $v_i$ is an endpoint of $e_{(i+1 \bmod k)}$ and $e_i$\;
    $W \gets [\;]$\; 
    \For{$i \in [1,k]$}{
        \If{$v_i$ is a tip in $F$ or $\mathsf{HasDangling}(F,H,v_i)$}{
            \textbf{continue}\;
        }
        $L, R \gets \expansion(e_i), \expansion(e_{(i+1 \bmod k)})$\;
        \If{$(N^+_H(v_i) \subseteq V(L)$ or $N^+_H(v_i) \subseteq V(R))$ and $(N^-_H(v_i) \subseteq V(L)$ or $N^-_H(v_i) \subseteq V(R))$}{
            \tcp{$v_i$ is good}
            $d \gets +$ if $(|N^+_H(v_i)|>0)$ else $-$\;
            $W.\mathsf{append}(v_i\hat{d})$\;
            $W.\mathsf{append}(v_id)$\;
        }
    }
    Report the pairs formed by the incidences in $W$ in consecutive positions starting from the second element, and lastly pair the last incidence with the first incidence of $W$\; \label{line:S-node-snarls}
}
\end{algorithm}

\snodesnarls*
\begin{proof}
    For conciseness, let $u = v_{i_j}$ and $v = v_{i_{(j + 1)\,\mathrm{mod}\,q}}$ for an arbitrary $j \in \{1, 2, \dots, q\}$.
    For separability, first note that after obtaining graph $G'$ by splitting $ud_u$ and $v\hat{d_v}$ in $G$, there remains a path from $u$ to $v$ through the edges of
    \[ E\left(\expansion(e_{i_j})\right) \cup E\left(\expansion(e_{(i_j+1)\,\mathrm{mod}\, k})\right) \cup \dots \cup E\left(\expansion\left(e_{\left(-1 + i_{(j+1)\,\mathrm{mod}\,q}\right)\,\mathrm{mod}\,k}\right)\right)\,. \]
    
    It remains to show that $u$ does not reach $u'$ and $v$ does not reach $v'$ in $G'$. All incidences of $u'$ are in $E\left(\expansion(e_{(i_j - 1)\,\mathrm{mod}\,k})\right)$ and all the incidences of $v'$ are in $E\left(\expansion(e_{i_{(j+1)\,\mathrm{mod}\,q}})\right)$. Further, there are no $u'd_u$ or $v'\hat{d}_v$ incidences because of splitting.
    Without loss of generality, assume for contradiction that there is a $u$-$u'$ path $p$ in $U(G')$. There then has to exist a last vertex $a$ on the path $p$ and its (not necessarily immediate) successor $b$ with the property that $a \in \{u', v'\}$ or
    \[ a \in V\left(\expansion(e_{i_{(j+1)\,\mathrm{mod}\,q}})\right) \cup V\left(\expansion(e_{(1 + i_{(j+1)\,\mathrm{mod}\,q})\,\mathrm{mod}\,k})\right) \cup \dots \cup V\left(\expansion\left(e_{\left(i_j - 1\right)\,\mathrm{mod}\,k}\right)\right) \]
    with $a \not\in \{u, v\}$, and
    \[ b \in V\left(\expansion(e_{i_j})\right) \cup  V\left(\expansion(e_{(i_j+1)\,\mathrm{mod}\,k})\right) \cup \dots \cup V\left(\expansion\left(e_{\left(-1 + i_{(j+1)\,\mathrm{mod}\,q}\right)\,\mathrm{mod}\,k}\right)\right)\,. \]
    By the definition of splitting, it cannot be the case that $a = v'$ and $b = v$, since there are no edges between $v$ and $v'$ and we have no mixed-sign dangling blocks.
    Similarly, it cannot be that $a = u'$ and $b = u$. On the other hand, we must have either $a = v'$ and $b = v$ or $a = u'$ and $b = u$, because we are operating within an S-node. By the contradiction, $u$ does not reach $u'$ and $v$ does not reach $v'$ in $G'$.
    
    For minimality, suppose for a contradiction that there is a vertex $w \in X \setminus \{u, v\}$ in the snarl component of $\{ud_u, v\hat{d}_v\}$ such that $\{ud_u, w\hat{d}_w\}$ and $\{wd_w, v\hat{d}_v\}$ are separable for some sign $d_w \in \signs$. Clearly, $w$ cannot have mixed dangling blocks. It also must be that $w$ is a vertex of the S-node, since a necessary condition for separability is that there cannot be a path that starts with $w+$ incidence and ends with $w-$ incidence and does not pass through $u$ or $v$. Let thus $w = v_l$. Since we assumed that $w$ is not a good vertex and there are no mixed-sign dangling blocks, either $\expansion(e_l)$ or $\expansion(e_{(l - 1)\,\mathrm{mod}\,k})$ must have both $w+$ and $w-$ incidences. Without loss of generality, assume this to be $e_{l}$. Then, the non-separability of $\{ud_u, w\hat{d}_w\}$ follows by there being a path from $w$ and $w'$ to $v$ if we split $ud_u$ and $w\hat{d}_w$.
\end{proof}

\begin{algorithm}[ht]
\caption{$\mathsf{FindSnarlsInPnodes}(F,H,T)$\label{alg:find-snarls-in-P}}
\KwIn{Sign-cut graph $F$ of a bidirected graph, maximal 2-connected subgraph $H$ of $G$, SPQR tree $T$ of $H$}
\For{each P-node $\mu$ of $T$}{
    $u,v \gets$ the vertices of $\skel(\mu)$\;
    $e_1,\dots,e_k \gets$ the edges of $\skel(\mu)$\;
    $X_1,\dots,X_k \gets \expansion(e_1),\dots,\expansion(e_k) \; (k\geq 3)$\;
    \If{$\mathsf{HasDangling}(F,H,u)$ or $\mathsf{HasDangling}(F,H,v)$ or $u$ is a tip in $F$ or $v$ is a tip in $F$}{
        \textbf{continue}\;
    }
    Build the sets $E^+_u, E^-_u, E^+_v, E^-_v$ as described in \Cref{prop:P-node-snarls}\;
    \tcp{Since $u$ and $v$ are non-tips in $F$ and have no dangling blocks, all of the above are non-empty}
    \For{$d_u,d_v \in\signs$}{
        \If{$E^{d_u}_u \neq \emptyset$, $E^{d_u}_u \cap E^{\hat{d}_u}_u = \emptyset$, $E^{d_v}_v \cap E^{\hat{d}_v}_v = \emptyset$, $E^{d_u}_u = E^{d_v}_v$}{
            \tcp{Equivalently, $E^{\hat{d}_u}_u \neq \emptyset$, $E^{d_u}_u \cap E^{\hat{d}_u}_u = \emptyset$, $E^{d_v}_v \cap E^{\hat{d}_v}_v = \emptyset$, $E^{\hat{d}_u}_u = E^{\hat{d}_v}_v$}
            \If{$|E^{d_u}_u| = 1$ and the pertaining node $\beta$ of the unique edge in $E^{d_u}_u$ is an S-node}{
                \textit{GoodVertices} $\gets $ all the good vertices in $\skel(\beta)$\;
                \If{ \textit{GoodVertices} $ \setminus \{u,v\} = \emptyset$}{
                    \tcp{\textit{GoodVertices} allows us to decide the condition above in constant time. For this reason, it is important that S-nodes are processed before P-nodes.}
                    Report $\{ud_u,vd_v\}$\; \label{line:P-node-snarls-1-1}
                }
            }
            \Else{
                Report $\{ud_u,vd_v\}$\; \label{line:P-node-snarls-1-2}
            }
            \If{$|E^{\hat{d}_u}_u| = 1$ and the pertaining node $\beta$ of the unique edge in $E^{\hat{d}_u}_u$ is an S-node}{
                \textit{GoodVertices} $\gets $ all the good vertices in $\skel(\beta)$\;
                \If{\textit{GoodVertices} $ \setminus \{u,v\} = \emptyset$}{
                    Report $\{u\hat{d}_u,v\hat{d}_v\}$\; \label{line:P-node-snarls-2-1}
                }
            }
            \Else{
                Report $\{u\hat{d}_u,v\hat{d}_v\}$\; \label{line:P-node-snarls-2-2}
            }
        }
    }
}
\end{algorithm}

\pnodesnarls*
\begin{proof}
    $(\Rightarrow)$ Suppose that $\{ud_u,vd_v\}$ is separable in $F$ with component $X$. We show that each condition described in the statement holds.

    If $u$ and $v$ are tips in $F$ with signs $d_u$ and $d_v$, respectively, then each condition clearly holds. By (1) of \Cref{thm:where-are-snarls-after-cutting} we can thus assume in the remainder of the proof that $u$ and $v$ are both non-tips in $F$.

    By \Cref{prop:mixed-dangling-blocks} it follows that $u$ and $v$ have no dangling blocks since $\{ud_u,vd_v\}$ is separable.

    If $E_u^{d_u} = \emptyset$ then $H$ has no incidences at $u$ with sign $d_u$. Since $F$ is a sign-cut graph, there is a block of $F$ containing opposite incidences at $u$, for otherwise $u$ is a non-tip and a sign-consistent vertex in $F$, contradicting the fact that $F$ is a sign-cut graph of $G$. In other words, $v$ has a dangling block with respect to $H$, but we already established that $u$ has no dangling blocks. Therefore $E_u^{d_u} \neq \emptyset$.

    If $E_u^{{d}_u} \cap E_u^{\hat{d}_u}$ contains an edge $e \in \skel(\mu)$ then $\expansion(e)$ has edges $\{ud_u,ad_a\}$ and $\{u\hat{d}_u,bd_b\}$.
    In $\expansion(e)$ there is an $a$-$b$ path avoiding $u$ and $v$ otherwise $a$ and $b$ would be in different split components of $\mu$, contradicting the fact that $a,b \in V(\expansion(e))$. Thus the graph resulting from splitting $ud_u$ and $vd_v$ has a $u$-$u'$ path. So $E_u^{{d}_u} \cap E_u^{\hat{d}_u} = \emptyset$. Symmetrically we get $E_v^{{d}_v} \cap E_v^{\hat{d}_v} = \emptyset$.
    
    If $E_u^{{d}_u} \neq E_v^{d_v}$ then, without loss of generality, there is an edge $e \in E_u^{{d}_u} \setminus E_v^{{d}_v}$. Since $E_u^{{d}_u} \cap E_u^{\hat{d}_u} = \emptyset$ and $E_v^{{d}_v} \cap E_v^{\hat{d}_v} = \emptyset$, it follows that every incidence of $\expansion(e)$ in $v$ has sign $\hat{d}_v$. So $\expansion(e)$ has a $u$-$v$ path $p$ that starts with a $ud_u$ incidence and ends with a $v\hat{d}_v$ incidence, which moreover does not contain a $vd_v$ incidence. Let $F'$ denote the graph resulting from splitting $ud_u$ and $vd_v$.
    If $E_u^{{d}_u} \cap E_v^{d_v} = \emptyset$ then $F'$ has no $u$-$v$ path, a contradiction. Otherwise $u$ and $v$ are connected in $F'$ and $p$ exists in $F'$ (where the occurrence of $v$ is replaced by $v'$), and hence there is a $v$-$v'$ path via $u$ in $F'$, a contradiction. Therefore $E_u^{{d}_u} = E_v^{d_v}$.

    $(\Leftarrow)$
    If $E^{d_u}_u \neq \emptyset$, $E^{d_u}_u \cap E^{\hat{d}_u}_u = \emptyset$, $E^{d_v}_v \cap E^{\hat{d}_v}_v = \emptyset$, $E^{d_u}_u = E^{d_v}_v$, and $u$ and $v$ have no dangling blocks with respect to $H$, then the separability of $\{ud_u,vd_v\}$ follows at once.

    
\end{proof}

\begin{algorithm}[ht]
\caption{$\mathsf{FindSnarlsBetweenRRnodes}(F,H,T)$\label{alg:find-snarls-in-RR}}
\KwIn{Sign-cut graph $F$ of a bidirected graph, maximal 2-connected subgraph $H$ of $G$, SPQR tree $T$ of $H$}
\For{$\{\nu,\mu\} \in E(T)$}{
    \If{$\nu$ and $\mu$ are R-nodes}{
        $e_{\mu} \gets$ the virtual edge in $\skel(\nu)$ pertaining to $\mu$\;
        $e_{\nu} \gets$ the virtual edge in $\skel(\mu)$ pertaining to $\nu$\;
        $u,v \gets$ the endpoints of $e_\nu,e_\mu$\;
        $X_\mu, X_\nu \gets \expansion(e_\mu), \expansion(e_\nu)$\;
        \If{$\mathsf{HasDangling}(F,H,u)$ or $\mathsf{HasDangling}(F,H,u)$ or $u$ is a tip in $F$ or $v$ is a tip in $F$}{
            \textbf{continue}\;
        }
        \For{$d_u,d_v \in\signs$}{
            \If{$N^{d_u}_H(u) \subseteq V(X_\mu)$ and $N^{\hat{d}_u}_H(u) \subseteq V(X_\nu)$ and $N^{d_v}_H(v) \subseteq V(X_\mu)$ and $N^{\hat{d}_v}_H(v) \subseteq V(X_\nu)$}{
                Report $\{ud_u, vd_v\}, \{u\hat{d}_u, v\hat{d}_v\}$\; \label{line:R-node-snarls}
            }
        }
    }
}
\end{algorithm}

\rnodesnarls*
\begin{proof}
    First we argue on the separability.
    Let $G'$ denote the graph after splitting $ud_u$ and $vd_v$.
    
    Since $u$ and $v$ have no dangling block with respect to $H$, every incidence in each block distinct from $H$ that intersects $u$ or $v$ has the same sign, so in $G'$ those blocks intersecting $u$ with $ud_u$ incidences will remain attached to $u$ and those with $u\hat{d}_u$ incidences are reattached to $u'$, and the same for $v$. Thus $u$ and $u'$ are not connected in $G'$ via any of these blocks, and the same for $v$ and $v'$.
    
    Further, notice that the tree-edge $\{\nu,\mu\}$ encodes a separation of $H$ where one side consists of $V(\expansion(e_\nu))$. Since $\expansion(e_\nu)$ is connected and all incidences at $u$ and $v$ have signs $d_u$ and $d_v$, respectively, every $u$-$v$ path in $\expansion(e_\nu)$ starts with a $ud_u$ incidence and ends with a $vd_v$ incidence, so $u$ and $v$ are connected in $G'$. Moreover, because all incidences of $\expansion(e_\mu)$ at $u$ and $v$ have signs $\hat{d}_u$ and $\hat{d}_v$, respectively, $u$ and $u'$ are not connected in $G'$ by a path containing $v$ or $v'$, and analogously $v$ and $v'$ are not connected by a path containing $u$ or $u'$. Therefore $\{ud_u,vd_v\}$ is separable.
    
    For minimality notice that $\expansion(e_\nu)$ has two internally vertex-disjoint $u$-$v$ paths since $\skel(\nu)$ is 3-connected. So we can apply \Cref{prop:snarls-disjoint-paths} and conclude that $\{ud_u,vd_v\}$ is a snarl.
\end{proof}

The next proposition is merely technical and will help on the proof of the main theorem.

\begin{prop}
\label{prop:reaches-in-expansion}
    Let $G$ be a bidirected graph, let $H$ be a maximal 2-connected subgraph of $G$, and let $T$ be the SPQR tree of $H$. Let $\mu$ be a node of $T$, let $e=\{u,v\}$ be a virtual edge of $\skel(\mu)$, and let $a \in V(H) \setminus \{u\}$ be a vertex. If $a \in V(\expansion(e))$ then $\expansion(e)$ has an $a$-$v$ path avoiding $u$.
\end{prop}
\begin{proof}
    If $a=v$ we are done, so let us argue the case $a \neq v$. Suppose for a contradiction that every $a$-$v$ path in $\expansion(e)$ contains $u$. Since $a$ is contained in a split component with respect to $\{u,v\}$, every $a$-$v$ path in $H$ also contains $u$. So $u$ is an $a$-$v$ cutvertex with respect to $H$, contradicting the fact that $H$ is 2-connected.
\end{proof}

\begin{prop}
\label{prop:R-node-mixed}
    Let $G$ be a bidirected graph, let $H$ be a maximal 2-connected subgraph of $G$, and let $T$ be the SPQR tree of $H$. Let $\nu$ be a node of $T$ such that $\skel(\nu)$ has a virtual edge $e_\mu=\{u,v\}$ pertaining to an R-node $\mu$. If $\expansion(e_\mu)$ has opposite incidences at $u$, then $\{ud_u,vd_v\}$ is not separable.
\end{prop}
\begin{proof}
    Suppose that $\expansion(e_\mu)$ has edges $\{ud_u,ad_a\}$ and $\{u\hat{d}_u,bd_b\}$.
    Notice that vertex $a$ is contained in the expansion of a virtual edge $e_a$ of $\skel(\mu)$ (and analogously $e_b$ for vertex $b$), and since $u$ and $v$ are not both the endpoints of these virtual edges (as $a,b\in V(\expansion(e_\mu))$, we have that $u$ is one endpoint of $e_a$ and the other endpoint, say $a'$, is distinct from $v$, and that $u$ is one endpoint of $e_b$ and the other endpoint, say $b'$, is distinct from $v$. Now notice that $\expansion(e_a)$ has an $a$-$a'$ path avoiding $u$ by \Cref{prop:reaches-in-expansion}, and analogously $\expansion(e_b)$ has a $b$-$b'$ path avoiding $u$.
    Since these paths are contained in $e_a$ (and $e_b$) and $v \notin \expansion(e_a),\expansion(e_b)$, they also avoid $v$.
    Moreover, $\skel(\mu)$ has an $a'$-$b'$ path avoiding $\{u,v\}$ because the skeleton of R-nodes are 3-connected. Therefore, $\expansion(e_\mu)$ has an $a$-$b$ avoiding $\{u,v\}$, and thus the graph resulting from splitting $ud_u$ and $vd_v$ has a $u$-$u'$ path, and so $\{ud_u,vd_v\}$ is not separable.
\end{proof}

We can now present the correctness proof of the algorithm. In short, the algorithm computes the sign-cut graphs of $G$. Then, for each sign-cut graph its SPQR tree is built, and from there those snarls whose vertices form a separation pair are detected. 

\begin{theorem}
\label{thm:snarls-correct}
    Let $G$ be a bidirected graph. The algorithm identifying snarls (\Cref{alg:snarls}) is correct, that is, it identifies all snarls of $G$ and only its snarls.
\end{theorem}
\begin{proof}

    \textbf{(Completeness.)}
    We argue that every snarl of $G$ is reported by the algorithm.

    Let $\{ud_u,vd_v\}$ be a snarl of $G$. By \Cref{lem:snarl-equivalence} there is a sign-cut graph of $G$ where $\{ud_u,vd_v\}$ is also a snarl.
    Thus, let $F$ denote the sign-cut graph where $\{ud_u,vd_v\}$ is a snarl. Clearly, $F$ is examined by the algorithm.
    By (1) of \Cref{thm:where-are-snarls-after-cutting} it follows that either $u$ and $v$ are both tips or both non-tips in $F$.
    If $u$ and $v$ are both tips then the snarl in question is encoded in the list described in Line~\ref{line:tip-tip-snarls} of \Cref{alg:snarls} (in $\mathcal{T}_i$ every two incidences form a snarl).
    Otherwise $u$ and $v$ are both non-tips in $F$. By (3) of \Cref{thm:where-are-snarls-after-cutting} it follows that $u$ and $v$ belong to the same block $H$ of $F$, and $u$ and $v$ have no dangling blocks with respect to $H$. Thus $u$ and $v$ are both non-tips only in $H$.
    Moreover, since $u$ (or $v$) is a non-tip, \Cref{thm:snarls-split-pairs} implies that the pair $\{u,v\}$ is an edge of $U(F)$ or it is a separation pair of a block of $F$. In fact, since the blocks of $F$ partition its edges and no two distinct blocks contain the same two vertices, in fact we have that $\{u,v\} \in U(H)$ or that $\{u,v\}$ is a separation pair of $H$. We argue on the two cases. Let $T$ denote the SPQR tree of $H$, let $F'$ denote the graph resulting from splitting $ud_u$ and $vd_v$, and let $X \subseteq F'$ denote the snarl component of $\{ud_u,vd_v\}$ restricted to $F'$.
    
    If $\{u,v\}$ is a separation pair of $H$ then \Cref{prop:spqr-tree-contains-split-pairs} implies that $T$ has an edge $\{\nu,\mu\}$ corresponding to the separation pair $\{u,v\}$ or $T$ has an S-node where $u$ and $v$ are nonadjacent. Therefore it is enough to analyze all the S- and P-nodes individually, and all virtual edges contained in the skeletons of the R-nodes. Importantly, we remark that the current assumptions do not exclude the possibility that $\{u,v\}$ is an edge of $H$. We discuss each of the possible cases.

    \begin{enumerate}
        \item Suppose that $\mu$ is an S-node of $T$ such that $\{u,v\}\subseteq V(\skel(\mu))$. If $u$ and $v$ are good in $\mu$ and consecutive in the (circular) list of good vertices with corresponding incidences $d_u$ and $d_v$ then $\{ud_u,vd_v\}$ is reported in Line~\ref{line:S-node-snarls} (notice that here, $\{u,v\}$ may be a real edge of $\skel(\mu)$). If $u$ and $v$ are good in $\mu$ and not consecutive in the (circular) list of good vertices, then there is a good vertex $w$ in between $u$ and $v$ such that $wd_w$ is an incidence violating minimality of $\{ud_u,vd_v\}$ (see the proof of \Cref{prop:S-node-snarls}), a contradiction to the fact that $\{ud_u,vd_v\}$ is a snarl. If $u$ and $v$ are not good then we require a careful argument, for which we do case analysis on the adjacency relation between $u$ and $v$.
        \begin{enumerate}
            \item If $u$ and $v$ are nonadjacent in $\skel(\mu)$ then $\{ud_u,vd_v\}$ is not separable. To see why, notice that at least one edge $e$ incident to $u$ is such that $\expansion(e)$ has edges $\{ud_u,ad_a\}$ and $\{u\hat{d}_u,bd_b\}$ (because $u$ is not good), and notice that $a,b\neq v$ because $u$ and $v$ are nonadjacent in $\skel(\mu)$. Let $w$ be the first vertex in $\skel(\mu)$ on an $a$-$v$ path in $\expansion(e)$ that avoids $u$ (such a path exists by \Cref{prop:reaches-in-expansion}). Notice that $w \neq v$ since $u$ and $v$ are not adjacent. Due to the structure of S-nodes, doing the same for $b$ also yields vertex $w$. Then $H$ has an $a$-$b$ without $v$, and also without $u$ by construction. So $F'$ has an $a$-$b$ path and thus it has a $u$-$u'$ path, a contradiction.
            
            \item Otherwise $u$ and $v$ are adjacent in $\skel(\mu)$, so let $e=\{u,v\} \in \skel(\mu))$. There are two cases to consider.
            
            If $e=\{u,w\}\in E(\skel(\mu))$ is an edge with $w \neq v$ such that $\expansion(e)$ has opposite incidences at $u$ then by an identical argument as the one just made above we get a contradiction on the separability of $\{ud_u,vd_v\}$; symmetrically, the same applies to $v$.

            The last case is thus when $e$ is such that $\expansion(e)$ has opposite incidences at $u$ and $v$, and the other two edges $e_u$ and $e_v$ of $\skel(\mu)$ adjacent at $u$ and $v$, respectively, are such that their expansions have only incidences of the same sign at $u$ and $v$.
            If the pertaining node of $e$ is an R-node then \Cref{prop:R-node-mixed} gives a contradiction to the separability of $\{ud_u,vd_v\}$.
            Thus, the pertaining node of $e$ is a P-node and the snarl is reported once P-nodes are analyzed, as shown below in item (2).
        \end{enumerate}

        The last case is when only one vertex between $u$ and $v$ is good. In this case $\{ud_u,vd_v\}$ is easily seen to not be separable (the argument is essentially the same as the one used in item (1a)).

        \item Suppose that $\mu$ is a P-node of $T$ such that $V(\skel(\mu))=\{u,v\}$. Since $\{ud_u,vd_v\}$ is separable, \Cref{prop:P-node-snarls} implies the conditions described in the statement. Due to minimality of $\{ud_u,vd_v\}$, there is no incidence $wd_w$ with $w \in \expansion(e)\setminus \{u,v\}$ such that $\{ud_u,wd_w\}$ and $\{vd_v,w\hat{d}_w\}$ are separable.
        We do case analysis on the number of edges in $E^{d_u}_u$.
        
        Suppose that $|E^{d_u}_u| = 1$ and let $e$ be the unique edge $e$ in $E^{d_u}_u$. If $e$ pertains to an S-node $\alpha$ then $\skel(\alpha)$ has no good vertices except $\{u,v\}$, for otherwise such a vertex paired with a sign $d_w\in\signs$ can serve as incidences $wd_w$ and $w\hat{d}_w$, thus violating minimality (the correctness of this step is essentially explained in the proof of \Cref{prop:S-node-snarls}). Thus, in this case, the algorithm reports the snarl in Line~\ref{line:P-node-snarls-1-1} (or Line~\ref{line:P-node-snarls-2-1}, symmetrically). If $e$ does not pertain to an S-node then the snarl is reported in Line~\ref{line:P-node-snarls-1-2} (symmetrically, Line~\ref{line:P-node-snarls-2-2}) and the same if $|E^{d_u}_u| > 1$.

        \begin{observation}\label{obs:edge-case}
            In this case, the algorithm also finds snarls when $\{u,v\}$ is not a separation pair. More precisely, when $\skel(\mu)$ has two real edges $\{ud_u,vd_v\}$ and $\{u\hat{d}_u,v\hat{d}_v\}$ and one virtual edge pertaining to an S-node whose expansion contains only incidences at $u$ and $v$ of signs ${d_u}$ and ${d_v}$, respectively, the snarl $\{ud_u,vd_v\}$ is reported.
        \end{observation}
        
        \item Suppose that $\mu$ is an R-node of $T$ with a virtual edge $\{u,v\}$. If the pertaining node $\nu$ of this virtual edge is an S- or a P-node then $\{ud_u,vd_v\}$ was reported before. So $\nu$ is an R-node. We claim that the incidences in $\expansion(e_\nu)$ all have the same sign in $u$, and those in $\expansion(e_\mu)$ all have the opposite sign in $u$, and the same for $v$. Suppose for a contradiction (and without loss of generality) that $\expansion(e_\mu)$ has incidences of opposite signs at $u$. Then \Cref{prop:R-node-mixed} gives a contradiction to the fact that $\{ud_u,vd_v\}$ is separable.
        Now notice that the conditions just established on the incidences of $u$ and $v$ are precisely those described in \Cref{alg:find-snarls-in-RR} and thus the snarl is reported in Line~\ref{line:R-node-snarls} when the assignments to the sign variables $d_u$ and $d_v$ are those matching that of the snarl.
    \end{enumerate}

    Otherwise $\{u,v\}$ is an edge of $H$ and is not a separation pair of $H$.
    We are in conditions of applying \Cref{prop:snarl-edge-case} from where two cases follow.
    We have $e=\{ud_u,vd_v\} \in E(H)$ and all incidences at $u$ and $v$ except for those in $e$ have signs $\hat{d_u},\hat{d}_v$, respectively, or $e=\{u\hat{d}_u,v\hat{d}_v\} \in E(H)$ and all incidences at $u$ and $v$ except for those in $e$ have signs $d_u,d_v$, respectively.
    In the former case the snarl $\{ud_u,vd_v\}$ is reported in Line~\ref{line:edge1-snarls}.
    In the latter case the snarl $\{ud_u,vd_v\}$ is reported in Line~\ref{line:edge2-snarls} (where the signs are written with the respective opposites) if also no S-node of $T$ contains both $u$ and $v$.
    Indeed, if no S-node of $T$ contains both $u$ and $v$ then the snarl is reported.
    
    Otherwise, let $\beta$ be the S-node whose skeleton contains $u$ and $v$; our goal is to either derive a contradiction by concluding that $\{ud_u,vd_v\}$ is not a snarl or to show that the snarl $\{ud_u,vd_v\}$ is reported in another phase of the algorithm.
    Notice that $u$ and $v$ are adjacent in $\skel(\beta)$ because $\{u,v\}$ is an edge of $H$. Thus let $e=\{u,v\} \in E(\skel(\beta))$. We do case analysis on the type of $e$.
    If $e$ is a real edge then the snarl is reported when S-nodes are analyzed: $u$ and $v$ are classified as good vertices due to the assumption on the incidences and moreover they are consecutive in the circular list of good vertices.
    Otherwise $e$ is a virtual edge. Since $\{u,v\}$ is not a separation pair, the pertaining node of this edge is a P-node $\alpha$. Notice that $\skel(\alpha)$ consists of exactly two real edges and one virtual edge (which pertains to $\beta$). Indeed, if $\skel(\alpha)$ has three real edges then it is not hard to see that $\{ud_u,vd_v\}$ is not separable.
    Moreover, all incidences in the expansion of $\{u,w\} \in \skel(\beta)$ for $w \neq v$ all have the same sign (due to the assumption on the incidences at $u$), and the same for the expansion of the other edge of $\skel(\beta)$ incident to $v$ where $u$ is not an endpoint.
    But then we are exactly in the conditions described in \Cref{obs:edge-case}, and thus the snarl is reported when P-nodes are analyzed.
    
    All cases were examined and thus every snarl of $G$ is reported by the algorithm.

    \textbf{(Soundness.)}
    We argue that the algorithm only reports snarls.
    
    By \Cref{lem:snarl-equivalence}, if $\{ud_u,vd_v\}$ is a snarl in a sign-cut graph of $G$ then it is also a snarl in $G$. Thus, let $F$ denote the sign-cut graph where the pair $\{ud_u,vd_v\}$ is reported. It is enough to show that $\{ud_u,vd_v\}$ is separable and minimal in $F$.

    If $ud_u$ and $vd_v$ are incidences of the set built in Line~\ref{line:tip-tip-snarls} of \Cref{alg:snarls} then $u$ and $v$ are both tips in $F$. By (2) of \Cref{thm:where-are-snarls-after-cutting} it follows that that $\{ud_u,vd_v\}$ is a snarl and thus the set $\mathcal{T}_i$ encodes only snarls.

    If $\{ud_u,vd_v\}$ is reported by virtue of Line~\ref{line:S-node-snarls} of \Cref{alg:find-snarls-in-S} then $ud_u$ and $vd_v$ are consecutive elements of $W$ and $u\neq v$. Moreover, the conditions expressed in \Cref{alg:find-snarls-in-S} clearly identify all and only good vertices. Thus \Cref{prop:S-node-snarls} implies that $\{ud_u,vd_v\}$ is a snarl.

    If $\{ud_u,vd_v\}$ and $\{u\hat{d}_u,v\hat{d}_v\}$ are reported in Line~\ref{line:R-node-snarls} of \Cref{alg:find-snarls-in-RR} then \Cref{prop:R-node-snarls} implies that $\{ud_u,vd_v\}$ is a snarl (the conditions of the statement match those in the algorithm). Applying \Cref{prop:R-node-snarls} symmetrically to the adjacent R-node implies that also $\{u\hat{d}_u,v\hat{d}_v\}$ is a snarl.
    
    If $\{ud_u,vd_v\}$ is reported in \Cref{alg:find-snarls-in-P} then \Cref{prop:P-node-snarls} implies that $\{ud_u,vd_v\}$ is separable (with component $X$).
    We argue on the minimality for each of the lines a pair of incidences is reported. Without loss of generality, we argue only for Line~\ref{line:P-node-snarls-1-1} and Line~\ref{line:P-node-snarls-1-2}.
    If $|E^{d_u}_u| > 1$ then $X$ has two internally vertex-disjoint $u$-$v$ paths (this is easily seen from the structure of P-nodes), so we can apply \Cref{prop:snarls-disjoint-paths} and conclude that $\{ud_u,vd_v\}$ is minimal, and so it is a snarl.
    If $|E^{d_u}_u|=1$ then the unique edge $e$ in $E^{d_u}_u$ is either real or virtual. If it is real then $\{ud_u,vd_v\}$ is clearly minimal. If it is virtual then the pertaining node $\beta$ of $e$ is an S- or an R-node since no two P-nodes are adjacent in $T$. We discuss both cases.

    \begin{itemize}
        \item Suppose that $\beta$ is an S-node. The algorithm reports $\{ud_u,vd_v\}$ if $\skel(\beta)$ has no good vertex except $\{u,v\}$. To see this is correct, notice first that $u$ and $v$ are good in $\skel(\beta)$ and thus $\{ud_u,vd_v\}$ is separable. For minimality it is enough to argue that no vertex $w \in V(\skel(\beta)) \setminus \{u,v\}$ is such that $\{ud_u,wd_w\}$ and $\{vd_v,w\hat{d}_w\}$ are separable. Since $w$ is not good but $u$ and $v$ are, it is not hard to see that the graph obtained by splitting $ud_u$ and $wd_w$ or $vd_v$ and $w\hat{d}_w$ is such that $w$ and $w'$ are connected.
        
        \item Suppose that $\beta$ is an R-node. Then $X$ has two internally vertex-disjoint $u$-$v$ paths because $\skel(\beta)$ is 3-connected and $X$ is a subgraph of $\skel(\beta)$ with the edge $\{u,v\}$ removed. Applying \Cref{prop:snarls-disjoint-paths}, we conclude that $\{ud_u,vd_v\}$ is a snarl. (Another way to see this is by noticing that the conditions expressed in \Cref{alg:find-snarls-in-P} are sufficient to apply \Cref{prop:R-node-snarls}, from where we can also conclude that $\{ud_u,vd_v\}$ is a snarl.)
    \end{itemize}

    If $\{ud_u,vd_v\}$ is reported in Line~\ref{line:edge1-snarls} or Line~\ref{line:multi-bridge-snarl} then the pair is clearly separable (with component $X$ restricted to $F$) and minimality follows from the fact that $V(X)=\{u,v\}$.
    
    If $\{u\hat{d}_u,v\hat{d}_v\}$ is reported in Line~\ref{line:edge2-snarls} then no skeleton of an S-node of $T$ contains the vertices $u$ and $v$. Moreover, $\{u\hat{d}_u,v\hat{d}_v\}$ is clearly separable (with component $X$ restricted to $F$). We argue on the minimality.
    Since $\{u,v\}$ is an edge of $U(H)$, $T$ has a node $\mu$ whose skeleton contains the real edge $\{u,v\}$, and thus $\{u,v\}\subseteq V(\skel(\mu))$. Since no skeleton of an S-node of $T$ contains $u$ and $v$, the vertices $u$ and $v$ are contained in a P- or in an R-node, and so $H$ has three internally vertex-disjoint $u$-$v$ paths (the existence of these paths is easily seen from the description of the P- and R-nodes, nonetheless we point to Lemma 2 of~\cite{di1996line}). So $X$ has two internally vertex-disjoint $u$-$v$ paths and hence we can apply \Cref{prop:snarls-disjoint-paths} to conclude that $\{u\hat{d}_u,v\hat{d}_v\}$ is a snarl.

    Every line where the algorithm reports a pair of incidences is analyzed and therefore the algorithm only reports snarls.
\end{proof}

In order to implement \Cref{alg:snarls} we have to support inclusion-neighborhood queries in constant time. For example, in \Cref{alg:find-snarls-in-RR}, we need to decide whether $N^+_H(u) \subseteq V(\expansion(e_\nu))$ and $N^{+}_H(u) \subseteq V(\expansion(e_\mu))$ (in different iterations of $d_u \in \signs$). This condition is equivalent to all $+$ incidences at $v_i$ being in $\expansion(e_\nu)$ or all $+$ incidences at $v_i$ being in $\expansion(e_\mu)$. Therefore, more generally, for a tree-edge $\{\nu,\mu\}$ creating a separation $(V(\expansion(e_\nu),V(\expansion(e_\nu))$, we need to to know the number of $+$ and $-$ incidences in $\expansion(e_\nu)$, and likewise the number of $+$ and $-$ incidences in $\expansion(e_\mu)$.
To support this in overall linear time, we need to reuse information that we computed for subtrees of the current tree edge. We explain this in \Cref{thm:snarls-time}. In fact, this idea will be further developed for superbubbles, where we will maintain other properties (such as acyclicity) in additional to incidence counts.

\begin{theorem}
\label{thm:snarls-time}
    Let $G$ be a bidirected graph. The algorithm identifying snarls (\Cref{alg:snarls}) can be implemented to run in $O(|V(G)|+|E(G)|)$.
\end{theorem}
\begin{proof}
    Block-cut trees can be built in linear time \cite{Hopcroft73blockcut} and the total size of the blocks is linear in $|V(G)|+|E(G)|$. 
    We also find the sign-cut graphs in linear time, since we only need to identify the sign-consistent cutvertices of the block-cut tree.
    The case when a block in a multi-bridge is trivial, so suppose that we are analyzing a block $H$ that is 2-connected. Let $|H|=|V(H)|+|E(H)|$. 
    We show that the rest of the algorithm runs in time $O(|H|)$, thus proving the desired bound.

    After building the SPQR tree $T$ of $H$, which can be built in $O(|H|)$ time~\cite{gutwenger2001linear}, the algorithm examines each of the possible cases where snarls can live in. Notice that the work done in \Cref{alg:find-snarls-in-S,alg:find-snarls-in-P,alg:find-snarls-in-RR} is constant-time with exception of the neighborhood queries to $u$ and $v$.
    
    To support deciding inclusions of incidences (as discussed above), we process all edges of $T$ with a DFS traversal starting in the root. Let $\nu$ be the parent of $\mu$ in $T$ and let $\{u,v\}$ denote the endpoints of $e_\nu \in \skel(\mu)$ and of $e_\mu \in \skel(\nu)$. We store at $u$ and $v$ the number of their $+$ and $-$ incidences in $\expansion(e_\mu)$. Assume that we have already computed this information (via the DFS order) for all tree edges to children of $\mu$ in $T$ (if $\mu$ is not a leaf). For all such tree edges to children of $\mu$ in which $u$ is present, we increment the incidence counts of $u$ by these values, and same for $v$. Moreover, we scan every real edge in $\skel(\mu)$ and use the signs of the edge to correspondingly increment the incidence counts for $u$ and $v$. Doing this, we process every real edge once because every edge of the input graph is a real edge in exactly one skeleton. Having the correct incidence counts for $u$ and $v$ in $\expansion(e_\mu)$, we can obtain their incidence counts in $\expansion(e_\nu)$ by subtracting from the total number of $+$ incidences of $u$ the $+$ incidence count of $u$ in $\expansion(e_\mu)$ (and same for $v$). This can again be obtained by paying only constant time per edge.

\end{proof}

\snarlscorrecttime*
\begin{proof}
    Follows from \Cref{thm:snarls-correct} and \Cref{thm:snarls-time}.
\end{proof}

\begin{algorithm}[ht]
\caption{$\mathsf{FindEdgeSnarls}(F,H,T)$\label{alg:find-edge-snarls}}
\KwIn{Sign-cut graph $F$ of a bidirected graph, maximal 2-connected subgraph $H$ of $F$, $T$ the SPQR tree of $H$}
\For{every edge $e = \{ud_u, vd_v\}$ of $H$}{
    \If{every edge of $F$ except $e$ is incident to $u$ and $v$ with signs $\hat{d}_u$ and $\hat{d}_v$, respectively}{
        Report $\{ud_u, vd_v\}$\; \label{line:edge1-snarls}
        \If{no S-node $\mu$ of $T$ is such that $\{u,v\}\subseteq V(\skel(\mu))$}{
            Report $\{u\hat{d}_u, v\hat{d}_v\}$\; \label{line:edge2-snarls}
        }
    }
}
\end{algorithm}

\begin{algorithm}[ht]
\caption{Snarls representation algorithm}
\label{alg:snarls}
\KwIn{Bidirected graph $G$}
\KwOut{A linear-size encoding of snarls of $G$ as subsets $\mathcal{T} = \{T_1, \dots, T_k\}$ and $\mathcal{S} = \{S_1, \dots, S_\ell\}$ of incidences, where all pairs of incidences in $T_i$ are snarls and each $S_j$ is a snarl.} 
$F_1,\dots,F_k \gets \mathsf{BuildSignCutGraphs}(G)$\;
$\mathcal{T} \gets \{\}$\;
$\mathcal{S} \gets \{\}$\;
\For{$\iink$}{
    \If{$F_i$ is an isolated vertex}{
        \textbf{continue}\;
    }
    $T_i \gets$ the set of incidences $vd_v$ where $v$ is a tip in $F_i$ with sign $d_v$\;\label{line:tip-tip-snarls}
    $\mathcal{T} \gets \mathcal{T} \cup \{T_i\}$\;
    \For{every block $H$ of $F_i$}{
        \If{$H$ is a multi-bridge}{
            \For{every edge $\{ud_u,vd_v\} \in E(H)$}{
                \If{every edge of $G$ except $e$ is incident to $u$ and $v$ with signs $\hat{d}_a$ and $\hat{d}_v$, respectively}{
                    $\mathcal{S} \gets \mathcal{S} \cup \{ud_u,vd_v\}$\; \label{line:multi-bridge-snarl}
                }
            }
        }
        \Else{
            $T \gets \mathsf{BuildSPQRTree}(H)$\;
            $\mathcal{S} \gets \mathcal{S} \cup \mathsf{FindSnarlsInSnodes}(F,H,T)$\;
            $\mathcal{S} \gets \mathcal{S} \cup \mathsf{FindSnarlsInPnodes}(F,H,T)$\;
            $\mathcal{S} \gets \mathcal{S} \cup \mathsf{FindSnarlsBetweenRRnodes}(F,H,T)$\;
            $\mathcal{S} \gets \mathcal{S} \cup \mathsf{FindEdgeSnarls}(F,H,T)$\;
        }
    }
}
\Return{$(\mathcal{T}, \mathcal{S})$}
\end{algorithm}

                    

\newpage
\section{Finding superbubbles via SPQR trees}
\label{sec:superbubbles}

In this section we develop a linear-time algorithm to find superbubbles in a directed graph $G$. At a high level, our approach is based on two observations. First, if $B$ is a superbubble with entry point $s$ and exit point $t$, then $\{s,t\}$ is a separation pair of the undirected counterpart of $G$, unless some (manageable) corner cases occur. Secondly, deciding whether a directed graph is a DAG with unique source and unique sink can be done in linear time in the size of the graph, and the computed information can be reused to avoid expensive recomputations for nested superbubbles.

\subsection{Preliminaries}

We begin by giving basic terminology on directed graphs and further definitions required for this section.

\paragraph{Directed graphs.}
A directed graph $G=(V,E)$ has vertex set $V(G)$ and edge set $E(G)$. The sets of \emph{out-neighbors} and \emph{in-neighbors} of a vertex $v$ in $G$ are denoted by $N^+_G(v)$ and $N^-_G(v)$, respectively. We denote an edge from $u$ to $v$ as $uv$ and a path $p$ through vertices $v_1,\dots,v_k$ as $p=v_1\dots v_k$. In this case we say that there is a \emph{$v_1$-$v_k$ path}, that $v_1$ and $v_k$ are \emph{connected} (if there is no such path then they are \emph{disconnected}), or that $v_1$ reaches $v_k$; the vertices of the path $p$ but the first and the last are its \emph{internal vertices}. Without loss of generality, we assume that directed graphs have no parallel arcs since they have no effect on superbubbles.
If $V' \subseteq V(G)$, $E' \subseteq E(G)$, and $E' \subseteq V' \times V'$, we say that $G'=(V',E')$ is a \emph{subgraph} of and write $G' \subseteq G$. Respectively, we say that $G$ is a \emph{supergraph} of $G'$.
We say that a subgraph $G'$ of $G$ is \emph{maximal} on a given property in the sense of the subgraph relation, i.e., when no proper supergraph of $G'$ contained in $G$ has that property.
The subgraph \emph{induced} by a subset $C \subseteq V(G),E(G)$ of vertices or edges of $G$ is denoted by $G[C]$. The \emph{vertex-induced} subgraph is the graph with vertex set $C$ and the subset of edges in $E(G)$ whose endpoints are in $C$. The \emph{edge-induced} subgraph has edge set $C$ and a vertex set consisting of all endpoints of edges in $C$.
A vertex~$v$ is a \emph{source} of $G$ if $|N^-_{G}(v)|=0$ and a \emph{sink} if $|N^+_{G}(v)|=0$.

Similarly to the main matter, we denote by $U(G)$ the undirected graph obtained from $G$ by removing the direction of every edge, and keeping parallel edges that possibly appear because (to avoid losing information, we assume that the edges of $G$ are labeled with unique identifiers that are retained).
As a convention, we refer to directed graphs by $G$ and to undirected graphs by $H$.
We say that vertex $v$ is an \emph{extremity} of a directed graph $G$ if $v$ is a source or sink of $G$, or a cutvertex of $U(G)$.

\paragraph{Additional definitions on SPQR trees.}

Let $G$ be a directed graph.
We will routinely solve subproblems on the skeletons of the nodes of $T$ where their virtual edges are assigned directions depending on the reachability relation of $G$ restricted to their expansions. So let $\mu$ be a node of $T$ and let $e_1=\{s_1,t_1\},e_2=\{s_2,t_2\},\dots,e_k=\{s_k,t_k\}$ be its virtual edges $(k \geq 2)$.
Define the set of directed edges $B_1 = \{ s_it_i : \text{$s_i$ reaches $t_i$ in $\expansion(e_i) ,\; i=1,\dots,k$} \}$ and $B_2 = \{ t_is_i : \text{$t_i$ reaches $s_i$ in $\expansion(e_i),\; i=1,\dots,k $} \}$. We define the \emph{directed skeleton} as $\dirskel(\mu) := (V(\skel(\mu)),B_1 \cup B_2)$.


\subsection{Superbubbles}

The notion of \emph{superbubble} is intended to characterize a specific type of subgraph that commonly arises in graphs built from biological data. These subgraphs encompass one of the main obstacles in the assembly problem. Superbubbles are a generalization of \emph{bubbles}: two vertices connected by many edges with identical directions. Superbubbles were proposed by Onodera~et~al.~\cite{onodera2013detecting} in the setting of directed graphs. They also presented an average-case linear-time algorithm enumerating superbubbles with quadratic complexity in the worst case. Linear-time algorithms enumerating superbubbles are known for directed graphs~\cite{gartner-revisited}, and linear-time algorithms tailored for DAGs are also known~\cite{brankovic2016linear}.
We next recap the definition of superbubbles on directed graphs. 
\begin{definition}[Superbubble, Onodera et al.~\cite{onodera2013detecting}]
\label{def:superbubbles-directed graphs}
    Let $G=(V,E)$ be a directed graph. An ordered pair of distinct vertices $(s,t)$ is a superbubble of $G$ if
    \begin{enumerate}[nosep]
        \item $t$ is reachable from $s$ (\emph{reachability});
        \item the set of vertices reachable from $s$ without using $t$ as an internal vertex coincides with the set of vertices reaching $t$ without using $s$ as an internal vertex, and call $B_{st}$ the induced subgraph of $G$ by this set of vertices (\emph{matching});
        \item $B_{st}$ is acyclic (\emph{acyclicity}); and
        \item no vertex in $B_{st}$ except $t$ forms a pair with $s$ that satisfies above properties (1)--(3) (\emph{minimality}).
    \end{enumerate}
\end{definition}

If only properties (1)--(3) are satisfied, then $(s,t)$ is a \emph{superbubbloid} (see~G{\"a}rtner and Stadler~\cite{gartner2019direct}).
If $(s,t)$ is a superbubble in $G$ then the \emph{interior} of $B_{st}$ is the set $V(B_{st}) \setminus \{s,t\}$. The simplest superbubble consists of a single arc $st$ or a set of parallel arcs $st$ where $N^+_s=\{t\}$, $N^-_t=\{s\}$, and $ts \notin E(G)$; we call this a \emph{trivial superbubble}. The interior of a superbubble does not contain sources or sinks of $G$ due to the matching property.

An equivalent definition of superbubble/superbubbloid is given by G{\"a}rtner et al.~\cite{gartner-revisited}.

\begin{definition}[Superbubbloid~\cite{gartner-revisited,gartner2019direct}]\label{def:gartner}
    Let $G$ be a digraph, $B \subseteq G$, and $s, t \in V(B)$. Then, $B$ equals $B_{st}$ of~\Cref{def:superbubbles-directed graphs} and satisfies the reachability and matching conditions if and only if the following conditions (1)--(4) are satisfied. Moreover, $B$ forms a \emph{superbubbloid} with entrance $s$ and exit $t$ if and only if 
    \begin{enumerate}[nosep] 
        \item every $u \in V(B)$ is reachable from $s$;
        \item $t$ is reachable from every $u \in V(B)$;
        \item if $u \in V(B)$ and $w \in V(G) \setminus V(B)$, then every $w$-$u$ path contains $s$;
        \item if $u \in V(B)$ and $w \in V(G) \setminus V(B)$, then every $u$-$w$ path contains $t$;
        \item if $uv$ is an edge in $B$, then every $v$-$u$ path in $G$ contains both $t$ and $s$; and
        \item $G$ does not contain the edge $ts$.
    \end{enumerate}
\end{definition}

A superbubble $(s,t)$ is a superbubbloid that is minimal in the sense that there is no $s'\in V(B_{st})-s$ such that $(s',t)$ is a superbubbloid.
In other words, a superbubble $(s,t)$ is just superbubbloid such that no vertex in $V(\Bst)\setminus \{s,t\}$ appears in every $s$-$t$ path of $\Bst$.
We also say that $B_{st}$ is a superbubbloid even though a superbubble is an ordered pair of vertices by definition.
One easily observes that if $(s,t)$ is a superbubbloid then no in-neighbour of $s$ and no out-neighbour of $t$ is in $\Bst$, and every out-neighbour of $s$ and every in-neighbour of $t$ is in $\Bst$.

We give some results on the relation between cutvertices and superbubbloids/superbubbles. Importantly, we show that cutvertices (in a broad sense) are not in the interior of superbubbles.

\begin{prop}
\label{prop:cutvertex-inside-bubble}
    Let $G$ be a directed graph and let $s,t$ be vertices of $G$. If $(s,t)$ is a superbubble of $G$ then no vertex in the interior of $B_{st}$ is an $s$-$t$ cutvertex in $U(\Bst)$.
\end{prop}

\begin{prop}
\label{prop:superbubble-independent-paths}
    Let $G$ be a directed graph and let $(s,t)$ be a superbubbloid of $G$ with graph $\Bst$. If $U(\Bst)$ has two internally vertex-disjoint $s$-$t$ paths then $(s,t)$ is a superbubble.
\end{prop}

\begin{lemma}[Superbubbles and cutvertices]\label{lem:bubbles-cutvertices}
    Let $G$ be a directed graph and let $(s,t)$ be a superbubble of $G$ with graph $\Bst$. Then no vertex in the interior of $B_{st}$ is a cutvertex of $U(G)$.
\end{lemma}
\begin{proof}
    We can assume that $\Bst$ contains at least three vertices.
    Suppose for a contradiction that the interior of $\Bst$ contains a cutvertex of $U(G)$. 
    There are two cases to analyze.

    \textbf{No block contains both $s$ and $t$:}
    Since no block contains both $s$ and $t$ there is a vertex $v$ whose removal disconnects $s$ from $t$ in $U(G)$. Vertex $v$ is thus reachable from $s$ without $t$ in $G$ and hence $v \in V(\Bst)$ because $(s,t)$ is a superbubble. Moreover, since every $s$-$t$ path in $U(G)$ contains $v$ and $U(\Bst) \subseteq U(G)$ (because $\Bst \subseteq G$), every $s$-$t$ path in $U(\Bst)$ contains $v$ and thus $v$ is an $s$-$t$ cutvertex with respect to $U(\Bst)$. Applying \Cref{prop:cutvertex-inside-bubble} gives a contradiction.

    \textbf{There is a block containing $s$ and $t$:}
    Let $v$ be a cutvertex in the interior of $\Bst$.
    Let $v'$ be the first vertex on the sequence of $s$-$v$ cutvertices in $U(G)$ (notice that the sequence of cutvertices between any two vertices is totally ordered). Then $s$ and $v'$ are in the same block, and thus so is $t$.
    
    Since $v$ is reachable from $s$ in $G$ without using $t$ (because $v \in V(\Bst)$) and every $s$-$v$ path in $U(G)$ contains $v'$, $v'$ is also reachable from $s$ without $t$. Therefore $v' \in V(\Bst)$ since $(s,t)$ is a superbubble.
    Further, since $s,t,v'$ are in the same block and $v'$ is a cutvertex of $U(G)$ there is a vertex $w \in V(G) \setminus \{s,t,v'\}$ such that $w$ belongs to a block containing $v'$ and not $s$ and $t$; moreover, $w$ can be chosen so that $v'w$ or $wv'$ is an edge of $G$.
    
    Notice now that every $s$-$w$ and $t$-$w$ path in $U(G)$ contains $v'$ and thus every $s$-$w$ and $w$-$t$ path in $G$ contains $v'$.
    If $w \in V(\Bst)$ then there is a cycle in $\Bst$ through $w$ and $v'$, a contradiction.
    Suppose now that $w \notin V(\Bst)$. Since $v' \in V(\Bst)$, $s$ reaches $v'$ without $t$ and $v'$ reaches $t$ without $s$ in $G$. So if $v'w \in E(G)$ then $s$ reaches $w$ without $t$ in $G$, and because $(s,t)$ is a superbubble it implies that $w \in V(\Bst)$, a contradiction; the case where $w'v \in E(G)$ follows symmetrically.
\end{proof}

By \Cref{lem:bubbles-cutvertices} a cutvertex of $U(G)$ can only be the entry or the exit of a superbubble. Therefore the superbubbles of $G$ are confined to the blocks of $G$ and there is a unique block that contains both the entrance and exit of of any given superbubble.
Then the task of computing superbubbles in a directed graph $G$ reduces to that of computing superbubbles in each block of $G$. Since block-cut trees can be built in linear time, if we can find superbubbles inside a block in linear time then we can find every superbubble of an arbitrary graph also in linear time.

Our approach to enumerate superbubbles is based on the following results.

\begin{theorem}[Superbubbles and split pairs]\label{thm:bubbles-split-pairs}
    Let $G$ be a weakly connected directed graph and let $(s,t)$ be a superbubble of $G$. Let $H_1,\dots,H_{\ell}$ be the blocks of $U(G)$ $(\ell \geq 1)$. Then $\{s,t\}$ is a split pair of some $H_i$ or $V(B_{st})=V(H_i)$.
\end{theorem}
\begin{proof}
    It follows from~\Cref{lem:bubbles-cutvertices} that $V(B_{st})$ is contained in a block of $U(G)$. Assume without loss of generality that $V(B_{st}) \subseteq V(H_1)$.
    We can assume that $|V(B_{st})| \geq 3$, $V(B_{st}) \neq V(H_1)$. We show that $\{s,t\}$ disconnects $H_1$.
    Let $u \in V(B_{st}) \setminus \{s,t\}$ and let $v \in V(H_1) \setminus V(B_{st})$.
    By 3 and 4 of \Cref{def:gartner}, every $u$-$v$ path in $G$ contains $t$ and every $v$-$u$ path in $G$ contains $s$. But then every $u$-$v$ path in $H_1$ contains $s$ or $t$, and therefore $\{s,t\}$ is a separation pair of $H_1$.
\end{proof}


The interesting case of~\Cref{thm:bubbles-split-pairs} is when the vertices identifying a superbubble form a separation pair of a block, as all the other cases can be dealt with a linear-time preprocessing step.
By~\Cref{prop:spqr-tree-contains-split-pairs} we know that every separation pair of $U(G)$ is encoded as the endpoints of a virtual edge of some node of the SPQR tree (except for nonadjacent vertices $u$ and $v$ contained in an S-node, but these will not cause problems). Conversely, the vertices of any virtual edge in a node of the SPQR tree forms a separation pair.
Therefore, by correct examination of the virtual edges of the SPQR tree we can obtain the complete set of superbubbles of $G$.
The high level idea of our algorithm is to examine every node/edge of the SPQR tree in an ordered way. In each step we gather information that depends on the subtrees the node/edge separates the tree into.

\begin{lemma}[Unique orientation at poles of acyclic components]\label{lem:unique-poles-from-acyclicity}
    Let $G$ be a directed graph and let $C \subseteq V(G)$, $|C| 
    \geq 2$, be such that $U(G[C])$ is connected and $G[C]$ is acyclic. Moreover, let $s,t \in C$ be such that for all other vertices $v \in C \setminus \{s,t\}$ there is no edge in $U(G)$ between $v$ and some $v' \in V(G) \setminus C$. If no vertex in $C \setminus \{s,t\}$ is a source or sink of $G$, then one vertex among $\{s,t\}$ is the unique source of $G[C]$ and the other vertex is the unique sink of $G[C]$.
\end{lemma}
\begin{proof}
    Notice that since $G[C]$ is acyclic, it has at least one source (relative to $G[C]$), say $v$. If $v \in C \setminus \{s,t\}$, then it is also a source of $G$, since by the hypothesis there is no edge in $U(G)$ between $v$ and some $v' \in V(G) \setminus C$. This contradicts the assumption that no vertex in $C \setminus \{s,t\}$ is a source of $G$. Therefore, any source of $G[C]$ belongs to $\{s,t\}$. y

    Symmetrically, we have that any sink of $G[C]$ belongs to $\{s,t\}$. Observe that neither $s$ nor $t$ can be both a source and a sink of $G[C]$, because otherwise it would be an isolated vertex with no edges in $G[C]$, which contradicts the assumption that $U(G[C])$ is connected. Therefore, since $G[C]$ has at least one source and at least one sink (being acyclic), one vertex among $\{s,t\}$ is the unique source of $G[C]$ and the other vertex is the unique sink in $G[C]$.
\end{proof}

\subsection{The superbubble finding algorithm}

\begin{figure}[h]
    \centering
    \includegraphics[scale=1]{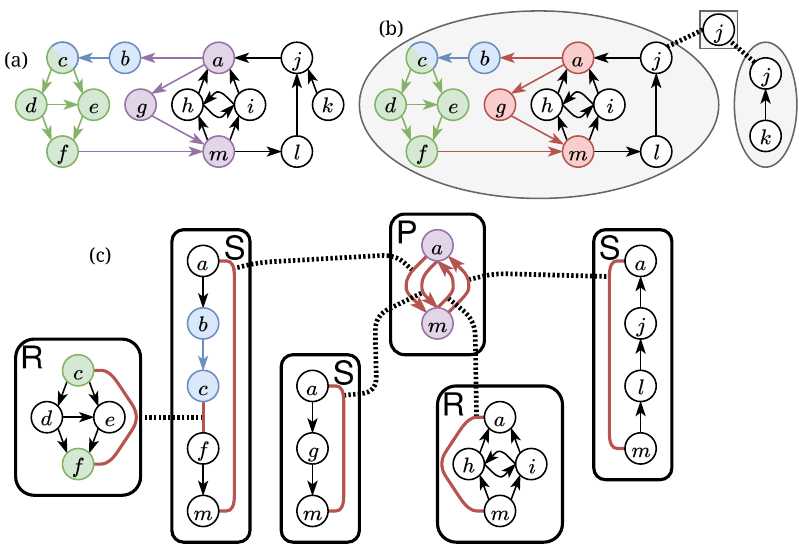}
    \caption{
    (a) A graph with highlighted superbubbles $(a, m)$ (purple, blue and green), $(b, c)$ (blue) and $(c, f)$ (green).
    Superbubbles $(b, c)$ and $(c, f)$ are nested inside $(a, m)$.
    For example, $(l, k)$ is not a superbubble because $l$ does not reach $k$; $(l, j)$ is not a superbubble because $l$ does not reach $k$ without using $j$ as an internal vertex, but $k$ reaches $j$ without using $l$ as an internal vertex; $(m, a)$ is not a superbubble because (among others condition violations) $B_{ma}$ contains the cycle $(h, i, h)$; and $(b, f)$ is a superbubbloid but it is not a superbubble because $c$ disproves minimality.
    (b) The block cut tree of the graph, with superbubbles highlighted in the same way.
    Oval nodes are blocks and square nodes are cutnodes.
    In the larger block, it looks like $(l, j)$ was a superbubble, but since it has an incoming arc from outside the block, it is not.
    (c) The SPQR tree of the larger block with virtual edges in bold red.
    The endpoints of the virtual edges that represent the superbubbles are highlighted.
    Superbubble $(a, m)$ consists of the two $(a, m)$ virtual edges in the P-node.
    Here, we also marked which of the two P-node-groups the virtual edges belong to by giving them a direction.
    Superbubble $(b, c)$ is found while examining each edge of the graph for being a trivial superbubble, and superbubble $(c, f)$ is found while analyzing the edge between the corresponding S-node and R-node.
    }
    \label{fig:superbubbles-spqr}
\end{figure}

In this section we develop an algorithm to find superbubbles in a directed graph $G$. Our algorithm is based on traversals over the SPQR trees of the blocks of $U(G)$.

In a preliminary phase we build the block-cut tree of $U(G)$ to obtain its blocks $H_1,\dots,H_\ell$ $(\ell\geq1)$. The algorithm reports pairs $(s,t)$ only within these blocks (recall \Cref{lem:bubbles-cutvertices}). Multi-bridges $\{s,t\}$ are handled with a simple predicate (essentially, these can only correspond to trivial superbubbles). Next, those superbubbles whose vertex set coincides with the vertex set of the current block are reported. Then, by \Cref{thm:bubbles-split-pairs}, any remaining superbubble in $H_i$ induces a separation pair in $H_i$, for which we build the SPQR tree $T_i$ of $H_i$ and have a dedicated routine to find superbubbles therein.
We can assume in what follows that the SPQR tree of $H_i$ does not consist of a single node, since otherwise every superbubble is either $H_i$ or a single edge by \Cref{thm:bubbles-split-pairs}.

There are two main pieces of information that are carried during the algorithm: reachability between endpoints of virtual edges and acyclicity of subgraphs. The minimality property will follow from the structure of the skeletons of the nodes of the SPQR tree. The matching property will also follow from structure and it interplays with the absence of sources/sinks in the skeletons.
The next result supports these claims and will be used later on frequently.

\begin{lemma}
\label{lem:separation-pair-is-superbubbloid}
    Let $G$ be a directed graph, let $\{s,t\}$ be a separation pair of a maximal 2-connected subgraph of $U(G)$, and let $H$ be the disjoint union of a nonempty subset of the split components of $\{s,t\}$. Let $K=G[V(H)]$. If there are no extremities of $G$ in $V(K)\setminus\{s,t\}$, $K$ is acyclic, $N^+_G(s) \subseteq V(K)$, $N^-_G(t) \subseteq V(K)$, and $ts \notin E(G)$, then $(s,t)$ is a superbubbloid of $G$ with graph $K$.
\end{lemma}
\begin{proof}
    Let $K_1,\dots,K_k$ denote the set of split components of $\{s,t\}$ whose union is $K$ $(k\geq1)$. Since split components are connected, $K$ has edges incident to $s$, and since every out-neighbour of $s$ is in $K$ by assumption we have $|N^+_G(s)| \geq 1$ (analogously, $|N^-_G(t)| \geq 1$). Notice that each $K_i$ is a DAG with unique source $s$ and unique sink $t$, since a source or sink of $G$ except $\{s,t\}$ in $K_i$ contradicts the absence of extremities except $\{s,t\}$ in $K$ and a cycle in $K_i$ contradicts the acyclicity of $K$.
    
    We show that $K$ fulfills each condition of \Cref{def:gartner}.

    \begin{enumerate}
        \item Let $u \in V(K_i)$ and let $p=v_1,\dots,v_\ell$ be a maximal path in $K_i$ containing $u$ $(\ell \geq 1)$.
        
        Notice that if $v_1\neq s$ then $v_1$ is not a source of $K_i$, thus $K_i$ has an in-neighbour $u$ of $v_1$. If $u$ is not in $p$ then $p$ can be prolonged, contradicting its maximality. Otherwise $u$ is in $p$ and thus there is a cycle in $K_i$, contradicting the acyclicity of $K$ as $K_i\subseteq K$. Therefore $v_1=s$, and analogously one can show that $v_\ell=t$. Thus, if $u=s$ or $u=t$ then $p$ witnesses that $s$ reaches $t$, otherwise $p$ witnesses that every vertex distinct from $s$ and $t$ is reachable from $s$ and reaches $t$. As $\bigcup_{i=1}^kV(K_i)=V(K)$ it follows that every vertex in $K$ is reachable from $s$ and reaches $t$.
        
        \item (Proved in the previous item.)
        
        \item Let $u\in V(K)$ and $w\in V(G) \setminus V(K)$ and let $K_i$ be the split component that contains $u$.
        Suppose for a contradiction that there is a $w$-$u$ path in $G$ avoiding $s$. Then there is a $w$-$u$ path $p$ in $U(G)$ avoiding $s$. This path contains $t$. Suppose otherwise. Then $p$ avoids both $s$ and $t$. Thus, $G$ has an edge with one endpoint in $V(G)\setminus V(K)$ and the other in $V(K_i)\setminus\{s,t\}$ for some $\iink$, which contradicts the fact that $K$ is the union of split components of $\{s,t\}$, or the vertex in $K_i$ is an endpoint of an edge not in $H$, so that vertex is a cutvertex of $U(G)$, which contradicts the fact that $G$ has no extremities except $\{s,t\}$.
        Thus $t$ has an in-neighbour $v$ in $V(G)\setminus V(K)$.
        
        Let $p=v_1,\dots,v_\ell$ be a maximal path in $K$ starting at $t$ $(\ell \geq 2)$. Since we showed in the previous item (1) that $s$ reaches $t$, $v_\ell \neq s$ otherwise $K$ has a cycle, a contradiction. Thus $v$ is distinct from $s$ and $t$. Since $K$ has no extremities except $\{s,t\}$, $v$ has in-neighbors, all of which are in $p$ (for otherwise $p$ can be prolonged, contradicting its maximality), and thus there is a cycle in $K$, a contradiction. Therefore every $w$-$u$ path in $G$ contains $s$.
        
        \item (Analogous to the previous item.)
        
        \item Let $uv \in E(K)$. Let $K_i$ denote the split component that contains $u$.
        If $G$ has a $v$-$u$ path $p$ avoiding $t$ then $p$ uses only vertices of $K_i$ because the out-neighbours of $s$ are all in $K$ and $K_i$ is a split component. Therefore $K_i$ has a $v$-$u$ path, which together with the edge $uv$ creates a cycle, contradicting the fact that $K_i$ is acyclic. Analogously, one can derive the same argument for $s$. Therefore every $v$-$u$ path in $G$ contains both $t$ and $s$.
        
        \item Direct by assumption.
    \end{enumerate}
\end{proof}

Let $\{\nu,\mu\} \in E(T_i)$ such that $\nu$ is the parent of $\mu$, and let $e_\nu$ be the virtual edge in $\skel(\mu)$ pertaining to $\nu$ and define $e_\mu$ analogously; let $\{s,t\}$ denote the endpoints of these virtual edges. In the edge $\{\nu,\mu\}$ we store two pieces of information, the state corresponding to the subgraph $\expansion(e_\mu)$ as $\state{\nu}{\mu}$ and the state corresponding to $\expansion(e_\nu)$ as $\state{\mu}{\nu}$. We say that $\state{\nu}{\mu}$ \emph{leaves} $\nu$ and that it \emph{enters} $\mu$. (This can be seen as a directed edge in the tree pointing from $\nu$ to $\mu$.) Notice that any state uniquely identifies a pair of virtual edges $e_\nu$ and $e_\mu$.
Let $X=\expansion(e_\mu)$. In $\mathsf{State_{\nu,\mu}}$ we store the following information:

\begin{itemize}
    \item $\noextremity{\nu}{\mu} := \true$ iff no vertex in $V(X)\setminus\{s,t\}$ is an extremity of $G$.
    
    \item $\acyclic{\nu}{\mu} := 
    \begin{cases}
        \Null, & \text{if $\noextremity{\nu}{\mu}$ is false,}\\
        \true, & \text{otherwise, if $X$ is acyclic,}\\
        \false, & \text{otherwise.}
    \end{cases}$\\
    \item $\reachesst{\nu}{\mu} := 
    \begin{cases}
        \Null, & \text{if $\acyclic{\nu}{\mu}$ is $\false$ or $\Null$,}\\ 
        \true, & \text{otherwise, if $s$ reaches $t$ in $X$,}\\
        \false, & \text{otherwise.}
    \end{cases}
    $\\
    \item $\reachests{\nu}{\mu} := 
    \begin{cases}
        \Null, & \text{if $\acyclic{\nu}{\mu}$ is $\false$ or $\Null$,}\\ 
        \true, & \text{otherwise, if $t$ reaches $s$ in $X$,}\\
        \false, & \text{otherwise.}
    \end{cases}
    $
\end{itemize}

With this information we can ``almost'' decide if a separation pair $\{s,t\}$ identifies a superbubble $(s,t)$, since if $\noextremity{\nu}{\mu}$ and $\acyclic{\nu}{\mu}$ are $\true$, $N^+_G(s),N^-_G(t) \subseteq V(\expansion(e_\mu))$, and  $ts\notin E(G)$, then $(s,t)$ is a superbubbloid with graph $\expansion(e_\mu)$ by~\Cref{lem:separation-pair-is-superbubbloid}. This also explains that most of our effort is in the computation of $\state{\nu}{\mu}$ for every edge $\{\nu,\mu\}$ of $T$, as the other conditions can be checked easily. The minimality property, as we will see, does not raise any additional challenges as it follows from the structure of the nodes of the SPQR tree.

To update the states we traverse the (rooted) SPQR tree $T_i$ of $H_i$, henceforth denoted for simplicity as $T$ and $H$, respectively.
The algorithm consists of three phases.

\begin{itemize}[label={}]
    \item \emph{Phase 1.} Process the edges $\{\nu,\mu\}$ of $T$ (with $\nu$ the parent of $\mu$) with a DFS traversal starting in the root, and compute all $\state{\nu}{\mu}$.
    \item \emph{Phase 2.} Process the nodes $\nu$ of $T$ with a BFS traversal starting in the root. For every child $\mu$ of $\nu$, we compute all $\state{\mu}{\nu}$.
    \item \emph{Phase 3.} Examine the separation pairs $\{s,t\}$ of $H$ via $T$ and use the information computed in the previous phases to decide whether $(s,t)$ or $(t,s)$ are superbubbles.
\end{itemize}



\paragraph{Phase 1.}
Phase 1 is a dynamic program over the edges of $T$.
Let $\nu$ be the parent of $\mu$ in $T$ and let $\{s,t\}$ denote the endpoints of $e_\nu \in \skel(\mu)$ and of $e_\mu \in \skel(\nu)$.
If $\mu$ has no children then the edges of its skeleton but $e_\nu$ are all real edges, and hence the problem of updating $\state{\nu}{\mu}$ is trivial: with DFS on $\dirskel(\mu)-st-ts$ we can decide $\noextremity{\nu}{\mu}$, $\acyclic{\nu}{\mu}$, $\reachesst{\nu}{\mu}$, and $\reachests{\nu}{\mu}$.
Otherwise $\mu$ has at least one virtual edge besides $e_\nu$. Let us denote the children of $\mu$ by $\mu_1,\dots,\mu_k$ $(k\geq 1)$ and denote the endpoints of the corresponding virtual edges $e_i$ in $\skel(\mu)$ as $\{s_i,t_i\}$ for all $i \in \{1,\dots,k\}$.
Assume recursively that $\state{\mu}{\mu_i}$ is solved and let $X = \expansion(e_\mu)$, $X_i = \expansion(e_i)$ for all $i \in \{1,\dots,k\}$, and let $K=\dirskel(\mu)-st-ts$.

We now describe how to compute the states $\state{\nu}{\mu}$.
\begin{description}
    \item[$\noextremity{\nu}{\mu}$:] We set $\noextremity{\nu}{\mu}$ to $\true$ if and only if no vertex in $V(K) \setminus \{s,t\}$ is an inner extremity and $\noextremity{\mu}{\mu_i}$ is $\true$ for all $i \in \{1,\dots,k\}$. To see this is correct we prove both implications.
    
    ($\Rightarrow$) Suppose no vertex in $V(X) \setminus \{s,t\}$ is an extremity. Then indeed, no vertex in $V(K) \setminus \{s,t\}$ is an extremity because $V(K) \subseteq V(X)$.
    Moreover, $\noextremity{\mu}{\mu_i}$ must be $\true$ for all $i \in \{1,\dots,k\}$, for otherwise an extremity $x$ in $X_i$ is different from both $s_i$ and $t_i$ and thus also different from $s$ and $t$, as it does not belong to $\skel(\mu)$ since $\{s_i,t_i\}$ is a separation pair.

    $(\Leftarrow)$ Suppose no vertex in $V(K) \setminus \{s,t\}$ is an extremity and $\noextremity{\mu}{\mu_i}$ is $\true$, for all $i \in \{1,\dots,k\}$. For a contradiction, assume that some $x \in V(X) \setminus \{s,t\}$ is an extremity. By the initial assumption, we have that $x$ cannot belong to $V(K)$. Thus, $x$ is also different from $s_i,t_i$, for all $i \in \{1,\dots,k\}$. Since $x \in V(X) \setminus \{s,t\}$, $x$ must belong to $X_i$, for some $i \in \{1,\dots,k\}$. Therefore, it is an extremity for it, since it is different from $s_i$ and $t_i$. This contradicts the initial assumption that $\noextremity{\mu}{\mu_i}$ is $\true$.

    \item[$\acyclic{\nu}{\mu}$:] If $\noextremity{\nu}{\mu}$ is $\false$ then we set $\acyclic{\nu}{\mu}$ to $\Null$, which is correct by definition. Thus, in the following we assume that $\noextremity{\nu}{\mu}$ is $\true$.
    
    If for some $\iink$ $\acyclic{\mu}{\mu_i}$ is $\Null$, then by definition $\noextremity{\mu}{\mu_i}$ is $\false$. Let thus $x$ be an extremity in $V(X_i) \setminus \{s_i,t_i\}$. Since $\{s_i,t_i\}$ is a separation pair, we have that $x \notin \{s,t\}$. Thus, $x \in V(X) \setminus \{s,t\}$, which contradicts the fact $\noextremity{\nu}{\mu}$ is $\true$.
    
    If for some $\iink$ $\acyclic{\mu}{\mu_i}$ is $\false$, then $X_i$ has a cycle, which implies that also $X$ contains a cycle because $X_i$ is a subgraph of $X$. Since $\noextremity{\nu}{\mu}$ is $\true$, then by definition we can set $\acyclic{\nu}{\mu}$ to $\false$.

    \begin{sloppypar}
    Finally, we are in the case where for every $\iink$, $\acyclic{\mu}{\mu_i}$ is $\true$, and therefore $\reachesst{\mu}{\mu_i}$ and $\reachests{\mu}{\mu_i}$ are $\true$ or $\false$. In other words, each $X_i$ is acyclic, and, importantly, the reachability in $X_i$ between the endpoints of each virtual edge $e_i$ are known.
    
    Then $K$ can be built explicitly and we can set $\acyclic{\nu}{\mu}$ to $\true$ if $K$ is acyclic and to $\false$ otherwise. To see this is correct, note that any cycle $C$ in $X$ can be mapped to a cycle in $K$: whenever $C$ uses edges of some $X_i$, it passes through $s_i$ (or $t_i$), and since $X_i$ is acyclic, it must return to $t_i$ (or $s_i$). This path of $C$ in $X_i$ between $s_i$ and $t_i$ (or between $t_i$ and $s_i$) can be mapped to the edge of $K$ that was introduced from $s_i$ to $t_i$, if $\reachesst{\mu}{\mu_i}$ is $\true$ (or from $t_i$ to $s_i$, if $\reachests{\mu}{\mu_i}$ is $\true$). Viceversa, every cycle $C$ in $K$ can be symmetrically mapped to a cycle in $X$ such that whenever $C$ uses some edge $s_it_i$ (or $t_is_i$) in $K$, we expand this edge into a path from $s_i$ to $t_i$ (or from $t_i$ in $s_i$) in $X_i$.
    \end{sloppypar}
    
    \item[$\reachesst{\nu}{\mu}$, $\reachests{\nu}{\mu}$:] At this point we have $\acyclic{\nu}{\mu}$ computed. If it is $\false$ or $\Null$ we set $\reachesst{\nu}{\mu}$ and $\reachests{\nu}{\mu}$ to $\Null$, which is correct by definition. Otherwise $X$ is acyclic, $V(X) \setminus \{s,t\}$ has no cutvertex, no sink and no source of $G$, and there is no edge between a vertex not in $V(X)$ and a vertex of $V(X) \setminus \{s,t\}$ since $\{s,t\}$ is a separation pair; moreover $U(X)$ is clearly connected. We can thus apply \Cref{lem:unique-poles-from-acyclicity} to conclude that one vertex between $s$ and $t$ is a source of $X$ and the other is a sink. In the former case we can set $\reachesst{\nu}{\mu}$ to $\true$, and in the latter case we can set $\reachesst{\nu}{\mu}$ to $\true$.
\end{description}

Let $\mu$ be a node of $T$ and let $\nu$ be its parent.
Since $T$ is a tree, each of its edges is processed exactly once during the DFS, thus every state of the form $\state{\nu}{\mu}$ is updated during the algorithm. Moreover, since every node $\mu$ has a unique parent in $T$ and $\dirskel(\mu)$ is built only when the edge $\{\nu,\mu\}$ is processed, each directed skeleton $\dirskel(\mu)$ is built only once during the algorithm. Since the computational work per edge $\{\nu,\mu\}$ is linear in the size of $\dirskel(\mu)$ and the total size of all skeletons is linear in the size of the current block $H$ (recall \Cref{lem:spqr-total-size}), Phase 1 runs in time $O( |V(H)| + |E(H)|)$.

\begin{algorithm}[t]
\caption{Superbubble finding -- Phase 1}
\label{alg:phase1}
\KwIn{Directed graph $G$, SPQR tree $T$ having at least two nodes}

\SetKwFunction{ProcessEdge}{ProcessEdge}
\SetKwProg{Fn}{Function}{:}{}
\Fn{\ProcessEdge{$\nu,\mu$}}{
    \tcp{$\nu$ is the parent of $\mu$}
    \For{$\mu_i \in \mathsf{children}_T(\mu)$}{
        \ProcessEdge{$\mu,\mu_i$}\;
    }
    \If{$\nu = \text{null}$}{
        \Return{}\;
    }
    $\{s,t\} \gets e_\mu$ (where $e_\mu \in \skel(\nu)$)\;
    \textit{noExtr} $\gets \text{false}$ iff $V(\skel(\mu)) \setminus \{s,t\}$ has an extremity of $G$\;

    Let $e_1,\dots,e_k$ denote the virtual edges of $\mu$ with pertaining nodes $\mu_1,\dots,\mu_k$ $(k \ge 0)$\;
    \textit{ThereIsNoExtremityBelow} $\gets \bigwedge_{i=1}^k \noextremity{\mu}{\mu_i}$\;
    \textit{ThereIsNoCycleBelow} $\gets \bigwedge_{i=1}^k \acyclic{\mu}{\mu_i}$\;
    $\noextremity{\nu}{\mu} \gets$ \textit{noExtr} $\land$ \textit{ThereIsNoExtremityBelow}\;

    \If{$\noextremity{\nu}{\mu}$ is $\false$}{
        $\acyclic{\nu}{\mu} \gets \Null$\;
    }
    \Else{
        \If{$\neg$ ThereIsNoCycleBelow}{
            $\acyclic{\nu}{\mu} \gets \false$\;
        }
        \Else{
            \tcp{We are in conditions to build $\dirskel(\mu) - st - ts$}
            $K \gets \dirskel(\mu) - st - ts$\;
            $\acyclic{\nu}{\mu} \gets \true$ iff $K$ is acyclic\; \tcp{Run DFS or BFS on $K$}
        }
    }

    \If{$\acyclic{\nu}{\mu}$ is $\true$}{
        \If{$(N^+_K(s) \cap V(K) \neq \emptyset)$}{
            \tcp{See \Cref{lem:unique-poles-from-acyclicity}}
            $\reachesst{\nu}{\mu} \gets \true$;
            $\reachests{\nu}{\mu} \gets \false$\;
        }
        \Else{
            $\reachesst{\nu}{\mu} \gets \false$;
            $\reachests{\nu}{\mu} \gets \true$\;
        }
    }
    \Else{
        $\reachesst{\nu}{\mu} \gets \Null$;
        $\reachests{\nu}{\mu} \gets \Null$\;
    }
}

$\rho \gets$ the root of $T$\;
\ProcessEdge{$\text{null}, \rho$}\;

\end{algorithm}

\paragraph{Phase 2.} In Phase 2 we compute the states $\state{\mu}{\nu}$ with $\nu$ the parent of $\mu$ by processing the \emph{nodes} of $T$ via Breadth-First Search, i.e., we compute the states ``pointing'' towards the root.
Notice that the dependencies between states behave differently from Phase 1. Now the relevant states for $\state{\mu}{\nu}$ are those leaving $\nu$ to its children except $\mu$, and the state leaving $\nu$ to its parent whenever $\nu$ is different from the root of $T$; the former states are known from Phase 1 and the latter state is known due to the breadth-first traversal order. If we follow the same strategy of computation as in Phase 1 then the algorithm may have a worst-case quadratic running time. For example, if $T$ consists only of $\rho$ with children $\mu_1,\dots,\mu_k$, then in order to update $\acyclic{\mu_i}{\nu}$ we would have to build $\dirskel(\nu)-s_it_i-t_is_i$ for each $i=1,\dots,k$, which would have a quadratic running time in $|V(G)| + |V(E)|$ whenever, e.g., $|V(\skel(\nu))| \geq |V(G)|/2$. To overcome this issue we examine the states $\acyclic{\mu_i}{nu}$ for $i=1,\dots,k$ all ``at once''.

Let $\nu$ be a node of $T$. Let $\mu_1,\dots,\mu_k$ denote the children of $\nu$ and denote the endpoints of the corresponding virtual edges in $\skel(\nu)$ as $e_i = \{s_i,t_i\}$ for $i \in \{1,\dots,k\}$ $(k \geq 1)$. To distinguish the reference edges $e_\nu$ belonging to each node $\mu_i$, we write $e_\nu^i$ for the edge $e_\nu$ in node $\mu_i$. Assume from the breadth-first traversal order that the states leaving $\nu$ to its parent are known and, for convenience, denote by $e_0=\{s_0,t_0\}$ the reference edge of $\nu$ and by $\mu_0$ the parent of $\nu$ (if $\nu$ is the root of $T$ then $\mu_0$ can be ignored during the following discussion). So the neighbours of $\nu$ in $T$ are the nodes $\mu_0,\mu_1\dots,\mu_k$.
Let $X_i = \expansion(e_\nu^i)$, $K=\dirskel(\nu)$, and $K_i=K-s_it_i-t_is_i$ for every $\iink$.

First we compute $\noextremity{\mu_i}{\nu}$ for each $\iink$ similarly to Phase 1.

\begin{description}
    \item[$\noextremity{\mu_i}{\nu}$:] We set $\noextremity{\mu_i}{\nu}$ to $\true$ if and only if no vertex in $V(K_i) \setminus \{s_i,t_i\}$ is an inner extremity and $\noextremity{\nu}{\mu_j}$ is $\true$ for $\iinkz$ and $j \neq i$. To see this is correct we prove both implications.
    
    ($\Rightarrow$) Suppose no vertex in $V(X_i) \setminus \{s_i,t_i\}$ is an extremity. Then indeed, no vertex in $V(K_i) \setminus \{s_i,t_i\}$ can be an extremity.
    Moreover, $\noextremity{\nu}{\mu_j}$ must be $\true$ for each $j \in \{1,\dots,k\}$ distinct from $i$, for otherwise an extremity $x$ in $\expansion(e_j)$ is different from both $s_j$ and $t_j$ and thus also different from $s_i$ and $t_i$, as it does not belong to $\skel(\nu)$ since $\{s_j,t_j\}$ is a separation pair.

    $(\Leftarrow)$ Suppose no vertex in $V(K_i) \setminus \{s_i,t_i\}$ is an extremity, $\noextremity{\nu}{\mu_j}$ is $\true$ for all $j \in \{1,\dots,k\}$ distinct from $i$, and $\noextremity{\nu}{p(\nu)}$ is $\true$. For a contradiction, assume that some $x \in V(X_i) \setminus \{s_i,t_i\}$ is an extremity. By the initial assumption, we have that $x$ cannot belong to $V(K_i)$. Thus, $x$ is also different from $s_i,t_i$. Since $x \in V(X_i) \setminus \{s_i,t_i\}$, $x$ must belong to $\expansion(e_{p(\nu)})$ or to $\expansion(e_j)$, for some $j$ distinct from $i$. Therefore, it is an extremity for it, since it is different from $s_i$ and $t_i$. This contradicts the initial assumption that $\noextremity{\mu}{\mu_i}$ is $\true$.
\end{description}

Then we compute the states $\acyclic{\mu_i}{\nu}$ for all $\iink$.
Notice that at this point the states $\reachesst{\nu}{\mu_i}$ and $\reachests{\nu}{\mu_i}$ are known for all $\iinkz$. We proceed by cases on the values of these states.

\begin{itemize}

    \item If there is an $\iinkz$ such that $\reachesst{\nu}{\mu_i}$ or $\reachests{\nu}{\mu_i}$ is $\Null$, then by definition $\acyclic{\nu}{\mu_i}$ is $\Null$ or $\false$. Then we proceed by cases.

    \begin{itemize}
        \item If $\acyclic{\nu}{\mu_i}$ is $\Null$ then $\noextremity{\nu}{\mu_i}$ is $\false$ by definition, and so there is an extremity $x \in V(\expansion(e_i)) \setminus \{s_i,t_i\}$. So for every $\jink$ distinct from $i$, vertex $x$ is an extremity also for $X_j$: $x \in V(X_j)$ because $\expansion(e_i)$ is a subgraph of $X_j$ and $x$ is different from $s_j,t_j$ since $x$ is different from $s_i,t_i$ and $\{s_i,t_i\}$ is a separation pair; thus each state $\acyclic{\mu_j}{\nu}$ is $\Null$.
        
        For the remaining state $\acyclic{\mu_i}{\nu}$ we proceed by cases. First, if $\noextremity{\mu_i}{\nu}$ is $\false$ then $\acyclic{\mu_i}{\nu}$ is $\Null$. We can thus assume that $\noextremity{\mu_i}{\nu}$ is $\true$, which implies that $\acyclic{\nu}{\mu_j}$ is $\true$ or $\false$ for each $\jinkz$ distinct from $i$.
        If some $\acyclic{\nu}{\mu_j}$ is $\false$ then $\expansion(e_j)$ has a cycle, and hence so does $X_i$ as it is a supergraph of $\expansion(e_j)$; therefore $\acyclic{\mu_i}{\nu}$ is $\false$. Otherwise every $\acyclic{\nu}{\mu_j}$ is $\true$ and thus we are in conditions to build $K_i$ since the states $\reachesst{\nu}{\mu_j}$ and $\reachests{\nu}{\mu_j}$ are not $\Null$ by definition. Then $\acyclic{\mu_i}{\nu}$ is $\false$ if and only if $K_i$ has a cycle because any cycle in $X_i$ can be mapped to a cycle in $K_i$ (similarly to the acyclicity update rule discussed in Phase 1).
        \item Otherwise $\acyclic{\nu}{\mu_i}$ is $\false$. So $\expansion(e_i)$ contains a cycle $C$. Then for every $\jink$ with $j\neq i$, $\acyclic{\mu_j}{\nu}$ is $\false$ since $C\in \expansion(e_i) \subseteq X_j$. For the remaining state $\acyclic{\mu_i}{\nu}$ we proceed identically as in the case above.
    \end{itemize}

    \item Otherwise $\reachesst{\nu}{\mu_i}$ and $\reachests{\nu}{\mu_i}$ are either $\true$ or $\false$ for all $\iinkz$. Therefore we are in conditions to build $K$. Moreover, notice that by definition $\acyclic{\nu}{\mu_i}$ is $\true$ for every $\iinkz$; in particular, there is no cycle in $K$ of the form $s_it_is_i$.

    If $K_i$ is acyclic for some $\iink$ then $\acyclic{\mu_i}{\nu}$ is $\true$ since also $\acyclic{\nu}{\mu_j}$ is $\true$ for every $\jinkz$ distinct from $i$ (this argument was established in Phase 1). However, we do not test the acyclicity of $K_i$ individually. First we make the simple observation that if $K$ is acyclic then so is $K_i$ because $K_i$ is a subgraph of $K$, in which case $\acyclic{\mu_i}{\nu}$ is $\true$.
    Otherwise $K$ has a cycle and we proceed with a case analysis.
    
    \begin{itemize}
        \item If $K$ has two edge-disjoint cycles $C_1,C_2$, then $\acyclic{\mu_i}{\nu}$ is $\false$ for every $\iink$. Indeed, suppose for a contradiction that $K_i$ does not contain $C_1$ nor $C_2$. Because $s_it_i$ and $t_is_i$ cannot both be present in $K$, we can assume without loss of generality that $t_is_i \notin K$. Then $K_i$ can be expressed as $K-s_it_i$. Since $K$ contains $C_1$ and $C_2$ and $K_i$ contains neither by assumption, both $C_1$ and $C_2$ intersect at $s_it_i$, a contradiction. Therefore $K_i$ contains at least one of $C_1$ or $C_2$ and hence $\acyclic{\mu_i}{\nu}$ is $\false$ for every $\iink$, as $C_1$ or $C_2$ can be mapped to a cycle in $X_i$.
        \item Otherwise any two cycles in $K$ share an edge. Without loss of generality, let us denote the \emph{virtual} edges of $K$ where any two cycles of $K$ intersect as $e_1, \dots, e_\ell$ $(\ell \leq k)$ and let $e_{\ell+1},\dots,e_k$ denote the remaining virtual edges of $K$ (those edges of $K$ that are real are irrelevant). Then $\acyclic{\mu_z}{\nu}$ is $\true$ for $z \in \{1,\dots,\ell\}$: $K_z$ is acyclic since $e_z$ is contained in every cycle of $K$ but $e_z \notin E(K_z)$, and since by assumption each $\acyclic{\mu_i}{\nu}$ is $\true$ for $\iink$, it follows that $X_z$ is acyclic.
        Complementarily, the state $\acyclic{\mu_z}{\nu}$ is $\false$ for every $z \in \{\ell+1,\dots,k\}$ because $e_z$ does not intersect at least one cycle of $K$, and thus $K_z$ contains a cycle which can be mapped to a cycle in $X_z$.
        We summarize this discussion in the next remark.
        \end{itemize}

    \begin{remark}\label{remark:cycle-hitters}
        Let $T$ be an SPQR tree, let $\nu \in V(T)$ be a node, and let $\mu_1,\dots,\mu_k$ be the neighbours of $\nu$ in $T$ whose corresponding virtual edges in $\nu$ are $e_1,\dots,e_k$. Suppose that $\acyclic{\nu}{\mu_i}$ evaluates to true for all $\iink$. Then $\acyclic{\mu_i}{\nu}$ is $\true$ if and only if $e_i$ intersects every cycle of $\dirskel(\nu)$.
    \end{remark}
    
    To address~\Cref{remark:cycle-hitters} while keeping the algorithm linear-time, it suffices to identify the edges that intersect every cycle of the directed skeleton in time proportional to its size. This is essentially the feedback arc set problem for the restricted case where every feedback set contains just one arc\footnote{In its generality, the feedback arc set problem is an NP-hard problem which asks if a directed graph $G$ has a subset of at most $k$ edges intersecting every cycle of $G$. Here, we are interested in enumerating all feedback-arcs.} (``arc'' and ``edge'' mean the same thing).

    Our subroutine works as follows. We start by testing if the graph is acyclic. If it is we are done. Otherwise we compute the strongly connected components (SCCs) of $G$. If there are multiple non-trivial SCCs then there are two disjoint cycles and no solution exists. Thus, the last case is when there is a single non-trivial SCC, where we then have to find feedback-arcs.
    
    The enumeration version of the feedback-arc problem can also be defined for \emph{vertices}. This version of the problem is solvable in linear-time, as shown by Garey and Tarjan in 1978~\cite{garey1978linear}. We rely on their algorithm and use a standard linear-time reduction from the feedback problem on arcs to the feedback problem on vertices (see, e.g., Even et al.~\cite{Even98}). We briefly describe how the reduction works. Subdivide each arc $uv$ of $G$ into two arcs $uw$ and $wv$, obtaining a subdivided graph $G'$. If an arc $uv$ is a feedback arc of $G$ then $w$ is a feedback vertex of $G'$ (deleting $w$ from $G'$ corresponds to deleting the arcs $uw$ and $wv$ in $G$), and the converse also holds. Notice, however, that $G'$ has feedback vertices that do not correspond to arcs of $G$, but those can be safely ignored.
    
    \begin{theorem}\label{thm:garey-tarjan-arcs}
        Given a directed graph $G=(V,E)$ with $n$ vertices and $m$ edges, there is an $O(m+n)$ time algorithm reporting every feedback arc of $G$.
    \end{theorem}
    
\end{itemize}

The states $\reachesst{\mu_i}{\mu},\reachests{\mu_i}{\mu}$ get updated for $\iink$ as in Phase 1.

\begin{description} 
    \item[$\reachesst{\mu_i}{\nu}$, $\reachests{\mu_i}{\nu}$:] At this point $\acyclic{\mu_i}{\nu}$ is known. If $\acyclic{\mu_i}{\nu}$ is $\false$ or $\Null$ then $\reachesst{\mu_i}{\nu}$ and $\reachests{\mu_i}{\nu}$ are $\Null$ by definition. Otherwise $\acyclic{\mu_i}{\nu}$ is $\true$, and thus so is $\noextremity{\mu_i}{\nu}$ by definition. Therefore, $X_i$ is acyclic, $V(X_i) \setminus \{s_i,t_i\}$ has no cutvertex, no sink and no source of $G$, and there is no edge between a vertex not in $V(X_i)$ and a vertex in $V(X_i) \setminus \{s_i,t_i\}$ since $\{s_i,t_i\}$ is a separation pair; moreover $U(X_i)$ is clearly connected. We can thus apply \Cref{lem:unique-poles-from-acyclicity} to conclude that one vertex between $s_i$ and $t_i$ is a source of $X_i$ and the other is a sink. In the former case we can set $\reachesst{\mu_i}{\nu}$ to $\true$, and in the latter case we can set $\reachests{\mu_i}{\nu}$ to $\true$.
\end{description}

Notice that each node $\nu$ of $T$ is processed exactly once during this phase by BFS properties. Moreover, this also implies that every state pointing from a node to its parent gets updated.
As the work done in $\nu$ is linear in the size of $\dirskel(\nu)$ (see \Cref{alg:phase2}), with \Cref{lem:spqr-total-size} we can conclude that Phase 2 runs in time $O(|V(H)|+|E(H)|)$.

\begin{lemma}\label{lem:phases-correct}
    \Cref{alg:phase1} and \Cref{alg:phase2} correctly compute the states $\state{\nu}{\mu}$ and $\state{\mu}{\nu}$ for every edge $\{\nu,\mu\}$ of $T$ and run in time $O(|V(H)|+|E(H)|)$ where $H$ is a block.
\end{lemma}

\begin{algorithm}[t!]
\caption{Superbubble finding -- Phase 2}
\label{alg:phase2}
\KwIn{Directed graph $G$, SPQR tree $T$ having at least two nodes}

$\rho \gets$ root of $T$, $\mathsf{Q} \gets \mathsf{queue()}$, $\mathsf{Q.push}(\rho)$\;

\While{$\mathsf{Q}$ is not empty}{
    $\nu \gets \mathsf{Q.pop()}$\;
    \lIf{$\nu$ has no children in $T$}{
        \textbf{continue}
    }

    Let $\mu_1,\dots,\mu_k$ be the children of $\nu$ with pertaining virtual edges $\{s_i,t_i\} = e_i$ $\in E(\skel(\nu))$, and $\mu_0$ the parent of $\nu$ with pertaining virtual edge $e_0 \in E(\skel(\nu))$\;
    \tcp{$\mu_0$ can be ignored if $\nu=\rho$}
    $\mathsf{Q.push}(\mu_i)$ \hspace{0.5em} $\forall i \in [1,k]$\;
    \textit{AllNodeExtremities} $\gets $ a set containing all the extremities of $G$ in $V(\skel(\nu))$\;
    \textit{AllEdgeExtremities} $\gets$ a set containing all the virtual edges $e_i \in E(\skel(\nu))$ such that $\noextremity{\nu}{\mu_i}$ is $\false$\;
    \For{$i \in [1,k]$}{
        \textit{noExtr} $\gets \true$ iff \textit{AllNodeExtremities} $\setminus \{s_i,t_i\} = \emptyset$\;
        \tcp{Notice that it suffices to store (up to) three extremities of $V(\skel(\mu))$ in order to update \textit{noExtr} in constant time}
        $\noextremity{\mu_i}{\nu} \gets $ (\textit{AllEdgeExtremities} $ \setminus \{\{s_i,t_i\}\} = \emptyset$) $\land$ \textit{noExtr}\;
        \tcp{Similarly, it suffices to store (up to) two virtual edges of $\skel(\nu)$ with corresponding extremity states leaving $\nu$ to $\false$}
        \If{$\noextremity{\mu_i}{\nu}$ is $\false$}{
            $\acyclic{\mu_i}{\nu} \gets \Null$\;
        }
    }
    \If{at least two states among $\{ \acyclic{\nu}{\mu_1}, \dots ,\acyclic{\nu}{\mu_k} \}$ evaluate to $\Null$}{
        $\acyclic{\mu_i}{\nu} \gets \Null \; \forall \iink$\;
    }
    \If{exactly one state among $\{ \acyclic{\nu}{\mu_1}, \dots ,\acyclic{\nu}{\mu_k} \}$ evaluates to $\Null$}{
        Let $\jink$ be such that $\acyclic{\nu}{\mu_j}=\Null$\;
        $Y \gets \{1,\dots,k\} \setminus \{j\}$\;
        $\acyclic{\mu_i}{\nu} \gets \Null, \; \forall i\in Y$\;
        \textit{AcyclicOutside} $\gets$ true iff $\bigwedge_{i\in Y \cup \{0\}} \acyclic{\nu}{\mu_i}$ is $\true$\;
        \lIf{$\neg$ AcyclicOutside}{
            $\acyclic{\mu_j}{\nu} \gets \false$
        }
        \Else{
            \tcp{We are in conditions to build $\dirskel(\mu) - s_jt_j - t_js_j$}
            $K \gets \dirskel(\nu) - s_jt_j - t_js_j$\;
            $\acyclic{\mu_i}{\nu} \gets \true$ iff $K$ is acyclic\;
        }
    }
    \If{no state among $\{ \acyclic{\nu}{\mu_1}, \dots ,\acyclic{\nu}{\mu_k} \}$ evaluate to $\Null$}{

        \If{at least two states among $\{ \acyclic{\nu}{\mu_1}, \dots ,\acyclic{\nu}{\mu_k} \}$ evaluate to $\false$}{
            $\acyclic{\mu_i}{\nu} \gets \false \; \forall \iink$\;
        }
        \If{exactly one state among $\{ \acyclic{\nu}{\mu_1}, \dots ,\acyclic{\nu}{\mu_k} \}$ evaluates to $\false$}{
            Let $\jink$ be such that $\acyclic{\nu}{\mu_j}=\false$\;
            $Y \gets \{1,\dots,k\} \setminus \{j\}$\;
            $\acyclic{\mu_i}{\nu} \gets \false, \; \forall i\in Y$\;
            \tcp{We are in conditions to build $\dirskel(\mu) - s_jt_j - t_js_j$}
            $K \gets \dirskel(\nu) - s_jt_j - t_js_j$\;
            $\acyclic{\mu_i}{\nu} \gets \true$ iff $K$ is acyclic\;
        }
        \If{no state among $\{ \acyclic{\nu}{\mu_1}, \dots ,\acyclic{\nu}{\mu_k} \}$ evaluate to $\Null$}{
            \tcp{We are in conditions to build $\dirskel(\mu)$}
            $A \gets \mathsf{FeedbackArcs(\dirskel(\nu))} \cap \{e_1,\dots,e_k\}$\;
            \tcp{$A$ contains those virtual edges of $\dirskel(\nu)$ which are feedback arcs in $\dirskel(\mu)$}
            $\acyclic{\mu_i}{\nu} \gets \true  \; \forall e_i \in A$\;
            $\acyclic{\mu_i}{\nu} \gets \false \; \forall e_i \notin A$\;
        }
    }
}
\end{algorithm}

\paragraph{Phase 3.}

\begin{algorithm}[t!]
\caption{Superbubble finding -- Phase 3}
\label{alg:phase3}
\KwIn{Directed graph $G$, SPQR tree $T$ having at least two nodes}

\For{every P-node $\mu$ of $T$}{
    Build the sets $E^+_s, E^-_t$ of $\mu$ as described in \Cref{prop:P-node}\;
    Build the sets $E^-_s, E^+_t$ analogously\;
    \If{$E^+_s = E^-_t$}{
        \tcp{Equivalently, $E^-_s = E^+_t$}
        Let $\{s,t\}$ be the vertex set of $\skel(\mu)$\;
        Let $e_1,\dots,e_k$ denote the edges in $\skel(\mu)$ $(k\ge3)$\;
        Let $\mu_1,\dots,\mu_\ell$ denote the pertaining nodes of edges in $E^+_s$ $(\ell\ge0)$\;
        Let $\mu'_1,\dots,\mu'_{\ell'}$ denote the pertaining nodes of edges in $E^-_s$ $(\ell'\ge0)$\;
        $\mathsf{assert}(\ell' = k - \ell)$\;
        
        \If{$\acyclic{\mu}{\mu_i}$ and $\noextremity{\mu}{\mu_i}$ are $\true$ for all $i \in \{1,\dots,\ell\}$}{
            \If{$N^+_G(s), N^-_G(t) \subseteq V(\bigcup_{e \in E^+_s} \expansion(e))$ and $ts \notin E(G)$}{
                \If{$\ell = 1$}{
                    \tcp{See \Cref{prop:superbubbles-in-S-nodes}}
                    Report $(s,t)$ if the pertaining node of the edge in $E^+_s$ is not an S-node\; \label{line:superbubbles-P-st-1}
                }
                \Else{
                    Report $(s,t)$\; \label{line:superbubbles-P-st}
                }
            }
        }

        \If{$\acyclic{\mu}{\mu'_i}$ and $\noextremity{\mu}{\mu'_i}$ are $\true$ for all $i \in \{1,\dots,\ell'\}$}{
            \If{$N^-_G(s), N^+_G(t) \subseteq V(\bigcup_{e \in E^-_s} \expansion(e))$ and $st \notin E(G)$}{
                \If{$\ell' = 1$}{
                    \tcp{See \Cref{prop:superbubbles-in-S-nodes}}
                    Report $(t,s)$ if the pertaining node of the edge in $E^-_s$ is not an S-node\; \label{line:superbubbles-P-ts-1}
                }
                \Else{
                    Report $(t,s)$\; \label{line:superbubbles-P-ts}
                }
            }
        }
    }
}

\For{every R-node $\mu$ of $T$}{
    \For{every neighbour $\nu$ of $T$ that is not a P-node}{
        Let $\{s,t\}=e_\mu \in \skel(\nu)$ be the virtual edge pertaining to $\mu$\;
        Let $X = \expansion(e_\mu)$\;
        \If{$\acyclic{\nu}{\mu}$ and $\noextremity{\nu}{\mu}$ are $\true$}{
            \If{$N^+_G(s), N^-_G(t) \subseteq V(X)$ and $ts \notin E(G)$}{
                \tcp{And hence $N^-_G(s), N^+_G(t) \subseteq \overline{V(X)}$}
                Report $(s,t)$\; \label{line:superbubbles-R-st}
            }
            \If{$N^+_G(t), N^-_G(s) \subseteq V(X)$ and $st \notin E(G)$}{
                \tcp{And hence $N^-_G(t), N^+_G(s) \subseteq \overline{V(X)}$}
                Report $(t,s)$\; \label{line:superbubbles-R-ts}
            }
        }
    }
}
\end{algorithm}

In Phase 3 the pairs $(s,t)$, $(t,s)$ such that $\{s,t\}$ is a separation pair are reported. \Cref{prop:spqr-tree-contains-split-pairs} tells us that every separation pair of $H$ is encoded in the edges of $T$ or it consists of a pair of nonadjacent vertices of an S-node. As we will show afterwards, pairs of nonadjacent vertices in an S-node do not form superbubbles due to minimality (this is essentially an application of \Cref{prop:cutvertex-inside-bubble}). Further, for similar reasons, if a pair of vertices are adjacent in the skeleton of an S-node and form a superbubble then the corresponding graph is not within the S-node. Hence, if $(s,t)$ is a superbubble such that $\{s,t\}$ is a separation pair of $H$, then there is a P-node of $T$ with vertex-set $\{s,t\}$, or there is an R-node of $T$ with a virtual edge $\{s,t\}$. We discuss informally the two cases, starting with the P-node.

Observe that SPQR trees encode not only every separation pair of the graph but also the respective sets of split components. The way in which these split components are put together to form the skeletons of the nodes is what defines the different types of nodes S,P, and R. In our application, i.e., finding superbubbles, examining only the natural separations encoded in the SPQR tree is not enough to ensure completeness. Consider for instance a P-node $\mu$ with $k \geq 4$ split components. The separations encoded in each of the $k$ edges of the SPQR tree incident to this P-node implicitly group $k-1$ expansions of the edges of $\mu$ and puts the vertices therein in one side of the separation, and the vertices on the expansion of remaining virtual edge is put on the remaining side. However, it is not hard to see that a graph of a superbubble could match, e.g., the union of the expansions of two virtual edges of $\mu$.
Thus we iterate over all the P-nodes of $T$ (with vertex set $\{s,t\}$) and proceed as follows. We group the virtual edges containing out-neighbours of $s$ and group the virtual edges containing in-neighbours of $t$ (also, the virtual edges containing in-neighbours (resp. out-neighbors) of $s$ (resp. $t$) are grouped); these have to match, otherwise $(s,t)$ is not a superbubble. Further, for each virtual edge in these (matching) sets, we check if the respective state leaving this P-node has the acyclicity and absence-of-extremities fields set to true. Finally, if all the out-neighbors (resp. in-neighbors) of $s$ (resp. $t$) are contained in candidate superbubble graph (which is given by the matching sets), and $ts \notin E(G)$, then $(s,t)$ is a superbubbloid by \Cref{lem:separation-pair-is-superbubbloid}. Minimality follows from the structure of P-nodes, and moreover the conditions in the algorithm are simple: if the resulting set has more than one edge then minimality clearly follows from \Cref{prop:superbubble-independent-paths}, otherwise we check if the pertaining node of the unique edge is not an S-node, since otherwise minimality is violated (by essential application of \Cref{prop:cutvertex-inside-bubble}).

The case of the R-node is simpler. For every R-node $\mu$ of $T$ and each edge adjacent to $\mu$ such that the other endpoint is not a P-node (since those were already examined), it suffices to check the conditions of \Cref{lem:separation-pair-is-superbubbloid} to conclude whether or not the current pair of vertices is a superbubble, as minimality is implied by the 3-connectivity of the R-nodes.
This phase is described in \Cref{alg:phase3}.

\paragraph{Correctness and runtime.}

\begin{algorithm}[t]
\caption{Superbubble finding algorithm}
\label{alg:superbubbles-main}
\KwIn{Directed graph $G$ without parallel edges}

Let $\mathcal{B}$ and $C \subseteq V(G)$ be the list of blocks and cutvertices of $U(G)$, respectively $(k \ge 1)$\;

\For{$H \in \mathcal{B}$}{
    \If{$H$ is a multi-bridge}{
        Let $s, t$ denote the vertices of $H$\;
        \If{$|E(H)| = 1$ and $N^+_s(G) = \{t\}$ and $N^-_t(G) = \{s\}$}{
            \tcp{Trivial superbubble}
            Report $(s,t)$\; \label{line:superbubbles-trivial}
        }
    }
    \Else{
        \tcp{$U(H)$ is 2-connected}
        \If{$H$ has exactly one source $s$ and one sink $t$ w.r.t.\ $H$ and $ts \notin E(G)$ \label{line:superbubbles-rn1}}{
            \If{$C \cap V(K)\setminus\{s,t\} = \emptyset$ and $N^+_G(s), N^-_G(t) \subseteq V(H)$ and $G[V(H)]$ is acyclic \label{line:superbubbles-rn2}}{
                \tcp{$V(\Bst)$ coincides with $V(H)$}
                Report $(s,t)$\; \label{line:superbubbles-block}
            }
        }
        $T \gets \mathsf{BuildSPQR}(H)$\;
        $\mathsf{Phase1}(G,T)$\;
        $\mathsf{Phase2}(G,T)$\;
        $\mathsf{Phase3}(G,T)$\;
    }
}
\end{algorithm}

We now give a series of results relating superbubbles and the different types of nodes of the SPQR tree. We begin with a simple result on S-nodes.

\begin{prop}\label{prop:superbubbles-in-S-nodes}
    Let $G$ be a directed graph, let $(s,t)$ be a superbubble of $G$ with graph $\Bst$, let $T$ be the SPQR tree of a maximal 2-connected subgraph of $U(G)$, let $\{\nu,\mu\}$ be an edge of $T$, and let $e_\mu \in \skel(\nu)$ be the virtual edge pertaining to node $\mu$. If $\{s,t\} = e_\mu$  and $\mu$ is an S-node then $\Bst \not\subseteq \expansion(e_\mu)$.
\end{prop}
\begin{proof}
    Suppose for a contradiction that $B_{st} \subseteq \expansion(e_\mu)$. Then $U(B_{st}) \subseteq U(\expansion(e_\mu))$.
    By definition of S-node, the graph $U(\expansion(e_\mu))$ is a split component of the split pair $\{s,t\}$ and contains a vertex $y$ separating $s$ and $t$ in $U(\expansion(e_\mu))$ (recall that S-nodes have at least three vertices). Since $U(\Bst) \subseteq U(\expansion(e_\mu))$ vertex $y$ is an $s$-$t$ cutvertex with respect to $U(\Bst)$. 
    The result now follows from \Cref{prop:cutvertex-inside-bubble}.
\end{proof}

Next we give a result for P-nodes. Essentially, we impose the usual conditions (e.g., acyclicity and absence of extremities) and group the virtual edges according to what expansions contain out-neighbors of one vertex and in-neighbors of the other vertex of the skeleton of the node. 

\begin{prop}[Superbubbles and P-nodes]\label{prop:P-node}
    Let $G$ be a directed graph and let $H$ be a maximal 2-connected subgraph of $U(G)$.
    Let $T$ be the SPQR tree of $H$ and let $\mu \in V(T)$ be a P-node. Let $e_1,\dots,e_k$ denote the edges of $\skel(\mu))$ with endpoints $\{s,t\}$ $(k\geq3)$. 
    Let $E^+_s = \{ e_i : V(\expansion(e_i))\cap N^+(s) \neq \emptyset \}$, $E^-_t = \{ e_i : V(\expansion(e_i))\cap N^-(t) \neq \emptyset \}$, and $K=\bigcup_{e\in E^+_s}\expansion(e)$.
    Then $(s,t)$ identifies a superbubbloid of $G$ with graph $K$ if and only if $E^+_s \neq \emptyset$, $E^+_s = E^-_t$, $N^+_G(s) \subseteq V(K)$, $N^-_G(t) \subseteq V(K)$, $ts \notin E(G)$, and for each $e \in E^+_s$ the graph $\expansion(e)$ is acyclic and does not contain extremities of $G$ except $\{s,t\}$.
\end{prop}
\begin{proof}
    $(\Rightarrow)$ Let $(s,t)$ be a superbubbloid of $G$ with graph $\Bst$ and let $\mu$ be a P-node whose skeleton has vertex set $\{s,t\}$. Since superbubbles are contained in the blocks of $G$ by \Cref{lem:bubbles-cutvertices} and $s,t \in V(H)$ it follows that $V(\Bst) \subseteq V(H)$. Further, since $\Bst$ contains all the out-neighbors of $s$ and $V(\Bst) \subseteq V(H)$, it follows that $E^+_s \neq \emptyset$ (analogously, $E^-_t \neq \emptyset$).
    We show that $K = \Bst$.
    
    We show that $K \subseteq \Bst$. Since $K$ and $\Bst$ are induced subgraphs it is enough to show that any vertex in $K$ is also in $\Bst$. 
    Let $u \in V(K)$. Then $u\in \expansion(e)$ for some $e\in E^+_s$. As established in the proof of (1) of \Cref{lem:separation-pair-is-superbubbloid}, $\expansion(e)$ has a path from $s$ to $t$ through $u$ since it is acyclic and has no extremities except $\{s,t\}$. Since $(s,t)$ is a superbubbloid, $u \in \Bst$.
    Now we show that $\Bst \subseteq K$. Suppose for a contradiction that $\Bst \not\subseteq K$. Since $K$ and $\Bst$ are induced subgraphs there is a vertex $v \in V(\Bst)\setminus V(K)$. So in particular, $s$ reaches $v$ without $t$ via some path. 
    Due to the structure of P-nodes, this path is contained in $\expansion(e)$ for some $e\in E(\skel(\mu))$. Thus, the first vertex following  $s$ in this path is also in $\expansion(e)$ and hence $\expansion(e)$ has an out-neighbour of $s$. Therefore $e\in E^+_s$ and hence $v \in V(K)$, a contradiction. We conclude that $K=\Bst$.

    The conditions $N^+_G(s) \subseteq V(K)$ and $N^-_G(t) \subseteq V(K)$ follow trivially since $\Bst=K$, and clearly $ts \notin E(G)$ because $(s,t)$ is a superbubbloid. 
    Further, for each $e \in E^+_s$ the graph $\expansion(e)$ is acyclic and does not contain extremities of $G$ except $\{s,t\}$, since a cycle or extremity except $\{s,t\}$ in some expansion would be a cycle or extremity in $K$, the former contradicting the acyclicity and the latter contradicting the matching property of superbubbloids.
    The equality $E^+_s = E^-_t$ follows at once from the fact that $\Bst=K=\bigcup_{e\in E^+_s}\expansion(e)$ and $N^+_G(s) \subseteq V(K)$ and $N^-_G(t) \subseteq V(K)$.

    $(\Leftarrow)$ Notice that $K$ has no extremities of $G$ except $\{s,t\}$ because each $\expansion(e)$ for $e\in E^+_s$ has no extremities of $G$ except $\{s,t\}$. Moreover, since each $\expansion(e)$ for $e\in E^+_s$ is acyclic, a cycle in $K$ contains vertices from different split components of $E^+_s$ and thus it contains $s$ or $t$, but $N^+_G(s),N^-_G(t) \subseteq V(K)$ and thus $K$ is acyclic.
    The set $\{s,t\}$ is a separation pair, $ts \notin E(G)$, and $K$ is the union of a subset of split components of $\{s,t\}$ by construction which moreover is nonempty since $E^+_s\neq\emptyset$ by assumption.
    So we are in conditions of applying \Cref{lem:separation-pair-is-superbubbloid} and conclude that $(s,t)$ is a superbubbloid of $G$ with graph $K$.
\end{proof}

Finally we show that superbubbloids within R-nodes are in fact superbubbles. More specifically, we show that no vertex in the skeleton of an R-node violates minimality (due to their connectivity, R-nodes contain too many paths between $s$ and $t$).

\begin{prop}[Superbubbles and R-nodes]\label{prop:R-node}
    Let $G$ be a directed graph and let $H$ be a maximal 2-connected subgraph of $U(G)$.
    Let $T$ be the SPQR tree of $H$ and let $\nu \in V(T)$ be an R-node. Let $e_\mu=\{s,t\} \in E(\skel(\nu))$ be a virtual edge with pertaining node $\mu$. Let $e_\nu \in E(\skel(\mu))$ be the virtual edge pertaining to $\nu$.
    If $N^+_G(s),N^-_G(t) \subseteq V(\expansion(e_\nu))$, $\expansion(e_\nu)$ is acyclic and has no extremities except $\{s,t\}$, $ts \notin E(G)$, then $(s,t)$ is a superbubble with graph $\expansion(e_\nu)$.
\end{prop}
\begin{proof}
    Let $K=\expansion(e_\nu)$.
    Notice that $\{s,t\}$ is a separation pair of $H$ and that $U(K)$ is a split component with respect to $\{s,t\}$.
    So we are in conditions of applying \Cref{lem:separation-pair-is-superbubbloid}, which implies that $(s,t)$ is a superbubbloid with graph $K$. Next we argue on the minimality.
    
    Notice that $\skel(\nu)$ has three internally vertex-disjoint $s$-$t$ paths since it is 3-connected. Then $\skel(\nu)$ without the edge $\{s,t\}$ has two internally vertex-disjoint $s$-$t$ paths and hence so does $U(K)$ (recall that split components are connected, so a path through the edges of $\skel(\nu)$ can be mapped to a path in $U(\expansion(e_\nu))$). Therefore $(s,t)$ is a superbubble by \Cref{prop:superbubble-independent-paths}.
\end{proof}

\begin{theorem}
    Let $G$ be a directed graph. The algorithm computing superbubbles (\Cref{alg:superbubbles-main}) is correct, that is, it finds every superbubble of $G$ and only its superbubbles, and it can be implemented in time $O(|V(G)| + |E(G)|)$.
\end{theorem}
\begin{proof}
     
     Let $H_1,\dots,H_{\ell}$ be the blocks of $U(G)$ $(\ell \geq 1)$.
     
    \textbf{(Completeness.)}
    We argue that every superbubble $(s,t)$ of $G$ is reported by the algorithm.
    
    If $(s,t)$ is a trivial superbubble then $N^+_G(s)=\{t\}$, $N^-_G(t)=\{s\}$, and $ts\notin E(G)$ by definition. These are exactly the conditions tested in Line~\ref{line:superbubbles-trivial}, and so $(s,t)$ is reported by the algorithm.
    Otherwise, if $V(B_{st})=V(H_i)$ for some $\iinl$, then $(s,t)$ is reported by the algorithm in Line~\ref{line:superbubbles-block}: by the matching property $B_{st}$ has at most one source $s$ and at most one sink $t$ of $G$, no vertex in the interior of $\Bst$ is a cutvertex of $U(G)$ by \Cref{lem:bubbles-cutvertices}, $\Bst$ is acyclic, $ts\notin E(G)$, and the usual neighborhood constraints on $s$ and $t$ hold. These conditions altogether are enough to report the pair $(s,t)$.
    
    Otherwise $\{s,t\}$ is a separation pair of a maximal 2-connected subgraph $H$ of $U(G)$ by \Cref{thm:bubbles-split-pairs}, and moreover $V(\Bst) \subset V(H)$. Let $T$ denote the SPQR tree of $H$. First, notice that no pair of nonadjacent vertices $u,v$ of an S-node identifies a superbubble (unless $V(B_{uv})=V(H)$, which is not the case by assumption). To see why, notice that such a superbubble $(u,v)$ with graph $B_{uv}$ contains a $u$-$v$ cutvertex with respect to $U(B_{uv})$, a contradiction to \Cref{prop:cutvertex-inside-bubble}. Thus it follows by \Cref{prop:spqr-tree-contains-split-pairs} that $\{s,t\}$ are endpoints of a virtual edge of a node $\mu$ of $T$. This virtual edge is associated with a tree edge $\{\nu,\mu\}$. Let $e_\mu$ be the virtual edge in $\nu$ pertaining to $\mu$ and let $e_\nu$ the virtual edge in $\mu$ pertaining to $\nu$.
    
    We first make some observations about the types of the nodes of $T$ and the location of $\Bst$ in $T$.
    If $\mu$ is an S-node then \Cref{prop:superbubbles-in-S-nodes} implies that $B_{st} \not\subseteq \expansion(e_\mu)$. 
    (Essentially $\state{\nu}{\mu}$ can be ignored). Symmetrically, $\state{\mu}{\nu}$ can be ignored whenever $\nu$ is an S-node.
    If $\mu$ is a P-node with vertex set $\{s,t\}$ then $\Bst$ can be expressed as the union of the expansions of the virtual edges of $\mu$ as described in \Cref{prop:P-node}. Therefore $\mu$ is examined separately by analyzing the states corresponding to $\expansion(e_i)$ for each $\iink$. Symmetrically, the same is done whenever $\nu$ is a P-node.
    Hence, the remaining virtual edges that encode superbubbles are those contained in the R-nodes. Moreover, if the pertaining node of that virtual edge is a P-node then $(s,t)$ is processed when analyzing P-nodes, so it suffices to analyze P-nodes individually and those tree-edges $\{\nu,\mu\}$ whenever $\nu$ is not a P-node and $\mu$ is an R-node. We analyze the two cases separately.
    
    \begin{itemize}
        \item \textbf{$\mu$ is a P-node:}
        Let $e_1,\dots,e_k$ be the edges in $\skel(\mu)$ whose endpoints are $\{s,t\}$ $(k \geq 3)$.
        Since $(s,t)$ is a superbubble, $(s,t)$ is also a superbubbloid and thus \Cref{prop:P-node} implies that
        $\Bst$ can be expressed, without loss of generality, as $\bigcup_{i=1}^{k'} \expansion(e_i)$ for some $k' < k$ ($k\neq k'$ since otherwise $V(\Bst)=V(H)$); further, it implies that $\expansion(e_i)$ is acyclic and has no extremities except $\{s,t\}$ for each $i=1,\dots,k'$, $E^+_s=E^-_t$, $N^+_G(s), N^-_G(t) \subseteq V(\Bst)$ and $ts \notin E(G)$.
        If $k' \neq 1$ then these conditions are enough to report $(s,t)$ as a superbubble (Line~\ref{line:superbubbles-P-st}). 
        Otherwise we have $k'=1$. If $e_1$ is a real edge then it was reported when analyzing the trivial superbubbles (notice that the conditions given by \Cref{prop:P-node} match those of a trivial superbubble). Otherwise $e_1$ is virtual and thus it has a pertaining node in $T$. Suppose for a contradiction that the pertaining node of $e_1$ is an S-node. Since $(s,t)$ is a superbubble and the out-neighbors of $s$ are contained in $\expansion(e_1)$, it implies that $\Bst \subseteq \expansion(e)$, from where \Cref{prop:superbubbles-in-S-nodes} gives a contradiction. Therefore the pertaining node of $e_1$ is not an S-node and $(s,t)$ is reported in Line~\ref{line:superbubbles-P-st-1}. 
        
        \item \textbf{$\mu$ is an R-node and $\nu$ is not a P-node:}
        In this case we have that $s$ and $t$ are the endpoints of $e_\mu$ (and $e_\nu)$.
        Suppose that $s$ has out-neighbors in both expansions. Let $u$ be such an out-neighbor. We claim that $V(H)=V(\Bst)$, a contradiction to the fact that we are under the assumption $V(H)\neq V(\Bst)$.
        First we show that $\expansion(e_\nu) \subseteq \Bst$.
        Suppose for a contradiction that $\expansion(e_\nu) \not\subseteq \Bst$. Then $\expansion(e_\nu)$ has a vertex or an edge not contained in $\Bst$. If it has an edge whose endpoints are contained in $\Bst$ then this edge is also contained in $\Bst$ because superbubbles are induced subgraphs with respect to vertices. Thus there is a vertex $x \in V(\expansion(e_\nu)) \setminus V(\Bst)$.
        We claim that there is an $x$-$u$ path $p$ in $U(\expansion(e_\nu))$ avoiding $s$ and $t$. We proceed by cases on the type of $\nu$.
            
            \begin{itemize}
                \item Suppose that $\nu$ is an S-node and let $e_1=\{s,t\},\dots,e_k$ $(k\geq 3)$ denote the edges of $\skel(\nu)$. Let $e_x$ be the edge in $\skel(\nu)$ whose expansion contains $x$ and define $e_u$ analogously.
                Notice that one endpoint $u'$ of $e_u$ is not $t$ (the other necessarily is $s$). Similarly, one endpoint of $e_x$ is distinct from $s$ and $t$ since $x \in \expansion(e_\nu)$; further, $x'$ can be chosen so that it is closest to $u'$ in the graph obtained by removing the edge $\{s,t\}$ from $\skel(\nu)$.
                Since $x\notin V(\Bst)$, $x$ is distinct from $s$ and $t$, so we are in conditions of applying \Cref{prop:reaches-in-expansion} and get an $x$-$x'$ path $p_x$ avoiding $s$ and $t$.\footnote{Notice that we can apply \Cref{prop:reaches-in-expansion} to directed graphs, where the out- and in-neighborhoods correspond to $+$ and $-$ incidences in the natural way.}
                Since $u$ is an out-neighbor of $s$ and $\nu$ is an S-node we have that $u$ is distinct from $t$ (and from $s$ trivially), so we are in conditions of applying \Cref{prop:reaches-in-expansion} and get a $u$-$u'$ path $p_u$ avoiding $s$ and $t$. Since $e_x,e_u \neq e_1$, $\skel(\nu)$ has an $x'$-$u'$ path avoiding $s$ and $t$ and thus $U(\expansion(e_\nu))$ has an $x$-$u$ path avoiding $s$ and $t$.
                
                \item Suppose that $\nu$ is an R-node. We can proceed identically as above and get vertices $x',u' \in \skel(\nu)$ with the desired properties by essential application of \Cref{prop:reaches-in-expansion}. Since R-nodes are 3-connected, $\skel(\nu)$ has an $x'$-$u'$ path avoiding both $s$ and $t$, and therefore $U(\expansion(e_\nu))$ has an $x$-$u$ path avoiding $s$ and $t$.
            \end{itemize}

        The path $p$ starts in a vertex not in $\Bst$ and ends in a vertex contained in $\Bst$.
        Let $b$ denote the first vertex in $p$ that is contained $\Bst$ (such a vertex exists since $b=u$ at the latest). Let $a$ be the vertex preceding $b$ in $p$ (such a vertex exists since $b \neq u$). Then $a \notin V(\Bst)$ because $b$ is the first vertex of $p$ contained in $\Bst$. Thus $U(\expansion(e_\nu))$ has an edge $\{a,b\}$, and hence $\expansion(e_\nu)$ has an edge $ab$ or $ba$. Since $b \in V(\Bst)$, $\expansion(e_\nu)$ has an $s$-$a$ path avoiding $t$ if $ba \in E(\expansion(e_\nu))$, and it has an $a$-$t$ path avoiding $s$ if $ab \in E(\expansion(e_\nu))$, thus $a \in V(\Bst)$, a contradiction. Therefore $\expansion(e_\nu) \subseteq \Bst$.

        Applying the argument above to $\expansion(e_\mu)$ (which is valid since $\mu$ is an R-node), we get $\expansion(e_\mu) \subseteq \Bst$. Since $\expansion(e_\mu) \cup \expansion(e_\nu) = H$ we have $H \subseteq \Bst$, and since superbubbles live within blocks, we get $H=\Bst$. Thus $V(\Bst)=V(H)$, as desired.

        So $s$ has out-neighbors in only one expansion between $\nu$ and $\mu$ and therefore the superbubble is contained in that expansion. By symmetry, $t$ has in-neighbors in only one expansion, and it is not hard to see that these expansions have to match.
        If the out-neighbors of $s$ are contained in $\expansion(e_\nu)$ and $\nu$ is an S-node then \Cref{prop:superbubbles-in-S-nodes} gives a contradiction. Thus the out-neighbors (resp. in-neighbors) of $s$ (resp. $t$) are fully contained in an expansion of an R-node. Further, we have that $\Bst$ is acyclic, has no extremities of $G$ expect $\{s,t\}$, and $ts \notin E(G)$, since $(s,t)$ is a superbubble. These conditions altogether are enough to report $(s,t)$ in Line~\ref{line:superbubbles-R-st} (when iterating over node $\mu$ if $\Bst \subseteq \expansion(e_\mu)$ and over node $\nu$ if $\Bst \subseteq \expansion(e_\nu)$). 
    \end{itemize}

    \textbf{(Soundness.)}
    Let $(s,t)$ be a pair of vertices reported by the algorithm. We show that $(s,t)$ is a superbubble of $G$.

    If the pair $(s,t)$ is reported by virtue of Line~\ref{line:superbubbles-trivial} then $(s,t)$ is a trivial superbubble by definition.
    
    If the pair $(s,t)$ is reported by virtue of Line~\ref{line:superbubbles-block} then $H$ has exactly one source $s$ and exactly one sink $t$ (with respect to $H$), no vertex in $H$ except $\{s,t\}$ is a cutvertex of $U(G)$, $H$ is acyclic, $N^+_G(s),N^-_G(t) \subseteq V(H)$, and $ts\notin E(G)$. If $\{s,t\}$ is also a separation pair then we are in conditions of applying \Cref{lem:separation-pair-is-superbubbloid} and conclude that $(s,t)$ is a superbubbloid; otherwise, we can also conclude that $(s,t)$ is a superbubbloid: a proof to that of \Cref{lem:separation-pair-is-superbubbloid} is possible - and very similar - if, instead of assuming that $\{s,t\}$ is a separation pair, we assume that no vertex except $\{s,t\}$ in a directed graph $K$ (with $U(K)$ connected) is an endpoint of an edge whose other endpoint is outside $K$, which are exactly the conditions we have in this case since no vertex except $\{s,t\}$ is a cutvertex of $G$ and blocks are connected; we omit the proof for the sake of brevity. For minimality, notice that $U(H)$ has two internally vertex-disjoint paths $s$ and $t$ since $H$ is 2-connected, so \Cref{prop:superbubble-independent-paths} implies that $(s,t)$ is a superbubble.
    
    Now we discuss the case when $\{s,t\}$ is a separation pair of a block $H$ of $G$. By symmetry it suffices to show that the pairs reported in Lines~\ref{line:superbubbles-P-st-1},~\ref{line:superbubbles-P-st}, and~\ref{line:superbubbles-R-st} are superbubbles.
    If $(s,t)$ is reported in Line~\ref{line:superbubbles-P-st} then \Cref{lem:separation-pair-is-superbubbloid} implies that $(s,t)$ is a superbubbloid with graph $K$; moreover, since $K=\bigcup_{e\in E^+_s}\expansion(e)$ consists of the union of $\ell \geq 2$ split components of $\{s,t\}$, which are disjoint, $U(K)$ has two internally vertex-disjoint $s$-$t$ paths, and hence \Cref{prop:superbubble-independent-paths} gives that $(s,t)$ is in fact a superbubble.
    If $(s,t)$ is reported in Line~\ref{line:superbubbles-R-st} then the fact that $(s,t)$ is a superbubble follows at once by \Cref{prop:R-node}.
    If $(s,t)$ is reported in Line~\ref{line:superbubbles-P-st-1} then the pertaining node of the unique edge in $E^+_s$ is an R-node (as no two P-nodes are adjacent in $T$), so we are conditions of applying \Cref{prop:R-node} and conclude that $(s,t)$ is a superbubble.

    \textbf{(Running time.)} Block-cut trees can be built in linear time \cite{Hopcroft73blockcut} and the total size of the blocks is linear in $|V(G)|+|E(G)|$. The case when a block in a multi-bridge is trivial, so suppose that we are analyzing a block $H$ that is 2-connected. Let $|H|=|V(H)|+|E(H)|$. We show that the rest of the algorithm runs in time $O(|H|)$, thus proving the desired bound.
    
    The conditions on Lines~\ref{line:superbubbles-rn1} and~\ref{line:superbubbles-rn2} are trivial and require $O(|H|)$ time altogether.
    The SPQR tree $T$ can be built in $O(|H|)$ time~\cite{gutwenger2001linear}. Phases 1 and 2 take $O(|H|)$ time by \Cref{lem:phases-correct}.
    For Phase 3, recall first that $T$ has $O(|H|)$ P-nodes as well as tree-edges by \Cref{lem:spqr-total-size}. Further, notice that the work done in each P-node and in each tree-edge entering an R-node takes constant-time with exception of the neighborhood queries of $s$ and $t$.
    To handle this type of queries, we can proceed identically as in the proof of \Cref{thm:snarls-time} and emulate the neighborhood queries in constant time. Therefore, for each P-node $\mu$ the algorithm spends $O(|E(\skel(\mu))|)$ time to build the sets described in \Cref{prop:P-node}, and for tree-edges entering R-nodes the algorithm spends a constant amount of time. The latter thus requires $O(|V(H)|)$ time altogether because $T$ has $O(|V(H)|)$ R-nodes at most, and the former requires $O(|H|)$ time altogether since the total number of edges of the skeletons in the nodes of $T$ is $O(|E(H)|)$ and $T$ has $O(|V(H)|)$ P-nodes at most (see \Cref{lem:spqr-total-size}).
\end{proof}

\newpage
\section{Additional experimental results}
\label{sec:additional-experimental-details}

\textbf{Experimental setup.}
We measure wall clock time and peak memory using the Linux \texttt{/usr/bin/time} command on all tools. To measure wall clock time and peak memory per phase in BubbleFinder, we use \texttt{std::chrono} and read the process status file \texttt{/proc/pid/status}, respectively. We use \texttt{vg} version \texttt{1.67} and BubbleGun version \texttt{1.1.9}. We compile our tool with \texttt{gcc} version 9.4 with the \texttt{-O3} optimization flag. All our experiments ran on an isolated AMD Ryzen Threadripper PRO 3975WX (32-Cores) with 504GiB of RAM, running Ubuntu 20.04.6 (kernel GNU/Linux 5.15.0, 64bit).

\begin{table}[h!]
\centering
\begin{adjustbox}{max width=\textwidth}
\begin{tabular}{llcccccccc}
\toprule
Group & Dataset &
\multicolumn{1}{c}{$n$} &
\multicolumn{1}{c}{$m$} &
\multicolumn{2}{c}{\cellcolor{famSnarl}\textbf{BubbleFinder}} &
\cellcolor{famSnarl}\textbf{vg snarls} &
\multicolumn{2}{c}{\cellcolor{famBiSB}\textbf{BubbleFinder}} &
\cellcolor{famBiSB}\textbf{BubbleGun} \\
& &
\multicolumn{1}{c}{$(\times 10^6)$} &
\multicolumn{1}{c}{$(\times 10^6)$} &
\cellcolor{famSnarl}{I/O+ALGO} &
\cellcolor{famSnarl}{all} &
\cellcolor{famSnarl}{} &
\cellcolor{famBiSB}{I/O+ALGO} &
\cellcolor{famBiSB}{all} &
\cellcolor{famBiSB}{} \\
\midrule
PGGB & \textsf{50x E.\ coli}    & $1.6$  & $2.1$
    & \textbf{1.08}  & 6.4 & \textbf{1.9}
    & 2.13 & 13.1
    & \textbf{1.8} \\
& \textsf{Primate Chr.~6} & $34.3$ & $47.1$
    & \textbf{7.82} & 141.95 & \textbf{33.7}
    & 46.80  & 291.2
    & \textbf{40.9} \\
& \textsf{Tomato Chr.~2}  & $2.3$  & $3.2$
    & 8.97 & 9.56 & \textbf{2.87}
    & 18.69 & 19.92
    & \textbf{8.39} \\
& \textsf{Mouse Chr.~19}  & $6.2$  & $8.6$
    & 5.12  & 27.96   & \textbf{7.78}
    & 9.01 & 53.13
    & \textbf{7.27} \\
\midrule
\texttt{vg} & \textsf{Chromosome 1}   & $18.8$ & $25.7$
    & \textbf{12.61} & 66.9 & \textbf{15.8}
    & 25.33 & \textbf{134.0}
    & \cellcolor[HTML]{F2F2F2} TO \\
& \textsf{Chromosome 10}  & $11.6$ & $15.9$
    & \textbf{7.68} & 41.2 & \textbf{9.7}
    & 15.34 & \textbf{82.3}
    & \cellcolor[HTML]{F2F2F2} TO \\
& \textsf{Chromosome 22}  & $3.2$  & $4.4$
    & \textbf{2.18}  & 11.5 & \textbf{2.8}
    & 4.44  & \textbf{23.0}
    & \cellcolor[HTML]{F2F2F2} TO \\
\midrule
DBG & \textsf{10x M.\ xanthus} & $1.6$ & $2.1$
    & \textbf{1.10} & 5.76 & \textbf{1.4}
    & 2.18 & 11.4
    & \textbf{1.7} \\
\bottomrule
\end{tabular}
\end{adjustbox}
\caption{Maximum resident set size (Max RSS) measured for each dataset and tool, in GiB. Tools in blue compute snarls, while tools in red compute superbubbles.
Cells containing ``TO'' indicate that the experiment timed out after three hours. Columns $n$ and $m$ denote the number of nodes and edges, respectively, given in units of $10^{6}$.
Columns labeled ``I/O+ALGO'' indicate that the timing does not include the building of the BC tree and the SPQR trees.}
\label{tab:max_rss}
\end{table}

\begin{figure}[h!]
    \centering
    \includegraphics[width=\textwidth]{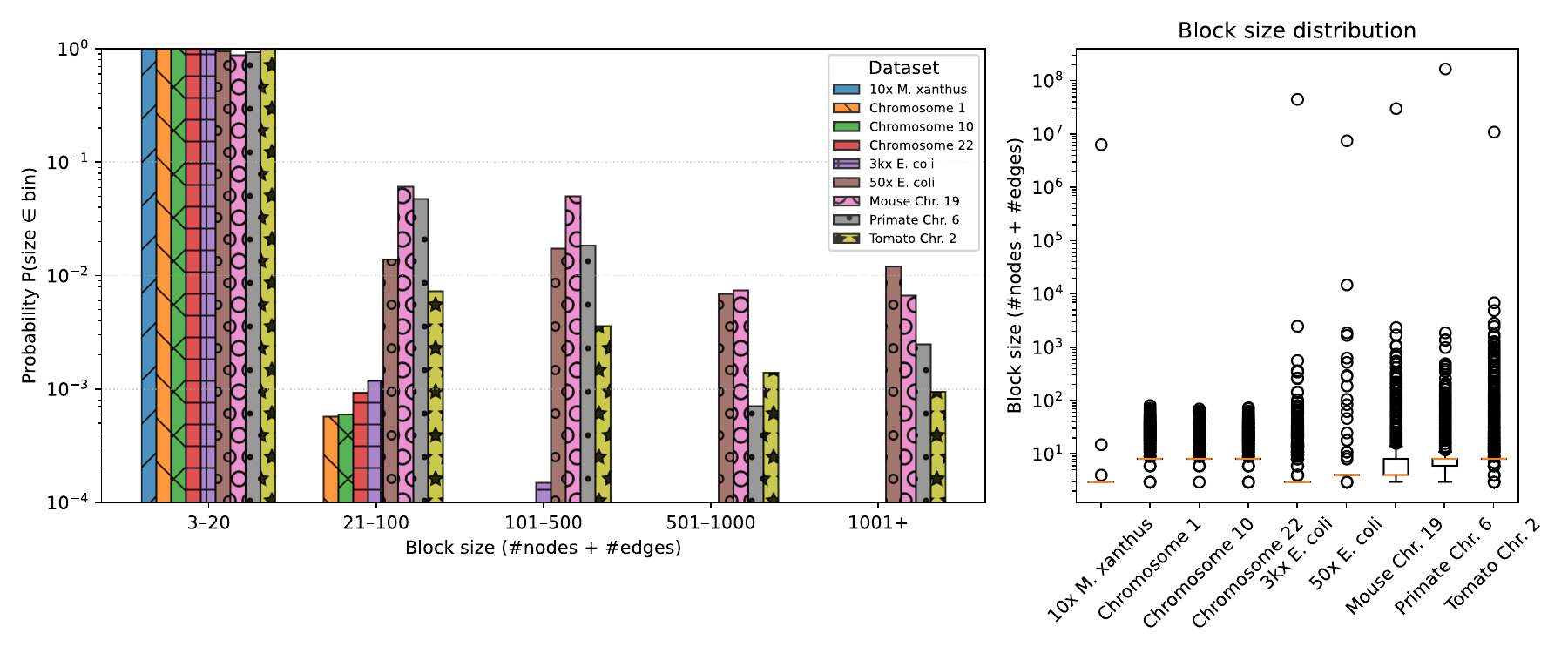}
    \caption{Distribution of BC-tree block sizes in some of our datasets.}
    \label{fig:block_dist}
\end{figure}

The datasets used in our benchmarks and their accession codes are listed in
\Cref{tab:dataset-accessions}.
Below we describe how we obtained or constructed the corresponding graphs.

\paragraph{PGGB graphs.}
The \textsf{50x~E.~coli} and \textsf{Primate~Chr.~6} graphs are the pangenome
graphs released with PGGB and were downloaded directly as GFA files from
Zenodo (see \Cref{tab:dataset-accessions}).
The \textsf{Tomato~Chr.~2} and \textsf{Mouse~Chr.~19} graphs were
reconstructed by running PGGB on the corresponding multi-genome FASTA files
from the same Zenodo record.
For each of these two datasets we downloaded the compressed FASTA
(\texttt{.fa.gz}), indexed it with
\texttt{samtools faidx \{fasta\}}, and then ran
\begin{verbatim}
pggb -t 4 -i {fasta} -o {output}
\end{verbatim}
using PGGB's default parameters (plus any dataset-specific options listed in
our Snakemake configuration).
Our pipeline then selected the \texttt{.smooth.final.gfa} produced by PGGB
(if present, otherwise a \texttt{.smooth.gfa} or generic \texttt{.gfa}) and
saved it as \texttt{name.pggb.gfa}.

\paragraph{Human variant graphs.}
To construct the variant graphs for human chromosomes~1, 10, and~22, we used
the 1000~Genomes Project Phase~3 VCFs on the GRCh37 reference.
For each chromosome \texttt{chr} we downloaded the reference FASTA and the
per-chromosome VCF listed in \Cref{tab:dataset-accessions}, and ran
\begin{verbatim}
vg construct -r {ref} -v {vcf} -R {chr} -m 3000000 -t 30 > {name}.vg
vg convert -f {name}.vg > {name}.gfa
\end{verbatim}
to obtain the corresponding GFA variation graph.
The resulting graphs for chromosomes~1, 10, and~22 are made available on
Zenodo (see \Cref{tab:dataset-accessions}) and are further processed as
described below.

\paragraph{De~Bruijn graph (\emph{Myxococcus xanthus}).}
For the \textsf{10x~\emph{M.~xanthus}} dataset, we selected the 10 assemblies
with NCBI accessions
GCA000012685.1, GCA000278585.2, GCA000340515.1,\\
GCA006400955.1, GCA006401215.1, GCA006401635.1,\\
GCA006402015.1, GCA006402415.1, GCA006402735.1, and\\
GCA900106535.1 (see \Cref{tab:dataset-accessions}).
We downloaded the corresponding genome FASTA files, concatenated them into a
single multi-FASTA, and built a compacted de~Bruijn graph of order~41 with
GGCAT~\cite{cracco2023ggcat} using
\begin{verbatim}
ggcat build --gfa-v1 -e -s 1 -k 41 -j 8 \
    mxanthus10.combined.fa -o mxanthus10.ggcat.gfa
\end{verbatim}
(where \texttt{-e -s 1} are the common options used in our workflow).

\paragraph{Graph preprocessing and bluntification.}
All graphs, regardless of origin (PGGB, \texttt{vg construct}, or GGCAT),
are processed uniformly by our Snakemake pipeline.
Each dataset first produces a ``raw'' GFA:
either from a reference+VCF via
\texttt{vg construct}/\texttt{vg convert},
from assemblies via GGCAT (\texttt{ggcat build}),
from pangenome FASTA input via PGGB (\texttt{pggb}),
or by downloading an existing GFA from Zenodo.

We then normalize every raw GFA into a \emph{blunt} GFA using
GetBlunted~\cite{eizenga2021getblunted}.
For each dataset \texttt{name} we run
\begin{verbatim}
get_blunted -i name.raw.gfa > name.bluntified.gfa
\end{verbatim}
and remove any header (\texttt{H}) lines to obtain \texttt{name.cleaned.gfa}.
Because \texttt{vg snarls} requires blunt graphs, we systematically use this
cleaned blunt GFA as input not only for \texttt{vg snarls}, but also for all
other tools (BubbleGun and BubbleFinder), ensuring that every tool operates
on exactly the same version of each graph.
If GetBlunted is not available or fails on a particular graph, the pipeline
falls back to a simple overlap-removal procedure that sets all edge overlaps
in the GFA to ``\texttt{*}''; this fallback bluntified GFA is then likewise
used for all tools.

\begin{table}[h]
    \centering
    \begin{tabular}{>{\raggedright}p{0.16\linewidth}|>{\raggedright}p{0.12\linewidth}|>{\raggedright}p{0.64\linewidth}}
    Dataset & Source & Accessions \tabularnewline\hline
        \textsf{50x E. coli} & Zenodo &
        Original PGGB GFA: \url{https://zenodo.org/records/7937947}
        \tabularnewline\hline

        \textsf{Primate Chr. 6} & Zenodo &
        Original PGGB GFA: \url{https://zenodo.org/records/7937947}
        \tabularnewline\hline

        \textsf{Tomato Chr. 2} & Zenodo &
        Input FASTA (PGGB): \url{https://zenodo.org/records/7937947}
        \tabularnewline\hline

        \textsf{Mouse Chr. 19} & Zenodo &
        Input FASTA (PGGB): \url{https://zenodo.org/records/7937947}
        \tabularnewline\hline

        \textsf{Chromosome 1} & 1000 Genomes Project &
        Reference: \url{https://ftp.1000genomes.ebi.ac.uk/vol1/ftp/technical/reference/human_g1k_v37.fasta.gz}\\
        VCF: \url{https://ftp.1000genomes.ebi.ac.uk/vol1/ftp/release/20130502/ALL.chr1.phase3_shapeit2_mvncall_integrated_v5b.20130502.genotypes.vcf.gz}
        \tabularnewline\hline

        \textsf{Chromosome 10} & 1000 Genomes Project &
        Reference: \url{https://ftp.1000genomes.ebi.ac.uk/vol1/ftp/technical/reference/human_g1k_v37.fasta.gz}\\
        VCF: \url{https://ftp.1000genomes.ebi.ac.uk/vol1/ftp/release/20130502/ALL.chr10.phase3_shapeit2_mvncall_integrated_v5b.20130502.genotypes.vcf.gz}
        \tabularnewline\hline

        \textsf{Chromosome 22} & 1000 Genomes Project &
        Reference: \url{https://ftp.1000genomes.ebi.ac.uk/vol1/ftp/technical/reference/human_g1k_v37.fasta.gz}\\
        VCF: \url{https://ftp.1000genomes.ebi.ac.uk/vol1/ftp/release/20130502/ALL.chr22.phase3_shapeit2_mvncall_integrated_v5b.20130502.genotypes.vcf.gz}
        \tabularnewline\hline

        \textsf{10x M. xanthus} & NCBI &
        GCA000012685.1, GCA000278585.2, GCA000340515.1, GCA006400955.1,\\
        GCA006401215.1, GCA006401635.1, GCA006402015.1, GCA006402415.1,\\
        GCA006402735.1, GCA900106535.1
        \tabularnewline\hline
    \end{tabular}
    \caption{Dataset accession codes.
    NCBI is the National Center for Biotechnology Information of the USA~\url{https://www.ncbi.nlm.nih.gov/}.
    All bluntified GFA graphs used in our benchmarks (one per dataset) are available in the Zenodo record
    \url{https://zenodo.org/records/17634355}.}
    \label{tab:dataset-accessions}
\end{table}

\end{document}